\documentclass[12pt, onecolumn]{IEEEtran} 
\usepackage{url}
\usepackage{makecell}
\usepackage{amsfonts}
\usepackage{amssymb}
\usepackage{amsmath}
\usepackage{graphicx, colordvi, psfrag}
\usepackage{calc,pstricks, pgf, xcolor}
\usepackage{epsfig, cite}
\usepackage{bm} 
\usepackage{bbm} 
\usepackage{enumerate}

\newcommand{\beq}[1]{\begin{equation}\label{#1}}
\newcommand{\eeq}{\end{equation}}

\newcommand{\beqn}[1]{\begin{eqnarray}\label{#1}}
\newcommand{\eeqn}{\end{eqnarray}}

\newtheorem{thmbody}{Theorem}
\newenvironment{thm}{
\begin{thmbody}
	}{
	\end{thmbody} 
	}
\newtheorem{dfnbody}{Definition}

\newtheorem{corbody}{Corollary}

\newtheorem{lemmabody}{Lemma}
\newenvironment{lemma}{
\begin{lemmabody}
	}{
	\end{lemmabody} 
	}
\newtheorem{propbody}{Proposition}

\newenvironment{proof}{
	{\it Proof:}
	}{
 $\Box$
	}

\usepackage{dsfont}
\begin{document}
\title{The Shannon Upper Bound for the Error Exponent}

\author{
\IEEEauthorblockN{Sergey Tridenski and Anelia Somekh-Baruch}\\
\IEEEauthorblockA{Faculty of Engineering\\Bar-Ilan University\\ Ramat-Gan, Israel\\
Email: tridens@biu.ac.il, somekha@biu.ac.il}
}

\maketitle
\begin{abstract}
For the discrete-time additive white generalized Gaussian noise channel with a generalized input power constraint,
with the 
respective shape and
power parameters $\geq 1$,
we derive an upper bound on the optimal block error exponent.
Explicit asymptotic upper bounds 
in the limit of a large block length $n$
are given for three special cases:
the Laplace noise channel and the Gaussian noise channel with the average absolute value constraint, 
and for the Laplace noise channel with the 
second 
power 
constraint. 
The derivation uses the method of types with finite alphabets of sizes 
depending on the block length $n$ and with the number of types sub-exponential in $n$.\footnote{
This work was supported by the Israel Science Foundation (ISF) grant \#1579/23.
}
\end{abstract}



%
%
%
%



\section{Introduction}\label{Int}


We study reliability of a generalized version of the additive white Gaussian noise (AWGN) channel, 
which is the discrete-time additive white generalized Gaussian noise (AWGGN) channel
with a generalized power constraint imposed on blocks of its inputs.
We assume the noise distribution symmetric and the generalized power parameters $\geq 1$ 
both for the noise distribution and the constraint.
As an example, consider the Laplace noise channel with a
conditional probability density function (PDF):
\begin{displaymath}
w(y \, | \, x) \; = \; \tfrac{\nu}{2}e^{-\nu|y \, - \, x|},
\;\;\;\;\;\; \nu > 0,
\end{displaymath}
and an average absolute value (i.e., power $1$) constraint imposed on blocks of $n$ inputs $(x_{1}, x_{2}, .\,.\,. \,, x_{n})$:
\begin{displaymath}
\frac{1}{n}\sum_{k \, = \, 1}^{n}|\,x_{k}\,| \; \leq \; s.
\end{displaymath}
The information capacity \cite[Eq.~9.8]{CoverThomas} 
of this channel is defined as the supremum
of the mutual information $I(\,{p\mathstrut}_{X}, \, w)$ over the channel input PDF ${p\mathstrut}_{X}$,
under the expected-absolute-value 
constraint $\mathbb{E}[|X|] \leq s$.
This capacity is upper-bounded by ${\log\mathstrut}_{\!2}(\nu s + 1)$ bits, as follows:
\begin{align}
\sup_{\substack{{p\mathstrut}_{X}:
\;
\mathbb{E}[|X|]\,\leq\,s
}}
I(\,{p\mathstrut}_{X}, \, w)
\;\; 
\underset{(a)}{\equiv} \;\; &
\sup_{\substack{{p\mathstrut}_{X}:
\;
\mathbb{E}[|X|]\,\leq\,s
}}
\;
\;
{\color{black}
\min_{\substack{\\ \eta \, > \, 0}}
}
\;
\Big\{
D\big(\, w \, \| \, \,{p\mathstrut}_{\eta} \, \, | \, \, {p\mathstrut}_{X}\big)
\, - \,
D\big(\,{p\mathstrut}_{Y} \, \| \, \,{p\mathstrut}_{\eta} \, \big)
\Big\}
\nonumber \\
\underset{(b)}{\leq} \;\; &
\sup_{\substack{{p\mathstrut}_{X}:
\;
\mathbb{E}[|X|]\,\leq\,s
}}
\;
\;
{\color{black}
\min_{\substack{\\ \eta \, > \, 0}}
}
\;
\,\,\,\,
D\big(\, w \, \| \, \,{p\mathstrut}_{\eta} \, \, | \, \, {p\mathstrut}_{X}\big)
\nonumber \\
\underset{(c)}{=} \;\; &
\sup_{\substack{{p\mathstrut}_{X}:
\;
\mathbb{E}[|X|]\,\leq\,s
}}
\;
\;
{\color{black}
\min_{\substack{\\ \eta \, > \, 0}}
}
\;
\,\,\,\,
\tfrac{1}{\ln 2}
\big[
\ln(\nu/\eta) \, + \, \eta 
\mathbb{E}_{\,{p\mathstrut}_{X}w}[|Y|] \, - \, 1
\big]
\nonumber \\
\underset{(d)}{=} \;\; &
\sup_{\substack{{p\mathstrut}_{X}:
\;
\mathbb{E}[|X|]\,\leq\,s
}}
\;
\;
{\color{white}
\min_{\substack{\\ \eta \, > \, 0}}
}
\;
\,\,\,\,
{\log\mathstrut}_{\!2}(\nu\mathbb{E}_{\,{p\mathstrut}_{X}w}[|Y|])
\label{eqShannonUpperB} \\
\underset{(e)}{=} \;\; &
\sup_{\substack{{p\mathstrut}_{X}:
\;
\mathbb{E}[|X|]\,\leq\,s
}}
\;
\;
{\color{white}
\min_{\substack{\\ \eta \, > \, 0}}
}
\;
\,\,\,\,
{\log\mathstrut}_{\!2}
\big(\nu\mathbb{E}\big[|X| + \tfrac{1}{\nu}e^{-\nu|X|}\big]
\big)
\;\; \underset{(f)}{=} \;\;
{\log\mathstrut}_{\!2}(\nu s + 1),
\label{eqExplicitSUB}
\end{align}
where in ($a$) we rewrite the mutual information equivalently using the Kullback–Leibler divergence between two PDFs,
one of which is ${p\mathstrut}_{\eta}(y) \triangleq \tfrac{\eta}{2}\exp\{-\eta|y|\}$
and another one is ${p\mathstrut}_{Y}(y) \equiv \int_{\mathbb{R}}{p\mathstrut}_{X}(x)w(y\,|\,x)dx$;
($b$) follows by the non-negativity of the divergence \cite[Eq.~8.57]{CoverThomas};
in ($c$) we assume that the mutual information and the divergences are defined in bits;
($d$) is the result of the minimization over $\eta > 0$;
in ($e$) we evaluate the conditional expectation given $X$ according to $w$;
and ($f$) follows as equality because the supremum is achieved asymptotically by the generalized density
${p\mathstrut}_{X}^{*}(x) \; = \; (1 - p)\delta(x) + p\delta(x - s/p)$, 
as $p \rightarrow 0$. The supremum of (\ref{eqShannonUpperB}) can be recognized as an example of the Shannon upper bound \cite[Theorem~18]{Shannon48}, 
\cite[Theorem~12.1.1]{CoverThomas}, \cite[Eq.~9b]{Dytso17}, \cite[Eq.~5]{Narayanan22}.
More specifically, (\ref{eqShannonUpperB})
is a special case of \cite[Eq.~9b]{Dytso17}. 
In this paper we use the technique leading to (\ref{eqShannonUpperB})
to derive an upper bound on the optimal exponent in the block error probability of
the AWGGN channel with a generalized power constraint.

In order to express the upper bound initially in terms of the mutual information and the divergence,
which are required for application of the above technique,
we extend the method of types \cite[Ch.~11.1]{CoverThomas}, \cite{Csiszar98},
as in \cite{TridenskiSomekhBaruh23}, to include
countable alphabets consisting of uniformly spaced real numbers, with the help of {\em generalized} power constraints on types.
The countable alphabets depend on the block length $n$ and the number of types
satisfying the generalized power constraints
is kept sub-exponential in $n$. 
As in \cite{TridenskiSomekhBaruh23}, in this paper the types correspond to empirical distributions
of
uniformly quantized real numbers 
in quantized versions of real channel input and output vectors
of length $n$. The quantized versions serve 
only for classification of actual real channel input and output vectors and not for 
the communication itself.
The uniform quantization step is 
different for the quantized versions of channel inputs and outputs,
and in both cases it
is chosen to be a decreasing function of $n$.

Our converse bound on the optimal error exponent of the AWGGN channel 
with a generalized power constraint for a given block length $n$,
presented in Theorem~\ref{thmErrorExp} of this paper,
contains an expression similar to that of the Shannon upper bound given by \cite[Eq.~9b]{Dytso17}, (\ref{eqShannonUpperB}), 
and recovers Shannon's sphere-packing converse bound on the error exponent of the AWGN channel with the 
second power constraint \cite[Eqs.~3,~4,~11]{Shannon59} 
in the limit of a large block length.
On the basis of Theorem~\ref{thmErrorExp}, we then derive three other asymptotic converse bounds, namely, for
the Laplace noise channel with the average absolute value constraint of the example above (Theorem~\ref{thmSUB11}),
for the AWGN channel with the average absolute value constraint (Theorem~\ref{thmSUB21}),
and for the Laplace noise channel with the 
second power constraint (Theorem~\ref{thmSUB12}).
As functions of the information rate $R$, all three explicit asymptotic bounds touch zero at 
$R$ equal to 
values of the Shannon upper bound on the information capacity,
such as (\ref{eqExplicitSUB}) in the case of Theorem~\ref{thmSUB11}.

In Sections~\ref{ComSys} and~\ref{Defs}, we describe the communication system and make preliminary
definitions. In Section~\ref{Main} we present the main results of the paper, which consist of the four
theorems with proofs, and illustrate them by graphs. 
The proof of Theorem~\ref{thmErrorExp} in Section~\ref{Main} relies on two lemmas developed in the later sections.
Section~\ref{MOT} provides an extension to the method of types. The
results of Section~\ref{MOT} are then applied in all the sections that follow. In Section~\ref{ConvLemma} we present
a converse lemma that is then used for derivation of an upper bound on the error
exponent in terms of types in Section~\ref{ErrExp}. Section~\ref{PDFtypePDF} connects between PDFs and types.

\section*{Notation}\label{Not}

Countable alphabets consisting of real numbers
are denoted by ${\cal X}_{n}$, ${\cal Y}_{n}$.
The set of types with denominator $n$ over 
${\cal X}_{n}$
is denoted by ${\cal P}_{n}({\cal X}_{n})$.
Capital `$P\,$' denotes probability mass functions, which are types:
${P\mathstrut}_{\!X}$, ${P\mathstrut}_{\!Y}$, ${P\mathstrut}_{\!XY}$, ${P\mathstrut}_{\!Y|X}$.
The type class and the support of a type ${P\mathstrut}_{\!X}$ are denoted by
$T({P\mathstrut}_{\!X})$ and ${\cal S}({P\mathstrut}_{\!X})$, respectively.
The expectation w.r.t. a probability distribution ${P\mathstrut}_{\!X}$ is denoted by $\mathbb{E}_{{P\mathstrut}_{\!X}}[\,\cdot\,]$.
Small `$p$' denotes probability density functions: ${p\mathstrut}_{X}$, ${p\mathstrut}_{Y}$, ${p\mathstrut}_{XY}$,
${p\mathstrut}_{Y|X}$.
Thin letters $x$, $y$ represent real values, while thick letters ${\bf x}$, ${\bf y}$ represent real vectors. 
Capital letters $X$, $Y$ 
represent random variables, 
boldface ${\bf Y}$ 
represents a random vector of length $n$.
The conditional type class of ${P\mathstrut}_{\!X|\,Y}$ given ${\bf y}$ is denoted by $T({P\mathstrut}_{\!X|\,Y}\,|\, {\bf y})$.
The 
quantized versions of variables are denoted by a superscript `$*$': $x^{*}_{k}$, ${\bf x}^{*}$, ${\bf Y}^{*}$.
Small $w$ stands for a conditional PDF, and
${W\mathstrut}_{\!n}$ stands for a discrete positive measure, which does not necessarily add up to $1$.
All information-theoretic quantities 
such as joint and conditional entropies
$H({P\mathstrut}_{\!XY})$, $H(Y\,|\,X)$, the mutual information
$I({P\mathstrut}_{\!XY})$,
$I\big({P\mathstrut}_{\!X}, {P\mathstrut}_{\!Y|X}\big)$,
$I\big({P\mathstrut}_{\!X}, \, {p\mathstrut}_{Y|X}\big)$, the Kullback-Leibler divergence
$D\big({P\mathstrut}_{\!Y|X}\,\|\, {W\mathstrut}_{\!n} \,|\,  {P\mathstrut}_{\!X}\big)$,
$D\big(\,{p\mathstrut}_{Y|X}\,\|\, \, w \, \,|\, {P\mathstrut}_{\!X}\big)$,
and the information rate $R$
are defined with respect to the
logarithm to a base $b>1$, denoted as
$\,{\log\mathstrut}_{\!b}(\cdot)$.
It is assumed that $0 \,{\log\mathstrut}_{\!b} (0) = 0$.
The natural logarithm is denoted as $\ln$.
The cardinality of a discrete set is denoted by $|\,\cdot\,|$, while the volume of a continuous region is denoted by $\text{vol}\,(\cdot)$.
The complementary set of a set $A$ is denoted by $A{\mathstrut}^{c}$.
Logical ``or'' and ``and'' are represented by the symbols $\lor$ and $\land$, respectively.
In Appendix B, $\widehat{p}{\mathstrut}_{XY}$ represents the rounded down version of the PDF ${p\mathstrut}_{XY}$.



\section{Communication system}\label{ComSys}


We consider communication over the time-discrete 
additive white {\em generalized} Gaussian noise
channel
with real channel inputs $x \in \mathbb{R}$ and channel outputs $y \in \mathbb{R}$
and a 
transition probability density
\begin{equation} \label{eqChannel}
w(y \, | \, x) \;\; \triangleq \;\; \tfrac{
q\nu^{1/q}}{2\Gamma(1/q)}\exp\big\{\!-\!\nu|y - x|^{q}\big\},
\;\;\;\;\;\; 0 < \nu \in \mathbb{R}, \;\;\; 1 \leq q \in \mathbb{R}.
\end{equation}

Communication is performed by blocks of $n$ channel inputs. Let $R > 0$ denote a nominal information rate.
Each block is used for transmission of one out of $M$ messages,
where $M = M(n, R) \triangleq \lfloor b^{\,n R}\rfloor$, for some logarithm base $b > 1$.
The encoder is a deterministic function $f\!\!: \{1, 2, .\,.\,.\, , \, M\} \rightarrow \mathbb{R}{\mathstrut}^{n}$,
which converts a message into a transmitted block,
such that
\begin{displaymath}
f(m) \; = \; {\bf x}(m)
\; = \; \big(
x_{1}(m), \,x_{2}(m), .\,.\,.\, , \, x_{n}(m)
\big),
\;\;\;\;\;\;\;\;\; m = 1, 2, .\,.\,.\, , \, M,
\end{displaymath}
where $x_{k}(m) \in \mathbb{R}$, for all $k = 1, 2, .\,.\,.\, , \, n$.
The set of all the 
codewords ${\bf x}(m)$, $m = 1, 2, .\,.\,.\, , \, M$, constitutes a codebook ${\cal C}$.
Each codeword ${\bf x}(m)$ in ${\cal C}$ satisfies the generalized power constraint with a real $r\geq 1$:
\begin{equation} \label{eqPowerConstraint}
\frac{1}{n}\sum_{k \, = \, 1}^{n}|\,x_{k}(m)\,|^{\,r} \; \leq \; s^{r},
\;\;\;\;\;\;\;\;\; m = 1, 2, .\,.\,.\, , \, M.
\end{equation}
The decoder is another deterministic function
$g\!\!: \mathbb{R}{\mathstrut}^{n} \rightarrow \{0, 1, 2, .\,.\,.\, , \, M\}$,
which converts the received block of $n$ channel outputs ${\bf y} \in \mathbb{R}{\mathstrut}^{n}$
into an estimated message, or, possibly, to 
a special error symbol `$0$':
\begin{equation} \label{eqDec}
g({\bf y})
\;\; = \;\;
\Bigg\{
\begin{array}{r l}
0, & \;\;\; {\bf y} \in \bigcap_{\,m \, = \, 1}^{\,M} {\cal D}{\mathstrut}_{m}^{c}, \\
m, & \;\;\; {\bf y} \in {\cal D}{\mathstrut}_{m}, \;\;\; 
m \in \{1, 2, .\,.\,.\, , \, M\},
\end{array}
\end{equation}
where each set ${\cal D}{\mathstrut}_{m} \subseteq \mathbb{R}{\mathstrut}^{n}$ is either an open region or the empty set, 
and the
regions are disjoint:
${\cal D}{\mathstrut}_{m} \cap \,{\cal D}{\mathstrut}_{m'} = \varnothing$
for $m \neq m'$. 
Observe that the maximum-likelihood decoder with open decision regions ${\cal D}{\mathstrut}_{m}^{*}\,$, defined for $m = 1, 2, .\,.\,.\, , \, M$
as
\begin{displaymath}
{\cal D}{\mathstrut}_{m}^{*}
\;\; = \;\;
\mathbb{R}{\mathstrut}^{n} \setminus
\bigcup_{
m':\;\;
(m' \, < \; m)
\; \lor \;
\big(\,m' \, > \, m \;\, \land \;\, {\bf x}(m') \, \neq \, {\bf x}(m)\,\big)
}
\Big\{
{\bf y}: \;
{\|
{\bf y} - {\bf x}(m')
\|\mathstrut}_{q}
\, \leq \,
{\|
{\bf y} - {\bf x}(m)
\|\mathstrut}_{q}
\Big\},
\end{displaymath}
is a special case of (\ref{eqDec}).
Note that the formal
definition of ${\cal D}{\mathstrut}_{m}^{*}$ includes the undesirable possibility of ${\bf x}(m') = {\bf x}(m)$ for $m' \neq m$.


\section{Definitions}\label{Defs}


For each $n$, we define two discrete countable alphabets ${\cal X}_{n}$ and ${\cal Y}_{n}$ 
as
one-dimensional lattices:
\begin{align}
&
\alpha, \beta, \gamma \in (0, 1), \;\;\;
\alpha + \beta + \gamma = 1,
\nonumber \\
&
\Delta_{\alpha,\,n} \; \triangleq \; 1/n^{\alpha},
\;\;\;
\Delta_{\beta,\,n} \; \triangleq \; 1/n^{\beta},
\;\;\;
\Delta_{\gamma,\,n} \; \triangleq \; 1/n^{\gamma},
\label{eqDelta} \\
& \Delta_{\alpha,\,n}\cdot \Delta_{\beta,\,n}\cdot \Delta_{\gamma,\,n} \;\; = \;\; 1/n,
\label{eqCube}
\end{align}
\begin{align}
&
{\cal X}_{n} \;\; \triangleq \;\;
\bigcup_{i \, \in \, \mathbb{Z}}\big\{i\Delta_{\alpha,\,n}\big\},
\;\;\;\;\;\;
{\cal Y}_{n} \;\; \triangleq \;\;
\bigcup_{i \, \in \, \mathbb{Z}}\big\{i\Delta_{\beta,\,n}\big\}.
\label{eqAlphabets}
\end{align}
For each $n$, we define also a discrete positive measure (not necessarily a distribution), which will approximate the channel $w$:
\begin{align}
{W\mathstrut}_{\!n}(y \, | \, x) \;\; & \triangleq \;\;
w(y \, | \, x)\cdot\Delta_{\beta,\,n},
\;\;\; \forall x \in {\cal X}_{n}, \; \forall y \in {\cal Y}_{n}.
\label{eqChanApprox}
\end{align}
The following set of Lipschitz-continuous functions will be used in the derivation of the error exponent:
\begin{align}
{\cal L} \; & \triangleq \;
\bigg\{
f\!\! : \mathbb{R} \rightarrow \mathbb{R}_{\,\geq\, 0}
\;\; \Big| \;\;
|f(y_{1}) - f(y_{2})| \leq  
K|y_{1} - y_{2}|,
\; \forall y_{1}, y_{2}\,;
\;
\int_{\mathbb{R}}f(y)dy = 1
\bigg\}, 
\label{eqLipschitz} \\
K \; & \triangleq \; 
\sup_{y\,<\,0}\frac{dw(y\,|\,0)}{dy}
\; = \; 
\Bigg\{
\begin{array}{l l}
\nu^{2}/2, & \;\;\; q = 1, \\
\tfrac{(q\nu)^{2/q}}{2\Gamma(1/q)}\big[\tfrac{(q\,-\,1)q}{e}\big]^{(q\,-\,1)/q}, & \;\;\; q > 1.
\end{array}
\nonumber
\end{align}
Note that ${\cal L}$ is a convex set and also 
each function $f \in {\cal L}$ is 
bounded and cannot exceed $\sqrt{K}$.
With parameters $\lambda$, $\kappa$, $\eta$, 
we define two auxiliary generalized Gaussian PDFs:
\begin{align}
{p\mathstrut}_{\lambda, \, \kappa}(y \, | \, x) 
\;\; & \triangleq \;\; 
\tfrac{q(\lambda\nu)^{1/q}}{2\Gamma(1/q)}\exp\big\{\!-\!\lambda\nu|y -\kappa x|^{q}\big\},
\;\;\;\;\;\; 0 < \lambda \leq 1, \;\;\; \kappa \in \mathbb{R},
\label{eqPlambdakappa} \\
{p\mathstrut}_{\eta}(y) 
\;\; & \triangleq \;\; 
\tfrac{r\eta^{1/r}}{2\Gamma(1/r)}\exp\big\{\!-\!\eta|y|^{r}\big\},
\;\;\;\;\;\;\;\;\;\;\;\;\;\;\;\;\;\, 0 < \eta \in \mathbb{R}.
\label{eqPeta}
\end{align}


\section{Main results}\label{Main}


In this section we present four theorems and their proofs.
Theorem~\ref{thmErrorExp} gives a converse bound on the optimal error exponent of the AWGGN channel 
with a generalized power constraint for a given block length $n$. 
The remaining three theorems give asymptotic converse bounds 
on the optimal error exponent in the limit of a large block length $n$  
for three special cases: the Laplace noise channel with the average absolute value constraint (Theorem~\ref{thmSUB11}),
the AWGN channel with the average absolute value constraint (Theorem~\ref{thmSUB21}),
and finally the Laplace noise channel with the  $2$nd power constraint (Theorem~\ref{thmSUB12}).
These results are followed by graphs depicting the three special cases at the end of the section.

The proof of the first theorem relies on Lemmas~\ref{LemAllCodebooks} and~\ref{LemPDFtoT},  
which appear in Sections~\ref{ErrExp} and~\ref{PDFtypePDF}, respectively.
The proofs of the subsequent Theorems~\ref{thmSUB11},~\ref{thmSUB21},~\ref{thmSUB12}, in turn, rely 
on Theorem~\ref{thmErrorExp}.

\bigskip

\begin{thm}[Shannon Upper Bound for the error exponent] \label{thmErrorExp}
{\em  Let the channel be given by (\ref{eqChannel}), let codebooks ${\cal C}$ satisfy (\ref{eqPowerConstraint}),
and let decoders $g$ be defined by (\ref{eqDec}).
Let $J \sim \text{Unif}\,\big(\{1, 2, .\,.\,.\, , \, M\}\big)$ be a random variable,
independent of the channel noise, and 
let ${\bf x}(J) \rightarrow {\bf Y}$ be the random channel-input and channel-output vectors, respectively. 
Let ${\cal X}_{n}$ be an alphabet defined as in (\ref{eqAlphabets}), (\ref{eqDelta})
with $\alpha \in \big(0, \tfrac{r}{1\,+\,2r}\big)$ 
and let $\alpha < \beta < \tfrac{r}{1\,+\,r}(1-\alpha)$.
Then for any $c_{XY}$ and $\epsilon > 0$ there exists $n_{0} = n_{0}(\alpha, \,\beta, \,s, \,c_{XY}, \,r, \,q, \,\nu, \, \epsilon)$
such that for any 
$n > n_{0}$}
\begin{align}
&
\sup_{{\cal C}}
\;
\sup_{g}
\;
\Big\{\!
- \tfrac{1}{n}\,{\log\mathstrut}_{\!b} \Pr \big\{
g({\bf Y}) \neq J \big\}
\Big\} \; \leq \;
\!\!
\max_{\substack{\\{P\mathstrut}_{\!X}:\\
{P\mathstrut}_{\!X}\,\in \, {\cal P}_{n}({\cal X}_{n}),
\\
\mathbb{E}_{{P\mathstrut}_{\!X}}[{|X|\mathstrut}^{r}] \; \leq \; s^{r} + \, \epsilon}}
{\color{black}
\sup_{\substack{\rho\,\geq\,0}}
}
\inf_{\substack{\\{\color{white}{P\mathstrut}_{\!X}}\lambda, \, \kappa:{\color{white}{P\mathstrut}_{\!X}}
\\
{\color{white}{P\mathstrut}_{\!X}}\lambda \, \in \, (0, \, 1], \; \kappa \, \in \, \mathbb{R}, {\color{white}{P\mathstrut}_{\!X}}
\\
\mathbb{E}_{{P\mathstrut}_{\!X}{p\mathstrut}_{\lambda, \, \kappa}}[{|Y-X|\mathstrut}^{q}] \; \leq \; c_{XY}
}}
\!\!\!\!
\Big\{\!
-\rho(R - \epsilon)
\, +
\nonumber \\
&
\tfrac{1}{\ln b}\Big[
\nu
\mathbb{E}_{{P\mathstrut}_{\!X}{p\mathstrut}_{\lambda, \, \kappa}}\big[{|Y-X|\mathstrut}^{q}\big]
\, + \,
\tfrac{\rho}{r}
\ln
\mathbb{E}_{{P\mathstrut}_{\!X}{p\mathstrut}_{\lambda, \, \kappa}}\big[{|Y|\mathstrut}^{r}\big]
\, + \,
\rho\ln
\tfrac{\Gamma(1 \, + \, 1/r)}{\Gamma(1 \, + \, 1/q)}
\, + \,
\tfrac{\rho}{q}\ln \nu
\, + \,
\tfrac{\rho\,+\,1}{q}\ln (\lambda/e)
\, + \,
\tfrac{\rho}{r}\ln(re)
\Big]\Big\}
\nonumber \\
&
\;\;\;\;\;\;\;\;\;\;\;\;\;\;\;\;\;\;\;\;\;\;\;\;\;\;\;\;\;\;\;\;\;\;\;\;\;\;\;\;\;\;\;\;\;\;\;\;\;\;\;\;\;\;\;\;\;\;\;\;
\;\;\;\;\;\;\;\;\;\;\;\;\;\;\;\;\;\;\;\;\;\;\;\;\;\;\;\;\;\;\;\;\;\;\;\;\;\;\;\;\;\;\;\;\;\;\;\;\;\;\;\;\;\;\;\;\;\;\;\;
\;\;\;\;\;\;\;\;\;\;\,\,
+ \, o(1),
\label{eqGenegalqr}
\end{align}
{\em where 
${p\mathstrut}_{\lambda, \, \kappa}$ are defined by (\ref{eqPlambdakappa}),
and $o(1)\rightarrow 0$, as $n\rightarrow \infty$, 
depending only on the parameters $\alpha$, $\beta$, $s$, $c_{XY}$, $\epsilon$, $r$, $q$, $\nu$.
The inequality holds
for all $R > 0$ with the possible exception of the single point on the $R$-axis
where the RHS of (\ref{eqGenegalqr}) as a function of $R$ may transition between $+\infty$ and a finite value.}
\end{thm}

\bigskip

\begin{proof}
Using Lemmas~\ref{LemAllCodebooks} and~\ref{LemPDFtoT} we can write 
with any $c_{XY}$ and $\epsilon > 0$
for $n$ sufficiently large:
\begin{align}
&
\sup_{{\cal C}}
\;
\sup_{g}
\;
\Big\{\!
- \tfrac{1}{n}\,{\log\mathstrut}_{\!b} \Pr \big\{
g({\bf Y}) \neq J \big\}
\Big\}
\nonumber \\
\overset{a}{\leq} \; &
\!
\max_{\substack{\\{P\mathstrut}_{\!X{\color{white}|}}\!\!:\\
{P\mathstrut}_{\!X}\,\in \, {\cal P}_{n}({\cal X}_{n}),
\\
\mathbb{E}_{{P\mathstrut}_{\!X}}[{|X|\mathstrut}^{r}] \; \leq \; s^{r} + \, \epsilon}}
\;\;
\min_{\substack{\\{P\mathstrut}_{\!Y|X}:\\
{P\mathstrut}_{\!XY}\,\in \, {\cal P}_{n}({\cal X}_{n}\,\times\, {\cal Y}_{n}),
\\
\mathbb{E}_{{P\mathstrut}_{\!XY}}[{|Y-X|\mathstrut}^{q}] \; \leq \; c_{XY} + \, \epsilon,
\\ I({P\mathstrut}_{\!X}, \, {P\mathstrut}_{\!Y|X}) \; \leq \; R \, - \, \epsilon
}}
\!\!\!\!\!\!\!
\Big\{
D\big({P\mathstrut}_{\!Y|X}\,\|\, {W\mathstrut}_{\!n} \,|\,  {P\mathstrut}_{\!X}\big)
\Big\} \; + \; o(1)
\label{eqCCExponent} \\
\overset{b}{\leq} \; &
\!
\max_{\substack{\\{P\mathstrut}_{\!X{\color{white}|}}\!\!:\\
{P\mathstrut}_{\!X}\,\in \, {\cal P}_{n}({\cal X}_{n}),
\\
\mathbb{E}_{{P\mathstrut}_{\!X}}[{|X|\mathstrut}^{r}] \; \leq \; s^{r} + \, \epsilon}}
\;\;
\inf_{\substack{{p\mathstrut}_{Y|X}:\\
{p\mathstrut}_{Y|X}(\, \cdot \, \,|\, x) \, \in \, {\cal L}, \; \forall x,
\\
\mathbb{E}_{{P\mathstrut}_{\!X}{p\mathstrut}_{Y|X}}
[{|Y-X|\mathstrut}^{q}] \; \leq \; c_{XY},
\\ I({P\mathstrut}_{\!X}, \,\, {p\mathstrut}_{Y|X}) \; \leq \; R \, - \, 2\epsilon
}}
\!\!\!\!\!\!\!\!
\Big\{
D\big(\,{p\mathstrut}_{Y|X}\,\|\, \, w \, \,|\, {P\mathstrut}_{\!X}\big)
\Big\} \; + \; o(1),
\label{eqInf11}
\end{align}
where ($a$) holds by Lemma~\ref{LemAllCodebooks};
then ($b$) holds by Lemma~\ref{LemPDFtoT} with $c_{X} = s^{r} \!+ \epsilon$,
for the alphabet parameters $\alpha \in \big(0, \tfrac{r}{1\,+\,2r}\big)$ and 
$\alpha < \beta < \tfrac{r}{1\,+\,r}(1-\alpha)$. 
Consider next the following upper bounds on the infimum of (\ref{eqInf11}) for a given ${P\mathstrut}_{\!X}$:
\begin{align}
&
{\color{white}
\sup_{\substack{\rho\,\geq\,0}}
}
\inf_{\substack{{p\mathstrut}_{Y|X}:\\
{p\mathstrut}_{Y|X}(\, \cdot \, \,|\, x) \, \in \, {\cal L}, \; \forall x,
\\
\mathbb{E}_{{P\mathstrut}_{\!X}{p\mathstrut}_{Y|X}}
[{|Y-X|\mathstrut}^{q}] \; \leq \; c_{XY},
\\ I({P\mathstrut}_{\!X}, \,\, {p\mathstrut}_{Y|X}) \; \leq \; R \, - \, 2\epsilon
}}
\!\!\!\!\!\!\!\!
\Big\{
D\big(\,{p\mathstrut}_{Y|X}\,\|\, \, w \, \,|\, {P\mathstrut}_{\!X}\big)
\Big\}
\label{eqInf1} \\
\overset{a}{=}
\;\; &
{\color{black}
\sup_{\substack{\rho\,\geq\,0}}
}
\inf_{\substack{{p\mathstrut}_{Y|X}:\\
{p\mathstrut}_{Y|X}(\, \cdot \, \,|\, x) \, \in \, {\cal L}, \; \forall x,
\\
\mathbb{E}_{{P\mathstrut}_{\!X}{p\mathstrut}_{Y|X}}
[{|Y-X|\mathstrut}^{q}] \; \leq \; c_{XY}{\color{white},}
}}
\!\!\!\!\!\!\!\!
\Big\{
D\big(\,{p\mathstrut}_{Y|X}\,\|\, \, w \, \,|\, {P\mathstrut}_{\!X}\big)
\, + \,
\rho
\Big[
I\big({P\mathstrut}_{\!X}, \, {p\mathstrut}_{Y|X}\big)
-
R
+
2\epsilon
\Big]
\Big\}
\label{eqInf2} \\
\overset{b}{\leq}  \;\; &
{\color{black}
\sup_{\substack{\rho\,\geq\,0}}
}
\inf_{\substack{\lambda, \, \kappa:\\
\lambda \, \in \, (0, \, 1], \; \kappa \, \in \, \mathbb{R},
\\
\,\mathbb{E}_{{P\mathstrut}_{\!X}{p\mathstrut}_{\lambda, \, \kappa}}[{|Y-X|\mathstrut}^{q}] \; \leq \; c_{XY}{\color{white},}\,
}}
\!\!\!\!\!\!\!\!
\Big\{
D\big(\,{p\mathstrut}_{\lambda, \, \kappa}\,\|\, \, w \, \,|\, {P\mathstrut}_{\!X}\big)
\, + \,
\rho
\Big[
I\big({P\mathstrut}_{\!X}, \, {p\mathstrut}_{\lambda, \, \kappa}\big)
 -
R
+
2\epsilon
\Big]
\Big\}
\nonumber \\
\overset{c}{\leq}  \;\; &
{\color{black}
\sup_{\substack{\rho\,\geq\,0}}
}
\inf_{\substack{\lambda, \, \kappa, \, \eta:\\
\lambda \, \in \, (0, \, 1], \; \kappa \, \in \, \mathbb{R},
\; \eta \, > \, 0, 
\\
\,\mathbb{E}_{{P\mathstrut}_{\!X}{p\mathstrut}_{\lambda, \, \kappa}}[{|Y-X|\mathstrut}^{q}] \; \leq \; c_{XY}{\color{white},}\,
}}
\!\!\!\!\!\!\!\!
\Big\{
D\big(\,{p\mathstrut}_{\lambda, \, \kappa}\,\|\, \, w \, \,|\, {P\mathstrut}_{\!X}\big)
\, + \, 
\rho
\Big[
D\big(\,{p\mathstrut}_{\lambda, \, \kappa}\,\|\, \, {p\mathstrut}_{\eta} \, \,|\, {P\mathstrut}_{\!X}\big)
-
R
+
2\epsilon
\Big]
\Big\}
\nonumber \\
\overset{d}{=}  \;\; &
{\color{black}
\sup_{\substack{\rho\,\geq\,0}}
}
\inf_{\substack{\lambda, \, \kappa, \, \eta:\\
\lambda \, \in \, (0, \, 1], \; \kappa \, \in \, \mathbb{R},
\; \eta \, > \, 0,
\\
\,\mathbb{E}_{{P\mathstrut}_{\!X}{p\mathstrut}_{\lambda, \, \kappa}}[{|Y-X|\mathstrut}^{q}] \; \leq \; c_{XY}{\color{white},}\,
}}
\!\!\!\!\!\!\!\!
\Big\{\tfrac{1}{\ln b}\Big[
\nu
\mathbb{E}_{{P\mathstrut}_{\!X}{p\mathstrut}_{\lambda, \, \kappa}}\big[{|Y-X|\mathstrut}^{q}\big]
\, + \,
\underbrace{
\rho\eta
\mathbb{E}_{{P\mathstrut}_{\!X}{p\mathstrut}_{\lambda, \, \kappa}}\big[{|Y|\mathstrut}^{r}\big]
\, - \,
\tfrac{\rho}{r}\ln\eta}_{\text{minimize over $\eta > 0$}}
\nonumber \\
&
\;\;\;\;\;\;\;\;\;\;\;\;\;\;\;\;\;\;\;\;\;\;\;\;\;\;\;\;\;\;\;\;\;\,\,\,
+ \,
\rho\ln
\tfrac{q\Gamma(1/r)}{r\Gamma(1/q)}
\, + \,
\tfrac{\rho}{q}\ln \nu
\, + \,
\tfrac{\rho\,+\,1}{q}\ln (\lambda/e)
\Big]
\, - \,
\rho(R - 2\epsilon)
\Big\}
\nonumber \\
\overset{e}{=}  \;\; &
{\color{black}
\sup_{\substack{\rho\,\geq\,0}}
}
\inf_{\substack{\lambda, \, \kappa:\\
\lambda \, \in \, (0, \, 1], \; \kappa \, \in \, \mathbb{R},
\\
\,\mathbb{E}_{{P\mathstrut}_{\!X}{p\mathstrut}_{\lambda, \, \kappa}}[{|Y-X|\mathstrut}^{q}] \; \leq \; c_{XY}{\color{white},}\,
}}
\!\!\!\!\!\!\!\!
\Big\{\tfrac{1}{\ln b}\Big[
\nu
\mathbb{E}_{{P\mathstrut}_{\!X}{p\mathstrut}_{\lambda, \, \kappa}}\big[{|Y-X|\mathstrut}^{q}\big]
\, + \,
\underbrace{
\tfrac{\rho}{r}
\ln\!\Big(re
\mathbb{E}_{{P\mathstrut}_{\!X}{p\mathstrut}_{\lambda, \, \kappa}}\big[{|Y|\mathstrut}^{r}\big]
\Big)}_{\text{$\sim$ Shannon Upper Bound}}
\nonumber \\
&
\;\;\;\;\;\;\;\;\;\;\;\;\;\;\;\;\;\;\;\;\;\;\;\;\;\;\;\;\;\;\;\;\;\,\,\,
+ \,
\rho\ln
\tfrac{q\Gamma(1/r)}{r\Gamma(1/q)}
\, + \,
\tfrac{\rho}{q}\ln \nu
\, + \,
\tfrac{\rho\,+\,1}{q}\ln (\lambda/e)
\Big]
\, - \,
\rho(R - 2\epsilon)
\Big\},
\label{eqShannonUB}
\end{align}
where 

($a$) holds because the infimum (\ref{eqInf1}) is a convex ($\cup$) function of $R$,
while the supremum (\ref{eqInf2}) as a function of $R$
is the closure of the lower convex envelope of the infimum (\ref{eqInf1}). 
Then (\ref{eqInf1}) and (\ref{eqInf2}) coincide
for all $R > 0$ with the possible exception of the {\em single point} on the $R$-axis
where the infimum (\ref{eqInf1}) may transition between $+\infty$ and a finite value.

($b$) holds because ${p\mathstrut}_{\lambda, \, \kappa} \in {\cal L}$ according to the definitions
 (\ref{eqLipschitz}) and (\ref{eqPlambdakappa}).

($c$) holds because 
$D\big(\,{p\mathstrut}_{\lambda, \, \kappa}\,\|\, \, {p\mathstrut}_{\eta} \, \,|\, {P\mathstrut}_{\!X}\big) \geq I\big({P\mathstrut}_{\!X}, \, {p\mathstrut}_{\lambda, \, \kappa}\big)$, with ${p\mathstrut}_{\eta}$ as defined in (\ref{eqPeta}).

($d$) holds as equality inside the infimum, separately for each $\lambda$, $\kappa$, $\eta$,
by the definitions (\ref{eqChannel}), (\ref{eqPlambdakappa}), (\ref{eqPeta}),
and the property of the gamma function $\Gamma(1 + t) = t\Gamma(t)$.

($e$) follows in the result of the minimization 
of a convex ($\cup$) function of $\eta > 0$.

We conclude that the inequality between (\ref{eqInf1}) and (\ref{eqShannonUB})
for a given ${P\mathstrut}_{\!X}$ 
holds for all $R > 0$ with the possible exception of the single point on the $R$-axis
where the supremum (\ref{eqShannonUB}) may transition between $+\infty$ and a finite value. 
This property 
is conserved also 
after both (\ref{eqInf1}) and (\ref{eqShannonUB})
are maximized over ${P\mathstrut}_{\!X}$ as in (\ref{eqInf11}). 
\end{proof}

Substitution of \cite[Eq.~13]{TridenskiSomekhBaruh23}
for ${p\mathstrut}_{\lambda, \, \kappa}$ for each $\rho \geq 0$ in (\ref{eqGenegalqr})
with $q = 2$, $r = 2$, and $c_{XY}\geq \frac{1}{2\nu} + s^{2} + \epsilon$
leads to Shannon's sphere-packing converse bound on the error exponent of the AWGN channel with the power constraint \cite[Eqs.~3,~4,~11]{Shannon59} in the limit of a large block length.
Next we turn to other special cases.

\bigskip

\begin{thm}[Case $q = 1$, $r = 1$] \label{thmSUB11}
{\em Let the channel be given by (\ref{eqChannel}) with $q = 1$, 
let codebooks ${\cal C}$ satisfy (\ref{eqPowerConstraint}) with $r = 1$,
and let decoders $g$ be defined by (\ref{eqDec}). 
Let $J \sim \text{Unif}\,\big(\{1, 2, .\,.\,.\, , \, M\}\big)$ be a random variable,
independent of the channel noise, and let ${\bf x}(J) \rightarrow {\bf Y}$ 
be the random channel-input and channel-output vectors, respectively. Then
}
\begin{align}
&
\limsup_{\substack{n\,\rightarrow\,\infty}}
\;
\sup_{{\cal C}}
\;
\sup_{g}
\;
\Big\{\!
- \tfrac{1}{n}\,{\log\mathstrut}_{\!b} \Pr \big\{
g({\bf Y}) \neq J \big\}
\Big\}
\nonumber \\
\leq
\;\;
&
\sup_{\substack{\rho\,\geq\,0}}
\;
\Big\{
\tfrac{1}{\ln b}
\min
\Big\{
\nu s,\;\;
\rho
\ln\!
\big(\lambda(\nu s, \rho)\nu s + 1\big)
\, + \,
\tfrac{1}{\lambda(\nu s, \, \rho)}
\, + \,
\ln\lambda(\nu s, \rho)
\, - \, 1
\Big\}
\, - \,
\rho R
\Big\},
\label{eqLaplace} \\
&
\;\;\;\;\;\;\;\;\;\;\;\;\;\;\;\;\;\;\;\;\;\;\;\;\;\;\;\;\;\;\;\;\;\;\;\;\;\;
\lambda(\nu s, \rho) \; = \;
\tfrac{1}{2(1\,+\,\rho)\nu s}
\big[
\nu s - 1
+
\sqrt{
(\nu s - 1)^{2}
+
4(1 + \rho)\nu s
}
\big].
\nonumber
\end{align}
\end{thm}

\bigskip

\begin{proof}
Let us rewrite the maximum of (\ref{eqGenegalqr}) for $q = 1$, $r = 1$, with $s + \epsilon = s'$
and $c_{XY} = s' + 1/\nu$:
\begin{align}
&
{\color{black}
\max_{\substack{\\{P\mathstrut}_{\!X}:\\
{P\mathstrut}_{\!X}\,\in \, {\cal P}_{n}({\cal X}_{n}),
\\
\mathbb{E}_{{P\mathstrut}_{\!X}}[|X|] \; \leq \; s'}}
}
\;\;
{\color{black}
\sup_{\substack{\rho\,\geq\,0}}
}
\;
{\color{black}
\inf_{\substack{\\{\color{white}{P\mathstrut}_{\!X}}\lambda, \, \kappa:{\color{white}{P\mathstrut}_{\!X}}
\\
{\color{white}{P\mathstrut}_{\!X}}\lambda \, \in \, (0, \, 1], \; \kappa \, \in \, \mathbb{R}, {\color{white}{P\mathstrut}_{\!X}}
\\
\mathbb{E}_{{P\mathstrut}_{\!X}{p\mathstrut}_{\lambda, \, \kappa}}[|Y-X|] \; \leq \; s' \, + \, 1/\nu
}}
}
\Big\{
\!-\!\rho(R - \epsilon)
\, +
\nonumber \\
&
\;\;\;\;\;\;\;\;\;\;\;\;\;\;\;\;\;\;\;\;\;\;\;\;\;\;\;\;\;\;\;\;\;\;\;\;\;\;\;\;\;\;\;\,\,
\tfrac{1}{\ln b}\Big[
\nu
\mathbb{E}_{{P\mathstrut}_{\!X}{p\mathstrut}_{\lambda, \, \kappa}}\big[|Y-X|\big]
\, + \,
\rho
\ln\!
\big(\lambda\nu
\mathbb{E}_{{P\mathstrut}_{\!X}{p\mathstrut}_{\lambda, \, \kappa}}\big[|Y|\big]
\big)
\, + \,
\ln\lambda
\, - \, 1
\Big]\Big\}
\nonumber \\
\overset{a}{=} \;\;
&
{\color{black}
\max_{\substack{\\{P\mathstrut}_{\!X}:\\
{P\mathstrut}_{\!X}\,\in \, {\cal P}_{n}({\cal X}_{n}),
\\
\mathbb{E}_{{P\mathstrut}_{\!X}}[|X|] \; \leq \; s'}}
}
\;\;
{\color{black}
\sup_{\substack{\rho\,\geq\,0}}
}
\;
{\color{black}
\inf_{\substack{\\{\color{white}{P\mathstrut}_{\!X}}\lambda, \, \kappa:{\color{white}{P\mathstrut}_{\!X}}
\\
{\color{white}{P\mathstrut}_{\!X}}\lambda \, \in \, (0, \, 1], \; \kappa \, \in \, \mathbb{R}, {\color{white}{P\mathstrut}_{\!X}}
\\
\mathbb{E}_{{P\mathstrut}_{\!X}{p\mathstrut}_{\lambda, \, \kappa}}[|Y-X|] \; \leq \; s' \, + \, 1/\nu
}}
}
\Big\{
\!-\!\rho(R - \epsilon)
\, +
\nonumber \\
&
\;\;
\tfrac{1}{\ln b}\Big[
\nu
\mathbb{E}_{{P\mathstrut}_{\!X}}\!\Big[|(\kappa - 1)X| + \tfrac{1}{\lambda\nu}e^{-\lambda\nu|(\kappa \, - \, 1)X|}\Big]
\, + \,
\rho
\ln\!
\Big(\lambda\nu
\mathbb{E}_{{P\mathstrut}_{\!X}}\!\Big[|\kappa X| + \tfrac{1}{\lambda\nu}e^{-\lambda\nu|\kappa X|}\Big]
\Big)
\, + \,
\ln\lambda
\, - \, 1
\Big]\Big\}
\nonumber \\
\overset{b}{\leq} \;\;
&
{\color{black}
\max_{\substack{\\{P\mathstrut}_{\!X}:\\
{P\mathstrut}_{\!X}\,\in \, {\cal P}_{n}({\cal X}_{n}),
\\
\mathbb{E}_{{P\mathstrut}_{\!X}}[|X|] \; \leq \; s'}}
}
\;\;
{\color{black}
\sup_{\substack{\rho\,\geq\,0}}
}
\;
{\color{black}
\inf_{\substack{\\{\color{white}{P\mathstrut}_{\!X}}\lambda, \, \kappa:{\color{white}{P\mathstrut}_{\!X}}
\\
{\color{white}{P\mathstrut}_{\!X}}\lambda \, \in \, (0, \, 1], \; \kappa \, \in \, \mathbb{R}, {\color{white}{P\mathstrut}_{\!X}}
\\
\mathbb{E}_{{P\mathstrut}_{\!X}{p\mathstrut}_{\lambda, \, \kappa}}[|Y-X|] \; \leq \; s' \, + \, 1/\nu
}}}
\Big\{
\!-\!\rho(R - \epsilon)
\, +
\nonumber \\
&
\;\;\;\;\;\;\;\;\;\;\;\;\;\;\;\;\;\;\;\;\;\;\;\;\;\;\;\;\;\;\;\;\;\;\;\;\;\;\;\;\;\;\;\;\;\;
\;\;\;\;\;\;\;\;\;\;\;
\tfrac{1}{\ln b}\Big[
\nu|\kappa - 1|s' \, + \, \tfrac{1}{\lambda}
\, + \,
\rho
\ln\!
\big(\lambda\nu|\kappa|s' + 1\big)
\, + \,
\ln\lambda
\, - \, 1
\Big]\Big\}
\label{eqBeforeMin} \\
\overset{c}{=} \;\;
&
\;\;\;\;\;\;\;\;\;\;\;\;\;\;\,\,\,\,\,\,
{\color{black}
\sup_{\substack{\rho\,\geq\,0}}
}
\;\;
\;\;\,\,
\Big\{
\! -\!
\rho(R - \epsilon)
\, + \,
\tfrac{1}{\ln b}
\min
\Big\{
\nu s',\;\;
\rho
\ln\!
\big(\lambda(\rho)\nu s' + 1\big)
\, + \,
\tfrac{1}{\lambda(\rho)}
\, + \,
\ln\lambda(\rho)
\, - \, 1
\Big\}
\Big\},
\label{eqWithEps} \\
&
\;\;\;\;\;\;\;\;\;\;\;\;\;\;\;\;\;\;\;\;\;\;\;\;\;\;\,\,
\;\;\,\,
\lambda(\rho) \; = \;
\tfrac{1}{2(1\,+\,\rho)\nu s'}
\big[
\nu s' - 1
+
\sqrt{
(\nu s' - 1)^{2}
+
4(1 + \rho)\nu s'
}
\big],
\label{eqLambdaEps}
\end{align}
where in ($a$) we take first the conditional expectations 
$\mathbb{E}_{{P\mathstrut}_{\!X}{p\mathstrut}_{\lambda, \, \kappa}}\big[|Y - X|\,|\,X\big]$   
and 
$\mathbb{E}_{{P\mathstrut}_{\!X}{p\mathstrut}_{\lambda, \, \kappa}}\big[|Y|\,|\,X\big]$ 
according to ${p\mathstrut}_{\lambda, \, \kappa}$ of (\ref{eqPlambdakappa}) and then 
in ($b$) we upper-bound the remaining expectations using the constraint $s'$ of the $\max$.
In order to justify $(c)$, we first minimize the expression of (\ref{eqBeforeMin})
over all $\lambda > 0$ and $\kappa \in \mathbb{R}$ and then check that 
the minimizing solutions $(\lambda^{*}, \kappa^{*})$ satisfy
the conditions 
under the $\inf$ of (\ref{eqBeforeMin}), 
as follows:
\begin{align}
\inf_{\substack{\\\lambda, \, \kappa:
\\
\lambda \, > \, 0, \; \kappa \, \in \, \mathbb{R}
}}
\,
&
\Big\{
\nu|\kappa - 1|s' 
\, + \,
\rho
\ln(\lambda\nu|\kappa|s' + 1)
\, + \, \tfrac{1}{\lambda}
\, + \,
\ln\lambda
\, - \, 1
\Big\}
\label{eqCocaveKappa} \\
\overset{c1}{=}
\;\;\;\;
\min
\;\;\;\;
&
\Big\{
\inf_{\substack{\\\lambda \, > \, 0}}
\big\{
\underbrace{
\nu s' \, + \, \tfrac{1}{\lambda}
\, + \,
\ln\lambda
\, - \, 1
}_{\kappa \, = \, 0}
\big\},\;\;
\inf_{\substack{\\\lambda \, > \, 0}}
\big\{
\underbrace{
\rho
\ln(\lambda\nu s' + 1)
\, + \,
\tfrac{1}{\lambda}
\, + \,
\ln\lambda
\, - \, 1
}_{\kappa \, = \, 1}
\big\}
\Big\}
\label{eqTwoInfimums} \\
\overset{c2}{=}
\;\;\;\;
\min
\;\;\;\;
&
\Big\{
\underset{\substack{\\ \\ \kappa \, = \, 0, \; \lambda \, = \, 1}}
{\nu s',}\;\;
\underbrace{
\rho
\ln\!
\big(\lambda(\rho)\nu s' + 1\big)
\, + \,
\tfrac{1}{\lambda(\rho)}
\, + \,
\ln\lambda(\rho)
\, - \, 1
}_{\kappa \, = \, 1, \; \lambda \, = \, \lambda(\rho)}
\Big\},
\label{eqAfterMin}
\end{align}
where for ($c1$) we observe that the expression in (\ref{eqCocaveKappa})
as a function of $\kappa$
is higher for $\kappa < 0$ comparing to its value at $\kappa = 0$ and 
is higher for $\kappa > 1$ comparing to its value at $\kappa = 1$,
while it is a concave ($\cap$) function of $\kappa \in [0, 1]$;
in ($c2$) we minimize both expressions over $\lambda > 0$ and obtain $\lambda(\rho) > 0$
as in (\ref{eqLambdaEps}).

Next we check if the solutions $(\lambda^{*}, \kappa^{*})$ of (\ref{eqAfterMin}) satisfy the conditions under the $\inf$ of (\ref{eqBeforeMin}). 
It can be verified that $\lambda(\rho)$ is a monotonically decreasing function with $\lambda(0) = 1$.
We conclude that $\lambda(\rho) \in (0, 1]$. 
Observe that the second infimum of (\ref{eqTwoInfimums}) 
grows with $\rho$.
Therefore the second term in (\ref{eqAfterMin}) grows with $\rho$ from zero
until it exceeds $\nu s'$.
In particular, substituting $\bar{\rho} =  1 + 2\nu s'$ in (\ref{eqLambdaEps}) we obtain $\lambda(\bar{\rho}) = \tfrac{1}{1 \, + \, \nu s'}$
and the second infimum higher than $\nu s'$:
\begin{align}
\bar{\rho}
\ln\!
\big(\lambda(\bar{\rho})\nu s' + 1\big)
+ 
\tfrac{1}{\lambda(\bar{\rho})}
+ \ln\lambda(\bar{\rho})
- 1
\, & = \,
(1 + 2\nu s')\ln(1 + 2\nu s')  -  2 (1 + \nu s')\ln(1 + \nu s') + \nu s' 
\, > \, \nu s'.
\nonumber
\end{align}
So if for some $\rho$ the second term in (\ref{eqAfterMin}) is smaller than $\nu s'$, 
then the achieving $\lambda^{*} = \lambda(\rho)$ is greater than $\lambda(\bar{\rho})$.
Then
\begin{align}
\mathbb{E}_{{P\mathstrut}_{\!X}{p\mathstrut}_{\lambda^{*}\!, \, \kappa^{*}}}\big[|Y-X|\big]
\; & \leq \;
|\kappa^{*} - 1|\,
\mathbb{E}_{{P\mathstrut}_{\!X}{p\mathstrut}_{\lambda^{*}\!, \, \kappa^{*}}}\big[|X|\big]
\, + \,
\mathbb{E}_{{P\mathstrut}_{\!X}{p\mathstrut}_{\lambda^{*}\!, \, \kappa^{*}}}\big[|Y \!- \kappa^{*} X|\big]
\nonumber \\
& \leq \;
\Bigg\{
\begin{array} {r l}
\!\! \tfrac{1}{{\lambda\mathstrut}^{*}\nu} \, < \, \tfrac{1}{\lambda(\bar{\rho})\nu} \, = \, s' + \tfrac{1}{\nu}, 
& \;\;\; \kappa^{*} = 1, \; \lambda^{*} = \lambda(\rho), \\
\!\! s' + \tfrac{1}{\nu}, & \;\;\; \kappa^{*} = 0, \; \lambda^{*} = 1.
\end{array}
\nonumber
\end{align}

In conclusion, in the case of $q = 1$, $r = 1$
the maximum of (\ref{eqGenegalqr}) can be replaced with (\ref{eqWithEps}). 
The resulting inequality holds for all $R > 0$,
except possibly for the single point $R = \epsilon$,
where the jump to $+\infty$ in (\ref{eqWithEps}) occurs.
Taking $\limsup_{\substack{\,n\,\rightarrow\,\infty}}$ on both sides of (\ref{eqGenegalqr}) and $\epsilon \rightarrow 0$
on the RHS,
we obtain (\ref{eqLaplace}).
\end{proof}

\bigskip

\begin{thm}[Case $q = 2$, $r = 1$] \label{thmSUB21}
{\em Let the channel be given by (\ref{eqChannel}) with $q = 2$, 
let codebooks ${\cal C}$ satisfy (\ref{eqPowerConstraint}) with $r = 1$,
and let decoders $g$ be defined by (\ref{eqDec}). 
Let $J \sim \text{Unif}\,\big(\{1, 2, .\,.\,.\, , \, M\}\big)$ be a random variable,
independent of the channel noise, and let ${\bf x}(J) \rightarrow {\bf Y}$ 
be the random channel-input and channel-output vectors, respectively. Then
}
\begin{align}
&
\limsup_{\substack{n\,\rightarrow\,\infty}}
\;
\sup_{{\cal C}}
\;
\sup_{g}
\;
\Big\{\!
- \tfrac{1}{n}\,{\log\mathstrut}_{\!b} \Pr \big\{
g({\bf Y}) \neq J \big\}
\Big\}
\label{eqGaussAbs} \\
\leq
\;\;
&
\sup_{\substack{\rho\,\geq\,0}}
\;
\Big\{
\tfrac{1}{\ln b}
\Big[
\rho
\ln\!
\big(s\sqrt{\nu\pi\lambda(\nu, s, \rho)} + 1\big)
\, + \,
\tfrac{1}{2}
\big[
\tfrac{1}{\lambda(\nu, \, s, \, \rho)} \, + \, \ln\lambda(\nu, s, \rho) \, - \, 1
\big]
\, + \,
\rho\ln\!
\big(2\sqrt{e}/\pi\big)
\Big]
\, - \,
\rho R
\Big\},
\nonumber
\end{align}
{\em where $\lambda(\nu, s, \rho)$ is the solution of the equation
$\rho = \Big[\frac{1}{\lambda} - 1\Big]\Big[\frac{1}{s\sqrt{\nu\pi\lambda}} + 1\Big]$
for $\lambda \in (0, 1]$.
}
\end{thm}

\bigskip

\begin{proof}
With $q = 2$ and $r = 1$ we further upper-bound the infimum of (\ref{eqGenegalqr}) 
restricting it to $\kappa = 1$:
\begin{align}
&
{\color{black}
\max_{\substack{\\{P\mathstrut}_{\!X}:\\
{P\mathstrut}_{\!X}\,\in \, {\cal P}_{n}({\cal X}_{n}),
\\
\mathbb{E}_{{P\mathstrut}_{\!X}}[|X|] \; \leq \; s'}}
}
\;\;
{\color{black}
\sup_{\substack{\rho\,\geq\,0}}
}
\;
{\color{black}
\inf_{\substack{\\{\color{white}{P\mathstrut}_{\!X}}\lambda:{\color{white}{P\mathstrut}_{\!X}}
\\
{\color{white}{P\mathstrut}_{\!X}}\lambda \, \in \, (0, \, 1], {\color{white}{P\mathstrut}_{\!X}}
\\
\mathbb{E}_{{P\mathstrut}_{\!X}{p\mathstrut}_{\lambda, \, 1}}[(Y-X)^{2}] \; \leq \; c_{XY}
}}
}
\Big\{\!
-\rho(R - \epsilon)
\, +
\nonumber \\
&
\;\;\;\;\;\;\;\;\,\,
\tfrac{1}{\ln b}\Big[
\nu
\mathbb{E}_{{P\mathstrut}_{\!X}{p\mathstrut}_{\lambda, \, 1}}\big[(Y-X)^{2}\big]
\, + \,
\rho
\ln
\mathbb{E}_{{P\mathstrut}_{\!X}{p\mathstrut}_{\lambda, \, 1}}\big[|Y|\big]
\, + \,
\tfrac{\rho\,+\,1}{2}\ln\lambda
\, - \,
\tfrac{1}{2}
\, + \,
\rho\ln\!
\big(2\sqrt{e\nu/\pi}\big)
\Big]\Big\}
\nonumber \\
\overset{a}{=} \;\;
&
{\color{black}
\max_{\substack{\\{P\mathstrut}_{\!X}:\\
{P\mathstrut}_{\!X}\,\in \, {\cal P}_{n}({\cal X}_{n}),
\\
\mathbb{E}_{{P\mathstrut}_{\!X}}[|X|] \; \leq \; s'}}
}
\;\;
{\color{black}
\sup_{\substack{\rho\,\geq\,0}}
}
\;\;\,\,\,
{\color{black}
\min_{\substack{\\
\lambda:
\;
(2\nu c_{XY}\!)^{-1} \, \leq \, \lambda \, \leq \, 1
}}
}
\;\;\;
\Big\{\!
-\rho(R - \epsilon)
\, +
\nonumber \\
&
\;\;\;\;\;\;\;\;\;\;\;\;\;\;\;\;\;\;\;\;\;\;\;\;\;\;\;\;\;\;\;\;\;\;\,\,\,\,
\tfrac{1}{\ln b}\Big[
\tfrac{1}{2\lambda}
\, + \,
\rho
\ln
\mathbb{E}_{{P\mathstrut}_{\!X}{p\mathstrut}_{\lambda, \, 1}}\big[|Y|\big]
\, + \,
\tfrac{\rho\,+\,1}{2}\ln\lambda
\, - \,
\tfrac{1}{2}
\, + \,
\rho\ln\!
\big(2\sqrt{e\nu/\pi}\big)
\Big]\Big\}
\nonumber \\
\overset{b}{\leq} \;\;
&
\;\;\;\;\;\;\;\;\;\;\;\;\;\;\;\;\;\,
{\color{black}
\sup_{\substack{\rho\,\geq\,0}}
}
\;\;\,\,\,
{\color{black}
\min_{\substack{\\
\lambda:
\;
(2\nu c_{XY}\!)^{-1} \, \leq \, \lambda \, \leq \, 1
}}
}
\;\;\;
\Big\{\!
-\rho(R - \epsilon)
\, +
\nonumber \\
&
\;\;\;\;\;\;\;\;\;\;\;\;\;\;\;\;\;\;\;\;\;\;\;\;\;\;\;\;\;\;\;\;\,\,
\tfrac{1}{\ln b}\Big[
\tfrac{1}{2\lambda}
\, + \,
\rho
\ln\!
\big(s' + 1/\sqrt{\lambda\nu\pi}\big)
\, + \,
\tfrac{\rho\,+\,1}{2}\ln \lambda
\, - \,
\tfrac{1}{2}
\, + \,
\rho\ln\!
\big(2\sqrt{e\nu/\pi}\big)
\Big]\Big\}
\label{eqTight} \\
\overset{c}{=} \;\;
&
\;\;\;\;\;\;\;\;\;\;\;\;\;\;\;\;\;\,
{\color{black}
\sup_{\substack{\rho\,\geq\,0}}
}
\;\;\,\,\,
{\color{white}
\min_{\substack{\\
\lambda:
\;
(2\nu c_{XY}\!)^{-1} \, \leq \, \lambda \, \leq \, 1
}}
}
\;\;\;
\Big\{\!
-\rho(R - \epsilon)
\, +
\nonumber \\
&
\,
\tfrac{1}{\ln b}\Big[
\tfrac{1}{2\lambda(c_{XY}\!, \, \rho)}
\, + \,
\rho
\ln\!
\big(s' + 1/\sqrt{\lambda(c_{XY}\!, \rho)\nu\pi}\big)
\, + \,
\tfrac{\rho\,+\,1}{2}\ln\lambda(c_{XY}\!, \rho)
\, - \, \tfrac{1}{2}
\, + \,
\rho\ln\!
\big(2\sqrt{e\nu/\pi}\big)
\Big]\Big\},
\label{eqMaxCXYEps}
\end{align}
where $s' = s + \epsilon$,
in ($a$) we substitute 
$\mathbb{E}_{{P\mathstrut}_{\!X}{p\mathstrut}_{\lambda, \, 1}}\big[(Y-X)^{2}\big] = \frac{1}{2\lambda\nu}$, 
while in ($b$) we use
$|Y| \leq |X|+|Y-X|$ 
with $\mathbb{E}_{{P\mathstrut}_{\!X}{p\mathstrut}_{\lambda, \, 1}}\big[|X|\big]\leq s'$
and
$\mathbb{E}_{{P\mathstrut}_{\!X}{p\mathstrut}_{\lambda, \, 1}}\big[|Y-X|\big] = 1/\sqrt{\lambda\nu\pi}$.
For ($c$) we assume $c_{XY} \geq \tfrac{1}{2\nu}$. 
Observe that the objective function of (\ref{eqTight}) tends to $+\infty$ as $\lambda \rightarrow 0$,
while its derivative w.r.t. $\lambda$ is positive for $\lambda > 1$:
\begin{equation} \label{eqDerivativeLambda}
\frac{1}{2\lambda} \, - \, \frac{1}{2\lambda^{2}} \, + \, 
\frac{\rho}{2\lambda} \, - \, \frac{\rho}{2\lambda \, + \, 2s'\sqrt{\nu\pi}\lambda^{3/2}}
\; > \; 0 , \;\;\;\;\;\; \lambda > 1.
\end{equation}
We conclude that the objective function of (\ref{eqTight}) attains its minimum over $\lambda > 0$
when its derivative is zero, which occurs by the expression for the derivative in (\ref{eqDerivativeLambda}) at $\lambda = \lambda(\rho)\in (0, 1]$ solving the equation
\begin{equation} \nonumber
\rho \; = \; \bigg[\frac{1}{\lambda} - 1\bigg]\bigg[\frac{1}{s'\sqrt{\nu\pi\lambda}} + 1\bigg],
\;\;\;\;\;\; \lambda \in (0, 1].
\end{equation}
It follows that the minimizer of (\ref{eqTight}) itself for $c_{XY} \geq \tfrac{1}{2\nu}$
is $\lambda^{*} = \lambda(c_{XY}\!, \rho) \triangleq \max\big\{\tfrac{1}{2\nu c_{XY}}, \;\lambda(\rho)\big\}$.

We conclude that in the case of $q = 2$, $r = 1$
the maximum of (\ref{eqGenegalqr}) can be replaced with the supremum (\ref{eqMaxCXYEps}). 
The resulting inequality holds for all $R > 0$,
except possibly for the boundary point
\begin{equation} \label{eqBoundaryPoint}
R \; = \; R_{\min} \; \triangleq \; {\log\mathstrut}_{\!b}
\big(s'\sqrt{\pi/(2c_{XY})} + 1\big) \, + \, 
{\log\mathstrut}_{\!b} \big(2\sqrt{e}/\pi\big) \, + \,
\epsilon,
\end{equation}
where the supremum (\ref{eqMaxCXYEps}) as a function of $R$ transitions between $+\infty$ and a finite value. 
With $\limsup_{\substack{\,n\,\rightarrow\,\infty}}$ on both sides of (\ref{eqGenegalqr}) we obtain that
the supremum (\ref{eqMaxCXYEps}) can replace also the RHS of (\ref{eqGaussAbs}).
The RHS of (\ref{eqGaussAbs}) itself corresponds to the minimization in (\ref{eqTight})
over $\lambda \in (0, 1]$ with $\epsilon = 0$. 
The boundary point (\ref{eqBoundaryPoint}) can be used 
to verify that (\ref{eqMaxCXYEps}) converges to
the RHS of (\ref{eqGaussAbs}) as $\epsilon \rightarrow 0$ and $c_{XY}\rightarrow +\infty$.
For this, consider any point $(R_{1}, E_{1})$ on the $R$-$E$-plane,
with $R_{1} > {\log\mathstrut}_{\!b} \big(2\sqrt{e}/\pi\big)$, 
that lies above the graph of the supremum (\ref{eqGaussAbs})
as a function of $R$.
By (\ref{eqBoundaryPoint}) for a sufficiently small $\epsilon > 0$ and a sufficiently large $c_{XY}$ there exists
a point $(R_{2}, E_{2})$, 
with $R_{\min} < R_{2} < R_{1}$ and $E_{2} > E_{1}$, that lies above the graph of the supremum 
(\ref{eqMaxCXYEps}) 
with the parameters $(\epsilon, c_{XY})$.
Note that for $\rho > \frac{E_{2}\,- \,E_{1}}{R_{1} \, - \, R_{2}}$ the inner minimum of (\ref{eqTight})
with the same $(\epsilon, c_{XY})$ is {\em both} lower than $E_{2}$ at $R = R_{2}$ {\em and} lower than 
$E_{1}$ at $R = R_{1}$. 
Since 
this minimum 
does not increase
when $\epsilon$ decreases and $c_{XY}$ increases,
while for $\rho \leq \frac{E_{2}\,- \,E_{1}}{R_{1} \, - \, R_{2}}$
it is a uniformly continuous function of $\epsilon \geq 0$ and $c_{XY}\geq \tfrac{1}{2\nu}$,
we find that there exist $\tilde{\epsilon} < \epsilon$ and $\tilde{c}_{XY} > c_{XY}$ such that both points $(R_{2}, E_{2})$
and $(R_{1}, E_{1})$ lie above the graph of 
(\ref{eqMaxCXYEps}) 
with the parameters $(\tilde{\epsilon}, \tilde{c}_{XY})$.
\end{proof}

\bigskip

{\em Remark:} There is no loss in 
($b$) of (\ref{eqBeforeMin}) in the proof of Theorem~\ref{thmSUB11}, and no loss in ($b$) of (\ref{eqTight})
 in the proof of Theorem~\ref{thmSUB21}.
One type asymptotically achieving equalities in both steps is:
\begin{displaymath}
{P\mathstrut}_{\!X}(0) \,=\, 1 - 1/n, \;\;\;\;\;\;
{P\mathstrut}_{\!X}\big(\Delta_{\alpha,\,n}\cdot \lfloor s'n/\Delta_{\alpha,\,n}\rfloor\big) \,=\, 1/n.
\end{displaymath}
This is similar to the step ($f$) of (\ref{eqExplicitSUB}).

\bigskip

\begin{thm}[Case $q = 1$, $r = 2$] \label{thmSUB12}
{\em Let the channel be given by (\ref{eqChannel}) with $q = 1$, 
let codebooks ${\cal C}$ satisfy (\ref{eqPowerConstraint}) with $r = 2$,
and let decoders $g$ be defined by (\ref{eqDec}). 
Let $J \sim \text{Unif}\,\big(\{1, 2, .\,.\,.\, , \, M\}\big)$ be a random variable,
independent of the channel noise, and let ${\bf x}(J) \rightarrow {\bf Y}$ 
be the random channel-input and channel-output vectors, respectively. Then
}
\begin{align}
&
\limsup_{\substack{n\,\rightarrow\,\infty}}
\;
\sup_{{\cal C}}
\;
\sup_{g}
\;
\Big\{\!
- \tfrac{1}{n}\,{\log\mathstrut}_{\!b} \Pr \big\{
g({\bf Y}) \neq J \big\}
\Big\}
\;\;
\leq
\;\;\,
\sup_{\substack{\rho\,\geq\,0}}
\;\;
\min_{\substack{\\i\,=\, 1, \,2}}
\;
\Big\{\! -\!\rho R \; +
\label{eqLapSq} \\
&
\tfrac{1}{\ln b}
\Big[
\tfrac{\rho}{2}
\ln\!
\Big(
\tfrac{1}{2}\big[\nu s\kappa_{i}(\nu s, \rho)\lambda_{i}(\nu s, \rho)\big]^{2} + 1\Big)
\, + \,
\big[1 - \kappa_{i}(\nu s, \rho)\big]\nu s
\, + \,
\tfrac{1}{\lambda_{i}(\nu s, \, \rho)} \, + \, \ln\lambda_{i}(\nu s, \rho) \, - \, 1
\, + \,
\tfrac{\rho}{2}\ln\tfrac{\pi}{e}
\Big]
\Big\},
\nonumber
\end{align}
{\em where with $\rho_{1} \triangleq 4 + 2\sqrt{6}$ and $\rho_{2} \triangleq \frac{2}{\nu s} + 4 + 3\nu s$
we define two competing solutions $i = 1, 2$ for each $\rho$:}
\begin{align}
\lambda_{1}(\nu s, \rho) \; & = \;
\Bigg\{
\begin{array}{l r}
\tfrac{1}{1 \, + \, \nu s}, & 
\;\;\;\;\;\;\,\, 
0 \, \leq \, \rho \, \leq \, \rho_{1}, \\
1 \, - \, \nu s\min\big\{\phi(\nu s, \rho), \; \tfrac{1}{1 \, + \, \nu s}\big\}, & \rho \, > \, \rho_{1},
\end{array}
\nonumber \\
\kappa_{1}(\nu s, \rho) \; & = \;
\Bigg\{
\begin{array}{l r}
1, & \;\;\; 0 \, \leq \, \rho \, \leq \, \rho_{1}, \\
\lambda_{1}^{-1}(\nu s, \rho)\min\big\{\phi(\nu s, \rho), \; \tfrac{1}{1 \, + \, \nu s}\big\}, & \rho \, > \, \rho_{1},
\end{array}
\nonumber \\
\phi(\nu s, \rho) \; & = \;
\tfrac{1}{2(1\,+\,\rho)\nu s}
\big[
\rho - \sqrt{\rho^{2} - 8(1 + \rho)}
\big],
\nonumber \\
\lambda_{2}(\nu s, \rho) \; & = \;
\Bigg\{
\begin{array}{l r}
\lambda(\nu s, \rho), & 
\;\;\;\;\;\;\;\;\;\;\;\;\;\;\;\;\;\;\;\;\;\;\;\;\;\;\;\;\;\;\;\;\;\;\;\;\;\;\,\,
0 \, \leq \, \rho \, \leq \, \rho_{2}, \\
\tfrac{1}{1 \, + \, \nu s}, & \rho \, > \, \rho_{2},
\end{array}
\label{eqCubic2}
\end{align}
{\em 
$\kappa_{2}(\nu s, \rho) \equiv 1$, and 
$\lambda(\nu s, \rho)$ in (\ref{eqCubic2}) is the solution of the equation
$\rho = \Big[\frac{1}{\lambda} - 1\Big]\Big[\frac{2}{(\nu s\lambda)^{2}} + 1\Big]$
for $\lambda \in (0, 1]$.
The inequality holds
for all $R > 0$ with the possible exception of the single point $R = \frac{1}{2}\,{\log\mathstrut}_{\!b} \,\frac{\pi}{e}$
where the RHS of (\ref{eqLapSq}) as a function of $R$ transitions between $+\infty$ and $\nu s/\ln b$.
}
\end{thm}

\bigskip

\begin{proof}
Let us rewrite the maximum of (\ref{eqGenegalqr}) for $q = 1$, $r = 2$, with 
$s^{2} \!+ \epsilon = \tilde{s}^{2}$:
\begin{align}
&
{\color{black}
\max_{\substack{\\{P\mathstrut}_{\!X}:\\
{P\mathstrut}_{\!X}\,\in \, {\cal P}_{n}({\cal X}_{n}),
\\
\mathbb{E}_{{P\mathstrut}_{\!X}}[X^{2}] \; \leq \; \tilde{s}^{2}}}
}
\;\;\;\;
{\color{black}
\sup_{\substack{\rho\,\geq\,0}}
}
\;\;\;\;\;\;\,\,
{\color{black}
\inf_{\substack{\\{\color{white}{P\mathstrut}_{\!X}}\lambda, \, \kappa:{\color{white}{P\mathstrut}_{\!X}}
\\
{\color{white}{P\mathstrut}_{\!X}}\lambda \, \in \, (0, \, 1], \; \kappa \, \in \, \mathbb{R}, {\color{white}{P\mathstrut}_{\!X}}
\\
\mathbb{E}_{{P\mathstrut}_{\!X}{p\mathstrut}_{\lambda, \, \kappa}}[|Y-X|] \; \leq \; c_{XY}
}}
}
\;\;\;\;\;\;\;
\Big\{\!
-\!\rho(R - \epsilon)
\, +
\nonumber \\
&
\;\;\;\;\;\;\;\;\;\;\;\;\;\;\;\;\;\;\;\;\,\,\,\,
\;\;\;\;\;\;\;\;\;\;\;\;\;\;\;\;\;\;\;\;\;\;\;\;\;
\tfrac{1}{\ln b}\Big[
\nu
\mathbb{E}_{{P\mathstrut}_{\!X}{p\mathstrut}_{\lambda, \, \kappa}}\big[|Y-X|\big]
\, + \,
\tfrac{\rho}{2}
\ln\!
\Big(
\tfrac{\pi}{2e}\lambda^{2}\nu^{2}
\mathbb{E}_{{P\mathstrut}_{\!X}{p\mathstrut}_{\lambda, \, \kappa}}\big[Y^{2}\big]
\Big)
\, + \,
\ln\lambda
\, - \, 1
\Big]\Big\}
\nonumber \\
\overset{a}{\leq} \;\;
&
{\color{black}
\max_{\substack{\\{P\mathstrut}_{\!X}:\\
{P\mathstrut}_{\!X}\,\in \, {\cal P}_{n}({\cal X}_{n}),
\\
\mathbb{E}_{{P\mathstrut}_{\!X}}[X^{2}] \; \leq \; \tilde{s}^{2}}}
}
\;\;\;\;
{\color{black}
\sup_{\substack{\rho\,\geq\,0}}
}
{\color{black}
\inf_{\substack{\\{\color{white}{P\mathstrut}_{\!X}}\lambda, \, \kappa:{\color{white}{P\mathstrut}_{\!X}}
\\
{\color{white}{P\mathstrut}_{\!X}}\lambda \, \in \, (0, \, 1], \; \kappa \, \in \, [0, \, 1], {\color{white}{P\mathstrut}_{\!X}}
\\
\mathbb{E}_{{P\mathstrut}_{\!X}{p\mathstrut}_{\lambda, \, \kappa}}[|(\kappa - 1)X| \, + \, |Y \!- \kappa X|] \; \leq \; c_{XY}
}}
}
\Big\{\!
-\!\rho(R - \epsilon)
\, +
\label{eqObjFunc01} \\
&
\,\,
\tfrac{1}{\ln b}\Big[
\nu
\mathbb{E}_{{P\mathstrut}_{\!X}{p\mathstrut}_{\lambda, \, \kappa}}
\big[|(\kappa - 1)X| + |Y \!- \kappa X|\big]
\, + \,
\tfrac{\rho}{2}
\ln\!
\Big(
\tfrac{\pi}{2e}\lambda^{2}\nu^{2}
\mathbb{E}_{{P\mathstrut}_{\!X}{p\mathstrut}_{\lambda, \, \kappa}}
\big[\kappa^{2}X^{2} + (Y \!- \kappa X)^{2}\big]
\Big)
\, + \,
\ln\lambda
\, - \, 1
\Big]\Big\}
\nonumber \\
\overset{b}{\leq} \;\;
&
{\color{white}
\max_{\substack{\\{P\mathstrut}_{\!X}:\\
{P\mathstrut}_{\!X}\,\in \, {\cal P}_{n}({\cal X}_{n}),
\\
\mathbb{E}_{{P\mathstrut}_{\!X}}[X^{2}] \; \leq \; \tilde{s}^{2}}}
}
\;\;\;\;
{\color{black}
\sup_{\substack{\rho\,\geq\,0}}
}
\;\;\;\;\;\;\;\;\,
{\color{black}
\inf_{\substack{\\{\color{white}{P\mathstrut}_{\!X}}\lambda, \, \kappa:{\color{white}{P\mathstrut}_{\!X}}
\\
{\color{white}{P\mathstrut}_{\!X}}
\kappa \, \in \, [0, \, 1], 
{\color{white}{P\mathstrut}_{\!X}}
\\
[\nu (c_{XY}- \,\tilde{s})]^{-1} \, \leq \, \lambda \, \leq \, 1
}}
}
\;\;\;\;\;\;\;\;\,
\Big\{\!
-\!\rho(R - \epsilon)
\, +
\nonumber \\
&
\;\;\;\;\;\;\;\;\;\;\;\;\;\;\;\;\;\;\;\;\;\;\;\;
\;\;\;\;\;\;\;\;\;\;\;\;\;\;\;\;\;\;\,\,
\;\;\;\,\,\,
\tfrac{1}{\ln b}\Big[
(1 - \kappa) \nu\tilde{s} 
\, + \, \tfrac{1}{\lambda}
\, + \,
\tfrac{\rho}{2}
\ln\!
\big(\tfrac{1}{2}\kappa^{2}\lambda^{2}\nu^{2}\tilde{s}^{2} + 1\big)
\, + \,
\tfrac{\rho}{2}\ln\tfrac{\pi}{e}
\, + \,
\ln\lambda
\, - \, 1
\Big]\Big\}
\nonumber \\
\overset{c}{\leq} \;\;
&
\;\;\;\;\;\;\;\;\;\;\;\;\;\;\;\;\;\,\,\,\,\,
{\color{black}
\sup_{\substack{\rho\,\geq\,0}}
}
\;\;\;\;\;\;
{\color{black}
\min_{\substack{\\
\zeta: \; \epsilon_{0} \, \leq \, \zeta \, \leq\, 1
}}
}
\;\;\;
{\color{black}
\min_{\substack{\\
\lambda: \;
\zeta \, \leq\, \lambda \, \leq \, 1
}}
}
\;\;\;\;\;\;\;
\Big\{\!
-\!\rho(R - \epsilon)
\, +
\label{eqMinoverZeta} \\
&
\;\;\;\;\;\;\;\;\;\;\;\;\;\;\;\;\;\;\;\;\;\;\;\;
\;\;\;\;\;\;\;\;\;\;\;\;\;\;\;\;\;\,\,\,\,
\;\;\;\,\,\,
\tfrac{1}{\ln b}\Big[
(1 - \zeta/\lambda)\nu\tilde{s} 
\, + \, \tfrac{1}{\lambda}
\, + \,
\tfrac{\rho}{2}
\ln\!
\big(\tfrac{1}{2}\zeta^{2}\nu^{2}\tilde{s}^{2} + 1\big)
\, + \,
\tfrac{\rho}{2}\ln\tfrac{\pi}{e}
\, + \,
\ln\lambda
\, - \, 1
\Big]\Big\}
\nonumber \\
\overset{d}{=} \;\;
&
\;\;\;\;\;\;\;\;\;\;\;\;\;\;\;\;\;\,\,\,\,\,
{\color{black}
\sup_{\substack{\rho\,\geq\,0}}
}
\;\;\;\;\;\;\;\;\;
{\color{black}
\min_{\substack{\\
i \, = \, 1, \, 2
}}
}
\;\;\;\;\;\;\;\;\;\;\;\;\;\;\;\;\;\;\;\;\;\;\;\;\;\,
\Big\{\!
-\!\rho(R - \epsilon)
\, +
\label{eqTwoMins} \\
&
\!\!\!\!
\tfrac{1}{\ln b}\Big[
\Big(1 - \tfrac{\zeta_{i}(\epsilon_{0}, \, \nu\tilde{s}, \, \rho)}{\lambda_{i}(\epsilon_{0}, \,\nu\tilde{s}, \,\rho)}\Big)\nu\tilde{s} 
\, + \, \tfrac{1}{\lambda_{i}(\epsilon_{0}, \,\nu\tilde{s}, \,\rho)}
\, + \,
\tfrac{\rho}{2}
\ln\!
\big(\tfrac{1}{2}\zeta^{2}_{i}(\epsilon_{0}, \nu\tilde{s}, \rho)\nu^{2}\tilde{s}^{2} + 1\big)
\, + \,
\tfrac{\rho}{2}\ln\tfrac{\pi}{e}
\, + \,
\ln\lambda_{i}(\epsilon_{0}, \nu\tilde{s}, \rho)
\, - \, 1
\Big]\Big\},
\nonumber
\end{align}
where the inequalities can be justified as follows:

($a$) In the second expectation we use  the equality 
$\mathbb{E}_{{P\mathstrut}_{\!X}{p\mathstrut}_{\lambda, \, \kappa}}\big[X(Y \!- \kappa X)\big] = 0$
according to (\ref{eqPlambdakappa});
in the first expectation and inside the condition under the $\inf$ we use the inequality 
$|Y - X| \leq |(\kappa-1)X|+|Y\!-\kappa X|$. This gives both a higher objective function and a stricter minimization condition. 
At this point, both the objective function and the expectation under the $\inf$ of (\ref{eqObjFunc01})
can only be larger for $\kappa < 0$ and for $\kappa > 1$, 
comparing to their values at $\kappa = 0$ and $\kappa = 1$, respectively.
So it suffices to minimize over $\kappa \in [0, 1]$.

($b$) Follows by the power constraint, 
by Jensen's inequality: 
$\mathbb{E}_{{P\mathstrut}_{\!X}}\big[|X|\big] \leq \big(\mathbb{E}_{{P\mathstrut}_{\!X}}\big[X^{2}\big]\big)^{1/2} \leq \tilde{s}$, 
and by the properties of the Laplace distribution: 
$\mathbb{E}_{{P\mathstrut}_{\!X}{p\mathstrut}_{\lambda, \, \kappa}}
\big[|Y \!- \kappa X|\big] = \tfrac{1}{\lambda\nu}$ and 
$\mathbb{E}_{{P\mathstrut}_{\!X}{p\mathstrut}_{\lambda, \, \kappa}}
\big[(Y \!- \kappa X)^{2}\big] = \tfrac{2}{\lambda^{2}\nu^{2}}$. 
Accordingly, we also apply a stricter minimization condition:
\begin{displaymath}
\mathbb{E}_{{P\mathstrut}_{\!X}{p\mathstrut}_{\lambda, \, \kappa}}
\big[|(\kappa - 1)X| + |Y \!- \kappa X|\big]
\; \leq \; \tilde{s} \, + \, \tfrac{1}{\lambda\nu}
\; \leq \; c_{XY},
\end{displaymath}
assuming in the formulation of the condition $c_{XY}$ large enough such that $c_{XY} > \tilde{s}$ and
$[\nu (c_{XY} \! - \tilde{s})]^{-1} \leq 1$.

($c$) 
Denoting $\epsilon_{0} \triangleq [\nu (c_{XY} \! - \tilde{s})]^{-1}$ 
and $\zeta \triangleq \kappa\lambda$, we apply a slightly stricter minimization condition
$\lambda \geq \zeta \geq \epsilon_{0}$
and minimize first over $\lambda \geq \zeta$ given $\zeta$, and then over $\zeta \geq \epsilon_{0}$.

For ($d$) we assume $\epsilon_{0} \leq  \tfrac{1}{1\,+\,\nu \tilde{s}}$ and 
split the minimization over $\zeta$ in (\ref{eqMinoverZeta}) in two parts.
We minimize separately over $\zeta \in \big[\epsilon_{0}, \; \tfrac{1}{1\,+\,\nu \tilde{s}}\big]$
and 
$\zeta \in \big[\tfrac{1}{1\,+\,\nu \tilde{s}}, \; 1\big]$,
and then take the minimum between the two results.

For the minimum over $\zeta \in \big[\epsilon_{0}, \; \tfrac{1}{1\,+\,\nu \tilde{s}}\big]$ we find:
\begin{align}
&
{\color{black}
\min_{\substack{\\
\zeta: \\
\epsilon_{0} \, \leq \, \zeta \, \leq\, 1/(1 \, + \, \nu \tilde{s})
}}
}
\;\;
{\color{black}
\min_{\substack{\\
{\color{white}\zeta}
\lambda:
{\color{white}\zeta}
\\
{\color{white}/}
\zeta \, \leq\, \lambda \, \leq \, 1
{\color{white}/}
}}
}
\,\,\,
\Big\{
\tfrac{\rho}{2}
\ln\!
\big(\tfrac{1}{2}\nu^{2}\tilde{s}^{2}\zeta^{2} + 1\big)
\, + \,
\ln\lambda
\, + \,
(1 - \nu\tilde{s}\zeta)/\lambda
\Big\}
\nonumber \\
\overset{d1}{=} \;
&
{\color{black}
\min_{\substack{\\
\zeta: \\
\epsilon_{0} \, \leq \, \zeta \, \leq\, 1/(1 \, + \, \nu \tilde{s})
}}
}
\;\;
{\color{white}
\min_{\substack{\\
{\color{white}\zeta}
\lambda:
{\color{white}\zeta}
\\
{\color{white}/}
\zeta \, \leq\, \lambda \, \leq \, 1
{\color{white}/}
}}
}
\,\,\,
\Big\{
\tfrac{\rho}{2}
\ln\!
\big(\tfrac{1}{2}\nu^{2}\tilde{s}^{2}\zeta^{2} + 1\big)
\, + \,
\ln(1 - \nu\tilde{s}\zeta)
\, + \,
1
\Big\}
\label{eqExpressionZeta} \\
\overset{d2}{=} \;
&
\;\;\;\;\;\;
\min_{\substack{\\
i \, = \, 1, \, 2
}}
\;\;
\;\;\;\;\;\;\;\;\;\;\;\;\;\;\;\;\;\;\;\;
\Big\{
\tfrac{\rho}{2}
\ln\!
\big(\tfrac{1}{2}\nu^{2}\tilde{s}^{2}\zeta^{2}_{i}(\epsilon_{0}, \nu\tilde{s}, \rho) + 1\big)
\, + \,
\ln\!\big(1 - \nu\tilde{s}\zeta_{i}(\epsilon_{0}, \nu\tilde{s}, \rho)\big)
\, + \,
1
\Big\},
\label{eqMinFirst} \\
& 
\;\;
\zeta_{1}(\epsilon_{0}, \nu\tilde{s}, \rho) \; = \;
\Bigg\{
\begin{array} {l r}
\tfrac{1}{1 \, + \, \nu\tilde{s}},
& \;\;\; 0 \, \leq \, \rho \, \leq \, 4 + 2\sqrt{6},  \\
\min\big\{\max\{\epsilon_{0}, \; \phi(\nu \tilde{s}, \rho)\}, \; \tfrac{1}{1 \, + \, \nu \tilde{s}}\big\},
& \rho \, > \, 4 + 2\sqrt{6},
\end{array}
\label{eqRho1} \\
& 
\;\;\;\;\;\;\,\,
\phi(\nu \tilde{s}, \rho) \; = \;
\tfrac{1}{2(1\,+\,\rho)\nu \tilde{s}}
\big[
\rho - \sqrt{\rho^{2} - 8(1 + \rho)}
\big],
\label{eqLocMin} \\
& 
\;\;
\zeta_{2}(\epsilon_{0}, \nu\tilde{s}, \rho) \; = \; \tfrac{1}{1 \, + \, \nu \tilde{s}},
\nonumber
\end{align}
where in ($d1$) the minimizer is $\lambda^{*} = 1 - \nu\tilde{s}\zeta$.
For ($d2$), consider the expression in (\ref{eqExpressionZeta})
as a function of $\zeta \in \big[0, \tfrac{1}{\nu \tilde{s}}\big)$.
This function equals $1$ at $\zeta = 0$ and tends to $-\infty$
as $\zeta \rightarrow \tfrac{1}{\nu \tilde{s}}$.
Its derivative is negative at $\zeta = 0$. 
Equating the derivative to zero, we obtain an equation:
\begin{equation} \label{eqZeros}
\rho\; = \; \frac{2 + \nu^{2}\tilde{s}^{2}\zeta^{2}}{\nu\tilde{s}\zeta(1 - \nu\tilde{s}\zeta)},
\;\;\;\;\;\; \zeta \in \big(0, \tfrac{1}{\nu\tilde{s}}\big).
\end{equation}
It has no solutions for $\rho < \rho_{1} \triangleq 4 + 2\sqrt{6}$
and a single solution for $\rho = \rho_{1}$,
so that the derivative is non-positive and the minimizer of (\ref{eqExpressionZeta}) is 
$\zeta_{1}(\epsilon_{0}, \nu\tilde{s}, \rho) = \frac{1}{1\,+\,\nu \tilde{s}}$ in these cases.
For $\rho > \rho_{1}$ there are two distinct solutions to (\ref{eqZeros}). 
We conclude that the smaller of the two solutions, which is given by (\ref{eqLocMin}) as $\phi(\nu \tilde{s}, \rho)$, 
corresponds to a local minimum,
while the larger solution corresponds to a local maximum.
It follows that for $\rho > \rho_{1}$ the minimizer of (\ref{eqExpressionZeta}) is either 
$\zeta_{1}(\epsilon_{0}, \nu\tilde{s}, \rho) = \min\big\{\max\{\epsilon_{0}, \; \zeta(\nu \tilde{s}, \rho)\}, \; \tfrac{1}{1 \, + \, \nu \tilde{s}}\big\}$
or $\zeta_{2}(\epsilon_{0}, \nu\tilde{s}, \rho) = \frac{1}{1\,+\,\nu \tilde{s}}$.
The latter possibility,  
which 
comes with $\lambda^{*} = \tfrac{1}{1 \, + \, \nu \tilde{s}}$,
can be discarded as it gives an
upper bound on the minimum over $\zeta \in \big[\tfrac{1}{1\,+\,\nu \tilde{s}}, \; 1\big]$,
which is considered next:
\begin{align}
&
{\color{black}
\min_{\substack{\\
\zeta: \\
1/(1 \, + \, \nu \tilde{s}) \, \leq \, \zeta \, \leq\, 1
}}
}
\;\;
{\color{black}
\min_{\substack{\\
{\color{white}\zeta}
\lambda:
{\color{white}\zeta}
\\
{\color{white}/}
\zeta \, \leq \, \lambda \, \leq \, 1
{\color{white}/}
}}
}
\;
\,\,\,
\Big\{
\tfrac{\rho}{2}
\ln\!
\big(\tfrac{1}{2}\nu^{2}\tilde{s}^{2}\zeta^{2} + 1\big)
\, + \,
\ln\lambda
\, + \,
(1 - \nu\tilde{s}\zeta)/\lambda
\Big\}
\nonumber \\
\overset{d3}{=} \;
&
{\color{black}
\min_{\substack{\\
\zeta: \\
1/(1 \, + \, \nu \tilde{s}) \, \leq \, \zeta \, \leq\, 1
}}
}
\;\;
{\color{white}
\min_{\substack{\\
{\color{white}\zeta}
\lambda:
{\color{white}\zeta}
\\
{\color{white}/}
\zeta \, \leq\, \lambda \, \leq \, 1
{\color{white}/}
}}
}
\;
\,\,\,
\Big\{
\tfrac{\rho}{2}
\ln\!
\big(\tfrac{1}{2}\nu^{2}\tilde{s}^{2}\zeta^{2} + 1\big)
\, + \,
\ln\zeta
\, + \,
1/\zeta
\, - \,
\nu\tilde{s}
\Big\}
\label{eqOF} \\
\overset{d4}{=} \; &
\;\;\;\;\;\,\,\,\,
\tfrac{\rho}{2}
\ln\!
\big(\tfrac{1}{2}\nu^{2}\tilde{s}^{2}\zeta_{2}^{2}(\nu \tilde{s}, \rho) + 1\big)
\, + \,
\ln\zeta_{2}(\nu \tilde{s}, \rho)
\, + \,
1/\zeta_{2}(\nu \tilde{s}, \rho)
\, - \,
\nu\tilde{s},
\label{eqMinSecond} \\
&
\;\;\;\;\;\;\;\;\;\;\;\;\;\;\;\;\;\;\;\;\;\,\,\,\,
\zeta_{2}(\nu \tilde{s}, \rho) \; = \;
\Bigg\{
\begin{array} {l r}
\zeta(\nu \tilde{s}, \rho),
& \;\;\; 0 \, \leq \, \rho \, \leq \, \tfrac{2}{\nu\tilde{s}} + 4 + 3\nu\tilde{s},  \\
\tfrac{1}{1 \, + \, \nu \tilde{s}},
& \rho \, > \, \tfrac{2}{\nu\tilde{s}} + 4 + 3\nu\tilde{s},
\end{array}
\label{eqRho2}
\end{align}
where in ($d3$) the minimizer is $\lambda^{*} \!= \zeta$.
For ($d4$), observe that the objective function of (\ref{eqOF})
tends to $+\infty$ as $\zeta\rightarrow 0$,
while its derivative is positive for $\zeta > 1$:
\begin{equation} \label{eqDerZeta}
\rho \,\frac{\nu^{2}\tilde{s}^{2}\zeta}{2 + \nu^{2}\tilde{s}^{2}\zeta^{2}}
\, + \,
\frac{1}{\zeta} \, - \, \frac{1}{\zeta^{2}}
\; > \; 0 , \;\;\;\;\;\; \zeta > 1.
\end{equation}
We conclude that the objective function of (\ref{eqOF}) attains its minimum over $\zeta > 0$
when its derivative is zero, which occurs by the expression for the derivative in (\ref{eqDerZeta}) at 
$\zeta = \zeta(\nu \tilde{s}, \rho) \in (0, 1]$ solving the equation
\begin{equation} \nonumber
\rho \; = \; \bigg[\frac{1}{\zeta} - 1\bigg]\bigg[\frac{2}{\nu^{2}\tilde{s}^{2}\zeta^{2}} + 1\bigg],
\;\;\;\;\;\; \zeta \in (0, 1].
\end{equation}
It follows that the minimizer of (\ref{eqOF}) itself is 
$\zeta_{2}(\nu \tilde{s}, \rho) = \max\big\{\tfrac{1}{1 \, + \, \nu \tilde{s}}, \;\zeta(\nu \tilde{s}, \rho)\big\}$,
equal to (\ref{eqRho2}).

By (\ref{eqMinFirst}) and (\ref{eqMinSecond}) 
we conclude that in the case of $q = 1$ and $r = 2$
the maximum of (\ref{eqGenegalqr}) can be replaced with the supremum (\ref{eqTwoMins})
with (\ref{eqRho1}), (\ref{eqRho2}), 
$\lambda_{1}(\epsilon_{0}, \nu\tilde{s}, \rho) = 1 - \nu\tilde{s}\zeta_{1}(\epsilon_{0}, \nu\tilde{s}, \rho)$,
and $\lambda_{2}(\epsilon_{0}, \nu\tilde{s}, \rho) = \zeta_{2}(\nu \tilde{s}, \rho)$.
The resulting inequality holds for all $R > 0$,
except possibly for the boundary point
\begin{equation} \nonumber 
R \; = \; 
\tfrac{1}{2}\,{\log\mathstrut}_{\!b}
\big(
\tfrac{1}{2}\nu^{2}\tilde{s}^{2}
\epsilon^{2}_{0} + 1\big) 
\, + \, 
\tfrac{1}{2}\,{\log\mathstrut}_{\!b} \,\tfrac{\pi}{e}
\, + \,
\epsilon,
\end{equation}
where the supremum (\ref{eqTwoMins}) as a function of $R$ transitions between $+\infty$ and a finite value. 
With $\limsup_{\substack{\,n\,\rightarrow\,\infty}}$ on both sides of (\ref{eqGenegalqr}) we obtain that
the supremum (\ref{eqTwoMins}) can replace also the RHS of (\ref{eqLapSq}).
The RHS of (\ref{eqLapSq}) itself corresponds to the minimization in (\ref{eqMinoverZeta})
over $\zeta \in (0, 1]$ with $\epsilon = 0$.
By the same reasoning as at the end of the proof of Theorem~\ref{thmSUB21},
we conclude that (\ref{eqTwoMins}) converges to the RHS of (\ref{eqLapSq}) 
as $\epsilon \rightarrow 0$ and $\epsilon_{0}\rightarrow 0$.
\end{proof}

Examples of the explicit bounds of Theorems~\ref{thmSUB11},~\ref{thmSUB21},~\ref{thmSUB12}
as functions of the information rate $R$ are shown on Fig.~\ref{graphs12} and Fig.~\ref{graph3}.
The bounds reach zero at $R(0)$ equal to the Shannon upper bound in each case.

\begin{figure}
\centering
\begin{minipage}{.5\textwidth}
  \centering
  \includegraphics[width=1\textwidth]{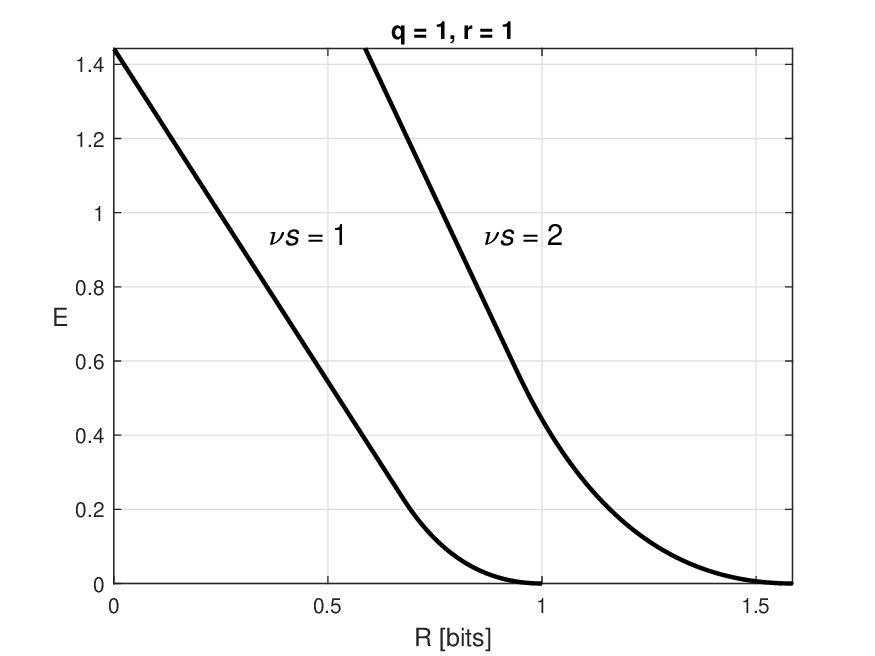}
\end{minipage}%
\begin{minipage}{.5\textwidth}
  \centering
  \includegraphics[width=1\textwidth]{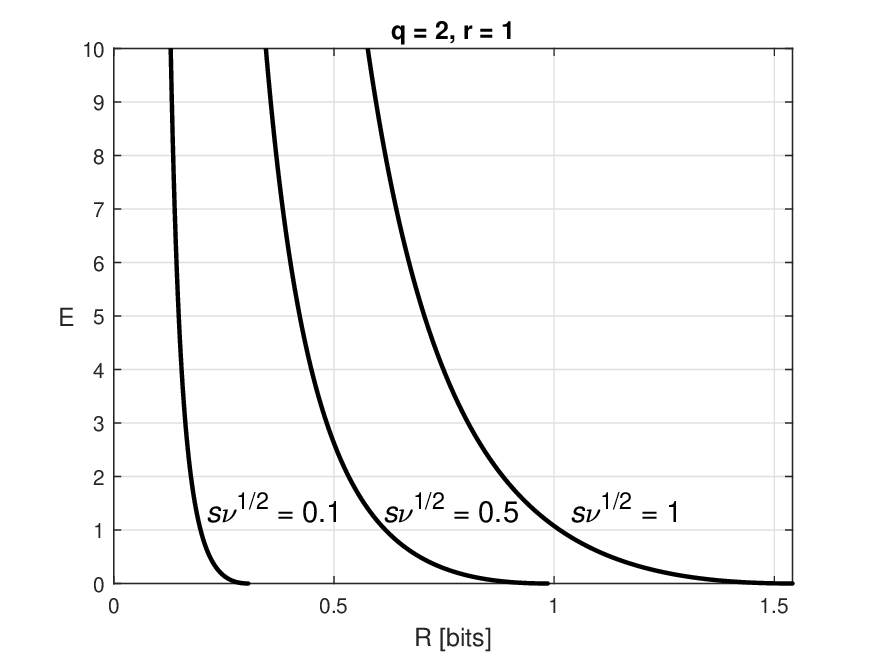}
\end{minipage}
\caption{Left: $E$ vs. $R$ [bits] 
for $q = 1$, $r = 1$ 
by (\ref{eqLaplace}) with $\nu s = 1$, $2$. 
Here $E(0) = \nu s/\ln 2$ and $R(0) = {\log\mathstrut}_{\!2}(1 + \nu s)$.\newline
Right: $E$ vs. $R$ [bits] 
for $q = 2$, $r = 1$  
by (\ref{eqGaussAbs})
with $s\sqrt{\nu} = 0.1$, $0.5$, $1$. Here $E\big({\log\mathstrut}_{\!2}\frac{2\sqrt{e}}{\pi}\!+\!\big) = +\infty$ and
$R(0) = {\log\mathstrut}_{\!2}(1 + s\sqrt{\nu\pi}) + {\log\mathstrut}_{\!2}\frac{2\sqrt{e}}{\pi}$.}
\label{graphs12}
\end{figure}

\begin{figure}[t]
\centering
\includegraphics[width=.50\textwidth]{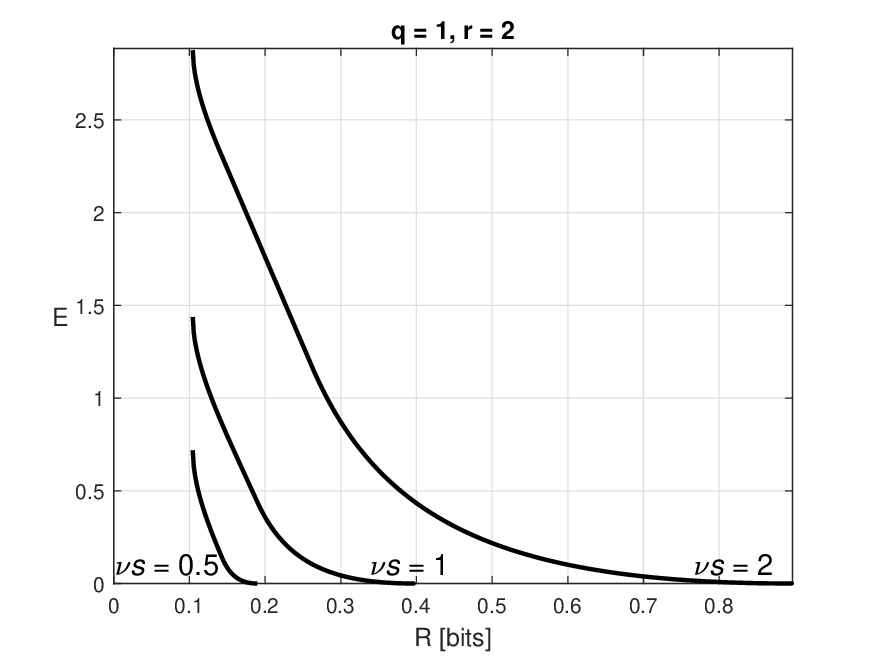}  
\caption{$E$ vs. $R$ [bits]  for $q = 1$, $r = 2$ by (\ref{eqLapSq}) 
with $\nu s = 0.5$, $1$, $2$. Here $E\big(\tfrac{1}{2}{\log\mathstrut}_{\!2}\tfrac{\pi}{e}\!+\big) = \frac{\nu s}{\ln 2}$ and 
$R(0) = \tfrac{1}{2}{\log\mathstrut}_{\!2}\big(1 + \tfrac{1}{2}\nu^{2}s^{2}\big) + \tfrac{1}{2}{\log\mathstrut}_{\!2}\tfrac{\pi}{e}$.}
\label{graph3}
\end{figure}

\section{Method of types}\label{MOT}


In this section we extend the method of types \cite{CoverThomas} to include the countable alphabets of uniformly spaced reals
(\ref{eqAlphabets})
by using 
generalized $p$th power constraints on the types, with a real $p \geq 1$. 
This is a generalization of \cite[Section~5]{TridenskiSomekhBaruh23},
where the $2$nd power is used. 
The method of types in the form of the results of this section is then used in the rest of the paper. 
It allows us to establish converse bounds in terms of types in Sections~\ref{ConvLemma} and~\ref{ErrExp} 
and is used in the Appendix B
dedicated to the proof of Lemma~\ref{LemQuant} of Section~\ref{PDFtypePDF}, connecting between PDFs and types.


\subsection*{\underline{Alphabet size}}



Consider all the types ${P\mathstrut}_{\!X} \in {\cal P}_{n}({\cal X}_{n})$
satisfying the generalized power constraint $\mathbb{E}_{{P\mathstrut}_{\!X}}\!\big[{|X|\mathstrut}^{p}\big] \leq c_{X}$ with a real $p\geq 1$.
Let ${\cal X}_{n}(c_{X}) \subseteq {\cal X}_{n}$ denote the subset of the 
alphabet
used by these types. In particular, every letter $x = i\Delta_{\alpha,\,n} \in {\cal X}_{n}(c_{X})$ must 
satisfy 
${P\mathstrut}_{\!X}(x){|x|\mathstrut}^{p} \leq c_{X}$ with ${P\mathstrut}_{\!X}(x)> 0$ for some ${P\mathstrut}_{\!X} \in {\cal P}_{n}({\cal X}_{n})$, while by the definition of a type we have ${P\mathstrut}_{\!X}(x)\geq 1/n$. This gives 
$|\, i\Delta_{\alpha,\,n}\,| \,\leq\, (n c_{X})^{1/p}$.
Then ${\cal X}_{n}(c_{X})$ is finite and by (\ref{eqDelta}) we obtain


\begin{lemma}[Alphabet size] \label{LemAlphabetSize}
{\em $\;\;|\, {\cal X}_{n}(c_{X}) \,| \;\; \leq \;\; 2c_{X}^{1/p} n^{1/p \, + \, \alpha} + 1 \;\; \leq \;\; (2c_{X}^{1/p} + 1) n^{1/p \, + \, \alpha}$.}
\end{lemma}


\subsection*{\underline{Size of a type class}}


For 
${P\mathstrut}_{\!XY} \in {\cal P}_{n}({\cal X}_{n}\times {\cal Y}_{n})$ let us define
\begin{align}
{\cal S}({P\mathstrut}_{\!XY}) \;\; & \triangleq \;\;
\big\{
(x, y) \in {\cal X}_{n}\times {\cal Y}_{n} : \;\; {P\mathstrut}_{\!XY}(x, y) > 0
\big\}.
\nonumber \\
{\cal S}({P\mathstrut}_{\!X}) \;\; & \triangleq \;\;
\big\{
x \in {\cal X}_{n}: \;\; {P\mathstrut}_{\!X}(x) > 0
\big\},
\;\;\;\;\;\;
{\cal S}({P\mathstrut}_{\!Y}) \;\; \triangleq \;\;
\big\{
y \in {\cal Y}_{n}: \;\; {P\mathstrut}_{\!Y}(y) > 0
\big\},
\nonumber
\end{align}


\begin{lemma}[Support of a type] \label{LemSupportXY}
{\em 
Let ${P\mathstrut}_{\!XY} \in {\cal P}_{n}({\cal X}_{n}\times {\cal Y}_{n})$ be a joint type,
such that $\mathbb{E}_{{P\mathstrut}_{\!X}}\!\big[{|X|\mathstrut}^{p}\big] \leq c_{X}$
and $\mathbb{E}_{{P\mathstrut}_{\!Y}}\!\big[{|Y|\mathstrut}^{p}\big] \leq c_{Y}$ with a real $p \geq 1$. 
Then}
\begin{align}
|\, {\cal S}({P\mathstrut}_{\!XY}) \,| \;\; & \leq \;\;
c_{1} \cdot n^{(2 \, + \, p(\alpha\,+\,\beta))/(2\,+\,p)},
\label{eqSXY} \\
|\, {\cal S}({P\mathstrut}_{\!X}) \,| \;\; & \leq \;\; 
c_{2}\cdot n^{(1 \, + \, p\alpha)/(1\,+\,p)},
\label{eqSX} \\
|\, {\cal S}({P\mathstrut}_{\!Y}) \,| \;\; & \leq \;\; 
c_{3}\cdot n^{(1 \, + \, p\beta)/(1\,+\,p)},
\label{eqSY}
\end{align}
{\em where $c_{1} \triangleq \Big[\tfrac{3\sqrt{\pi}}{2}\Big(2^{\max\!\big\{\!0, \, \frac{1}{2} - \frac{1}{p}\!\big\}}\cdot(c_{X} + c_{Y})^{1/p} \, + \, \tfrac{1}{\sqrt{6}}\Big) \Big]^{\frac{2p}{2\,+\,p}}$,
$c_{2} \triangleq (4c_{X}^{1/p} + 1)^{\frac{p}{1\,+\,p}}$, $c_{3} \triangleq (4c_{Y}^{1/p} + 1)^{\frac{p}{1\,+\,p}}$.}
\end{lemma}
\begin{proof}
Follows as a corollary from Lemma~\ref{LemSupportXk} of the Appendix A.
\end{proof}

\bigskip

\begin{lemma}[Size of a type class] \label{LemTypeSize}
{\em Let ${P\mathstrut}_{\!XY} \in {\cal P}_{n}({\cal X}_{n}\times {\cal Y}_{n})$ be a joint type,
such that $\mathbb{E}_{{P\mathstrut}_{\!X}}\!\big[{|X|\mathstrut}^{p}\big] \leq c_{X}$
and $\mathbb{E}_{{P\mathstrut}_{\!Y}}\!\big[{|Y|\mathstrut}^{p}\big] \leq c_{Y}$ with a real $p \geq 1$. 
Then}
\begin{align}
H({P\mathstrut}_{\!XY})
\; - \;
c_{1}\,
\frac{{\log\mathstrut}_{\!b}\,(n + 1)}{n^{\gamma p/(2\,+\,p)}}
\;\; & \leq \;\;
\frac{1}{n}\,{\log\mathstrut}_{\!b}\,|\, T({P\mathstrut}_{\!XY}) \,|
\;\; \leq \;\;
H({P\mathstrut}_{\!XY}),
\label{eqJointExp} \\
H({P\mathstrut}_{\!X})
\; - \;
c_{2}\,
\frac{{\log\mathstrut}_{\!b}\,(n + 1)}{n^{(1\,-\,\alpha)p/(1\,+\,p)}}
\;\; & \leq \;\;
\frac{1}{n}\,{\log\mathstrut}_{\!b}\,|\, T({P\mathstrut}_{\!X}) \,|
\;\;
\;\;
\leq \;\;
H({P\mathstrut}_{\!X}),
\label{eqMargExpX} \\
H({P\mathstrut}_{\!Y})
\; - \;
c_{3}\,
\frac{{\log\mathstrut}_{\!b}\,(n + 1)}{n^{(1\,-\,\beta)p/(1\,+\,p)}}
\;\; & \leq \;\;
\frac{1}{n}\,{\log\mathstrut}_{\!b}\,|\, T({P\mathstrut}_{\!Y}) \,|
\;\;
\;\,\,
\leq \;\;
H({P\mathstrut}_{\!Y}),
\label{eqMargExpY}
\end{align}
{\em where $c_{1}$, $c_{2}$, and $c_{3}$ are defined as in Lemma~\ref{LemSupportXY}.}
\end{lemma}

\bigskip

\begin{proof}
As in the proof of \cite[Lemma~5]{TridenskiSomekhBaruh23}, we
observe that the standard type-size bounds (see, e.g., \cite[Lemma~2.3]{CsiszarKorner}, \cite[Eq.~11.16]{CoverThomas}) can be rewritten as
\begin{equation} \label{eqMOTSup}
\frac{1}{(n + 1)^{|\, {\cal S}({P\mathstrut}_{\!XY}) \,|}}\,b^{\,n H({P\mathstrut}_{\!XY})}
\;\; \leq \;\;
|\, T({P\mathstrut}_{\!XY}) \,|
\;\; \leq \;\;
b^{\,n H({P\mathstrut}_{\!XY})}.
\end{equation}
Here
$|\, {\cal S}({P\mathstrut}_{\!XY}) \,|$ 
can be replaced
with its upper bound (\ref{eqSXY}) of Lemma~\ref{LemSupportXY}.
This gives (\ref{eqJointExp}).
The remaining bounds of (\ref{eqMargExpX}) and (\ref{eqMargExpY})
are obtained similarly using respectively (\ref{eqSX}) and (\ref{eqSY}) of Lemma~\ref{LemSupportXY}.
\end{proof}


Since it holds for any ${\bf y} \in T({P\mathstrut}_{\!Y})$ that
$|\, T({P\mathstrut}_{\!X|\,Y}\,|\, {\bf y})\,| \, = \, |\, T({P\mathstrut}_{\!XY}) \,| \, / \, |\, T({P\mathstrut}_{\!Y}) \,|$,
and similarly for ${\bf x} \in T({P\mathstrut}_{\!X})$,
as a corollary of the previous lemma we also obtain

\bigskip

\begin{lemma}[Size of a conditional type class] \label{LemCondTypeSize}
{\em Let ${P\mathstrut}_{\!XY} \in {\cal P}_{n}({\cal X}_{n}\times {\cal Y}_{n})$ be a joint type,
such that $\mathbb{E}_{{P\mathstrut}_{\!X}}\!\big[{|X|\mathstrut}^{p}\big] \leq c_{X}$
and $\mathbb{E}_{{P\mathstrut}_{\!Y}}\!\big[{|Y|\mathstrut}^{p}\big] \leq c_{Y}$ with a real $p \geq 1$. 
Then}
\begin{align}
H(X\,|\,Y)
\; - \;
c_{1}\,
\frac{{\log\mathstrut}_{\!b}\,(n + 1)}{n^{\gamma p/(2\,+\,p)}}
\;\; & \leq \;\;
\frac{1}{n}\,{\log\mathstrut}_{\!b}\,|\, T({P\mathstrut}_{\!X|\,Y}\,|\, {\bf y}) \,|
\;\; \leq \;\;
H(X\,|\,Y)
\; + \;
c_{3}\,
\frac{{\log\mathstrut}_{\!b}\,(n + 1)}{n^{(1\,-\,\beta)p/(1\,+\,p)}},
\label{eqXgivenY} \\
H(Y\,|\,X)
\; - \;
c_{1}\,
\frac{{\log\mathstrut}_{\!b}\,(n + 1)}{n^{\gamma p/(2\,+\,p)}}
\;\; & \leq \;\;
\frac{1}{n}\,{\log\mathstrut}_{\!b}\,|\, T({P\mathstrut}_{\!Y|X}\,|\, {\bf x}) \,|
\,
\;\; \leq \;\;
H(Y\,|\,X)
\; + \;
c_{2}\,
\frac{{\log\mathstrut}_{\!b}\,(n + 1)}{n^{(1\,-\,\alpha)p/(1\,+\,p)}},
\label{eqYgivenX}
\end{align}
{\em for ${\bf y} \in T({P\mathstrut}_{\!Y})$ and ${\bf x} \in T({P\mathstrut}_{\!X})$ respectively, where $c_{1}$, $c_{2}$, and $c_{3}$ are defined as in Lemma~\ref{LemSupportXY}.}
\end{lemma}



\subsection*{\underline{Number of types}}


Let ${\cal P}_{n}\big({\cal X}_{n}, \,c_{X}\big)$
be the set of all the types ${P\mathstrut}_{\!X} \in {\cal P}_{n}({\cal X}_{n})$
satisfying the generalized power constraint $\mathbb{E}_{{P\mathstrut}_{\!X}}\!\big[{|X|\mathstrut}^{p}\big] \leq c_{X}$
with a real parameter $p\geq 1$. 
Then its cardinality can be upper-bounded as follows:
\begin{equation} \label{eqNumberofTypes}
\big|\,  {\cal P}_{n}\big({\cal X}_{n}, \,c_{X}\big) \,\big|
\;\; \overset{a}{\leq} \;\;
\big|\,  {\cal P}_{n}\big({\cal X}_{n}(c_{X})\big) \,\big|
\;\; \overset{b}{\leq} \;\;
(n + 1)^{|\, {\cal X}_{n}(c_{X})\,|}
\;\; \overset{c}{\leq} \;\;
(n + 1)^{(2c_{X}^{1/p} \, + \, 1) n^{1/p \, + \, \alpha}},
\end{equation}
where ($a$) follows by the definition of ${\cal X}_{n}(c_{X})$
preceding Lemma~\ref{LemAlphabetSize}, ($b$) follows by \cite[Eq.~11.6]{CoverThomas},
and ($c$) follows by Lemma~\ref{LemAlphabetSize}.
This bound is sub-exponential in $n$ if $\alpha < 1 - 1/p$.
This can be further improved and made sub-exponential in $n$ for all $\alpha \in (0, 1)$
using (\ref{eqSX}) of Lemma~\ref{LemSupportXY}, as follows.


\begin{lemma}[Number of types] \label{LemNumofTypes}
\begin{equation} \label{eqImprovement}
\big|\,  {\cal P}_{n}\big({\cal X}_{n}, \,c_{X}\big) \,\big| \;\; \leq \;\;
\big((n + 1)c\big)^{\tilde{c}\, n^{(1\,+\, p\alpha)/(1\,+\,p)}},
\end{equation}
{\em where $c \,\triangleq\, (2c_{X}^{1/p} + 1)^{1/(1 \, + \, 1/p \, + \, \alpha)}$
and $\tilde{c} \,\triangleq\, (1 + 1/p + \alpha)(4 c_{X}^{1/p} + 1)^{p/(1\,+\,p)}$.}
\end{lemma}

\bigskip

\begin{proof}
Denoting $k \, \triangleq \, |\, {\cal X}_{n}(c_{X})\,|$ and
$\ell \, \triangleq \, \max_{\,{P\mathstrut}_{\!X} \, \in \;{\cal P}_{n}({\cal X}_{n}, \; c_{X})} \,|\, {\cal S}({P\mathstrut}_{\!X}) \,|$,
we can upper-bound as follows
\begin{equation} \nonumber 
\big|\,  {\cal P}_{n}\big({\cal X}_{n}, \,c_{X}\big) \,\big| \;\; \leq \;\;
\tbinom{k}{\ell} (n + 1)^{\ell} \;\; \leq \;\;
k^{\ell}(n + 1)^{\ell}.
\end{equation}
Substituting 
for
$k$ and $\ell$ their upper bounds of Lemma~\ref{LemAlphabetSize} (with $n + 1$) and Lemma~\ref{LemSupportXY},
we obtain (\ref{eqImprovement}).
\end{proof}


\section{Converse lemma}\label{ConvLemma}


In this section 
we establish a converse Lemma~\ref{LemConvLem}, which is then used for the error exponent in Section~\ref{ErrExp}.

In order to determine exponents in channel probabilities, 
it is convenient to 
take hold of
the exponent in the channel probability {\em density}.
Let ${\bf x} = (x_{1}, x_{2}, .\,.\,. \,, x_{n})\in \mathbb{R}{\mathstrut}^{n}$
be a vector of $n$ channel inputs
and let ${\bf x}^{*} = (x_{1}^{*}, x_{2}^{*}, .\,.\,. \,, x_{n}^{*}) \in {\cal X}_{n}^{n}$
be its quantized version, with components
\begin{equation} \label{eqQuantizer}
x_{k}^{*} \; = \; Q_{\alpha}(x_{k}) \; \triangleq \; \Delta_{\alpha,\,n}\cdot \lfloor x_{k}/\Delta_{\alpha,\,n} + 1/2\rfloor,
\;\;\;\;\;\; k = 1, .\,.\,.\,, n.
\end{equation}
Similarly,
let ${\bf y} = (y_{1}, y_{2}, .\,.\,. \,, y_{n}) \in \mathbb{R}{\mathstrut}^{n}$
be a vector of $n$ channel outputs
and let ${\bf y}^{*} = (y_{1}^{*}, y_{2}^{*}, .\,.\,. \,, y_{n}^{*}) \in {\cal Y}_{n}^{n}$
be its quantized version, with $y_{k}^{*} = Q_{\beta}(y_{k})$ for all $k = 1, .\,.\,.\,, n$.
Then we have the following

\bigskip

\begin{lemma}[PDF exponent] \label{LemPDFExponent}
{\em Let ${\bf x}\in \mathbb{R}{\mathstrut}^{n}$ and ${\bf y}\in \mathbb{R}{\mathstrut}^{n}$
be two channel input and output vectors,
with their respective quantized versions $({\bf x}^{*}, {\bf y}^{*}) \in T({P\mathstrut}_{\!XY})$,
such that $\mathbb{E}_{{P\mathstrut}_{\!XY}}\!\big[{|Y-X|\mathstrut}^{q}\big]\leq c_{XY}$ with $q$ of (\ref{eqChannel}). 
Then
}
\begin{align}
&
\big|
-\tfrac{1}{n}\,{\log\mathstrut}_{\!b}\, w({\bf y} \, | \, {\bf x})
\, + \, 
\tfrac{1}{n}\,{\log\mathstrut}_{\!b}\, w({\bf y}^{*} \, | \, {\bf x}^{*})
\,\big|
\;\; \leq \;\; 
\tfrac{q\nu}{2\ln b}(\Delta_{1,\,n} \, + \, \Delta_{2,\,n})
\Big[c_{XY}^{1/q} \, + \, 
\tfrac{1}{2}(\Delta_{1,\,n} \, + \, \Delta_{2,\,n}) 
\Big]^{q\,-\,1}.
\nonumber
\end{align}
\end{lemma}


\begin{proof}
The exponent can be equivalently rewritten as
\begin{equation} \label{eqRealW}
-\tfrac{1}{n}\,{\log\mathstrut}_{\!b}\, w({\bf y} \, | \, {\bf x})
\;\; \equiv \;\;
{\log\mathstrut}_{\!b}\,\Big(\tfrac{2\Gamma(1/q)}{q \nu^{1/q}}\Big) \; + \;
\tfrac{\nu}{\ln b}\cdot
\tfrac{1}{n}{\|{\bf y} - {\bf x}\|\mathstrut}^{q}_{q}.
\end{equation}
Defining $v_{k} \, \triangleq \, (y_{k} - x_{k}) - (y_{k}^{*} - x_{k}^{*})$,
with the triangle inequality we obtain:
\begin{align}
{\|{\bf y} - {\bf x}\|\mathstrut}_{q} \; & \leq \; {\|{\bf y}^{*} - {\bf x}^{*}\|\mathstrut}_{q} \, + \,
{\|{\bf v}\|\mathstrut}_{q}\,,
\nonumber \\
\tfrac{1}{n}{\|{\bf y} - {\bf x}\|\mathstrut}_{q}^{q} \; & \leq \; 
\Big[\big(\tfrac{1}{n}{\|{\bf y}^{*} - {\bf x}^{*}\|\mathstrut}_{q}^{q}\big)^{1/q} \, + \,
\big(\tfrac{1}{n}{\|{\bf v}\|\mathstrut}_{q}^{q}\big)^{1/q}
\Big]^{q},
\nonumber \\
\tfrac{1}{n}{\|{\bf y} - {\bf x}\|\mathstrut}_{q}^{q} \; & \leq \;
\tfrac{1}{n}{\|{\bf y}^{*} - {\bf x}^{*}\|\mathstrut}_{q}^{q} \, + \,
q\Big[\big(\tfrac{1}{n}{\|{\bf y}^{*} - {\bf x}^{*}\|\mathstrut}_{q}^{q}\big)^{1/q} \, + \,
\big(\tfrac{1}{n}{\|{\bf v}\|\mathstrut}_{q}^{q}\big)^{1/q}
\Big]^{q\,-\,1}
\big(\tfrac{1}{n}{\|{\bf v}\|\mathstrut}_{q}^{q}\big)^{1/q},
\label{eqDerivative} \\
\tfrac{1}{n}{\|{\bf y} - {\bf x}\|\mathstrut}_{q}^{q} \; & \leq \;
\tfrac{1}{n}{\|{\bf y}^{*} - {\bf x}^{*}\|\mathstrut}_{q}^{q} \, + \,
q\Big[c_{XY}^{1/q} \, + \,
\tfrac{1}{2}(\Delta_{1,\,n} + \Delta_{2,\,n})
\Big]^{q\,-\,1}\cdot
\tfrac{1}{2}(\Delta_{1,\,n} + \Delta_{2,\,n}),
\label{eqUpperB}
\end{align}
where in (\ref{eqDerivative}) we use the monotonically increasing derivative of the function $f(t) = |t|^{q}$, $t > 0$;
and then (\ref{eqUpperB}) follows because $|\, v_{k} \,| \,\leq\, \tfrac{1}{2}(\Delta_{\alpha,\,n} + \Delta_{\beta,\,n})$
and by the condition of the lemma.

For the bound from below:
\begin{align}
{\|{\bf y} - {\bf x}\|\mathstrut}_{q} \; & \geq \; {\|{\bf y}^{*} - {\bf x}^{*}\|\mathstrut}_{q} \, - \,
{\|{\bf v}\|\mathstrut}_{q}\,.
\label{eqTri}
\end{align}
If the RHS of (\ref{eqTri}) is negative:
\begin{align}
{\|{\bf y}^{*} - {\bf x}^{*}\|\mathstrut}_{q} \; & \leq \;
{\|{\bf v}\|\mathstrut}_{q}\,,
\nonumber \\
\tfrac{1}{n}{\|{\bf y}^{*} - {\bf x}^{*}\|\mathstrut}_{q}^{q} \; & \leq \;
\tfrac{1}{n}{\|{\bf v}\|\mathstrut}_{q}^{q}\,,
\nonumber \\
-\tfrac{1}{n}{\|{\bf v}\|\mathstrut}_{q}^{q} \, + \, \tfrac{1}{n}{\|{\bf y}^{*} - {\bf x}^{*}\|\mathstrut}_{q}^{q} \; & \leq \;
0,
\nonumber \\
-\big[\tfrac{1}{2}(\Delta_{\alpha,\,n} + \Delta_{\beta,\,n})\big]^{q} \, + \, \tfrac{1}{n}{\|{\bf y}^{*} - {\bf x}^{*}\|\mathstrut}_{q}^{q} \; & \leq \;
\tfrac{1}{n}{\|{\bf y} - {\bf x}\|\mathstrut}_{q}^{q}\,.
\label{eqNegative}
\end{align}
If the RHS of (\ref{eqTri}) is positive, we proceed as before:
\begin{align}
{\|{\bf y} - {\bf x}\|\mathstrut}_{q} \; & \geq \; {\|{\bf y}^{*} - {\bf x}^{*}\|\mathstrut}_{q} \, - \,
{\|{\bf v}\|\mathstrut}_{q} \; \geq \; 0,
\nonumber \\
\tfrac{1}{n}{\|{\bf y} - {\bf x}\|\mathstrut}_{q}^{q} \; & \geq \; 
\Big[\big(\tfrac{1}{n}{\|{\bf y}^{*} - {\bf x}^{*}\|\mathstrut}_{q}^{q}\big)^{1/q} \, - \,
\big(\tfrac{1}{n}{\|{\bf v}\|\mathstrut}_{q}^{q}\big)^{1/q}
\Big]^{q},
\nonumber \\
\tfrac{1}{n}{\|{\bf y} - {\bf x}\|\mathstrut}_{q}^{q} \; & \geq \;
\tfrac{1}{n}{\|{\bf y}^{*} - {\bf x}^{*}\|\mathstrut}_{q}^{q} \, - \,
q\Big[\big(\tfrac{1}{n}{\|{\bf y}^{*} - {\bf x}^{*}\|\mathstrut}_{q}^{q}\big)^{1/q}\Big]^{q\,-\,1}
\big(\tfrac{1}{n}{\|{\bf v}\|\mathstrut}_{q}^{q}\big)^{1/q},
\label{eqDerivative2} \\
\tfrac{1}{n}{\|{\bf y} - {\bf x}\|\mathstrut}_{q}^{q} \; & \geq \;
\tfrac{1}{n}{\|{\bf y}^{*} - {\bf x}^{*}\|\mathstrut}_{q}^{q} \, - \,
q \big[c_{XY}^{1/q}\big]^{q\,-\,1}\cdot
\tfrac{1}{2}(\Delta_{1,\,n} + \Delta_{2,\,n}),
\label{eqLowerB}
\end{align}
where in (\ref{eqDerivative2}) we use the monotonically increasing derivative of the function $f(t) = |t|^{q}$, $t > 0$;
and in (\ref{eqLowerB}) we use $|\, v_{k} \,| \,\leq\, \tfrac{1}{2}(\Delta_{\alpha,\,n} + \Delta_{\beta,\,n})$
and the condition of the lemma again.

The exponent with the quantized versions
$-\frac{1}{n}\,{\log\mathstrut}_{\!b}\, w({\bf y}^{*} \, | \, {\bf x}^{*})$, in turn,
can also be rewritten similarly to (\ref{eqRealW}). 
Then the result of the lemma follows from comparison of the minimum of the two lower bounds
(\ref{eqLowerB}) and (\ref{eqNegative}) with the upper bound (\ref{eqUpperB}).
\end{proof}

The following lemma will be used for the upper bound on the error exponent.

\bigskip

\begin{lemma}[Conditional probability of correct decoding] 
\label{LemConvLem}
{\em Let ${P\mathstrut}_{\!XY} \in {\cal P}_{n}({\cal X}_{n}\times {\cal Y}_{n})$ be a joint type,
such that $\mathbb{E}_{{P\mathstrut}_{\!X}}\!\big[{|X|\mathstrut}^{p}\big] \leq c_{X}$, $\mathbb{E}_{{P\mathstrut}_{\!Y}}\!\big[{|Y|\mathstrut}^{p}\big] \leq c_{Y}$, $p \geq 1$, and 
$\mathbb{E}_{{P\mathstrut}_{\!XY}}\!\big[{|Y-X|\mathstrut}^{q}\big]\leq c_{XY}$. 
Let ${\cal C}$ be a codebook, 
such that the quantized versions (\ref{eqQuantizer}) of 
its codewords ${\bf x}(m)$, $m = 1, 2, .\,.\,.\, , \, M(n, R)$,
are all of the 
type ${P\mathstrut}_{\!X}$: 
\begin{displaymath} 
{\bf x}^{*}(m) \; = \; Q_{\alpha}({\bf x}(m))
\; = \;
\big(Q_{\alpha}(x_{1}(m)), Q_{\alpha}(x_{2}(m)), .\,.\,.\, , \,Q_{\alpha}(x_{n}(m))\big)
\; \in \;
T({P\mathstrut}_{\!X}), \;\;\; \forall m.
\end{displaymath}
Let $J \sim \text{Unif}\,\big(\{1, 2, .\,.\,.\, , \, M\}\big)$ be a random variable,
independent of the channel noise, and let ${\bf x}(J) \rightarrow {\bf Y}$ be the random channel-input and channel-output vectors, respectively.
Let ${\bf Y}^{*} = Q_{\beta}({\bf Y})\in {\cal Y}_{n}^{n}$. Then}
\begin{displaymath}
\Pr \Big\{
g({\bf Y}) = J \; \big| \;
\big({\bf x}^{*}(J), \,{\bf Y}^{*}\big) \, \in \, T({P\mathstrut}_{\!XY})
\Big\}
\;\; \leq \;\;
b^{\,-n\big(\widetilde{R} \, - \, I({P\mathstrut}_{\!XY})  
\, + \, o(1)\big)},
\end{displaymath}
{\em where $\widetilde{R} = \frac{1}{n}\,{\log\mathstrut}_{\!b}\,M(n, R)$,
and $o(1)\rightarrow 0$, as $n\rightarrow \infty$, 
depending only on 
$\alpha$, $\beta$, $c_{X}$, $c_{Y}$, $c_{XY}$, $p$, $q$, $\nu$.}
\end{lemma}

\bigskip

The proof relies on Lemma~\ref{LemPDFExponent}
and on the method of types in the form of 
(\ref{eqMargExpX}) of Lemma~\ref{LemTypeSize}, and (\ref{eqXgivenY}) of Lemma~\ref{LemCondTypeSize}
of the current paper.
Otherwise, the proof is the same as for Lemma~10 in \cite{TridenskiSomekhBaruh23}.

In the next section we derive a converse bound on the error exponent in terms of types.


\section{Error exponent}\label{ErrExp}

The end result of this section is given by Lemma~\ref{LemAllCodebooks} and represents a converse bound on the error exponent by the method of types.

\begin{lemma}[Error exponent of 
mono-composition codebooks]
\label{LemConstComp}
{\em 
Let ${P\mathstrut}_{\!X} \in {\cal P}_{n}({\cal X}_{n})$ be a type,
such that $\mathbb{E}_{{P\mathstrut}_{\!X}}\!\big[{|X|\mathstrut}^{r}\big] \leq \tilde{c}_{X}$ with $r \geq 1$,
and
let ${\cal C}$ be a codebook, 
such that the quantized versions (\ref{eqQuantizer}) of 
its codewords ${\bf x}(m)$, $m = 1, 2, .\,.\,.\, , \, M(n, R)$,
are all of the 
type ${P\mathstrut}_{\!X}$, that is:
\begin{displaymath}
{\bf x}^{*}(m) \; = \; Q_{\alpha}({\bf x}(m))
\; = \;
\big(Q_{\alpha}(x_{1}(m)), Q_{\alpha}(x_{2}(m)), .\,.\,.\, , \,Q_{\alpha}(x_{n}(m))\big)
\; \in \;
T({P\mathstrut}_{\!X}), \;\;\; \forall m.
\end{displaymath}
Let $J \sim \text{Unif}\,\big(\{1, 2, .\,.\,.\, , \, M\}\big)$ be a random variable,
independent of the channel noise, and let ${\bf x}(J) \rightarrow {\bf Y}$ be the random channel-input and channel-output vectors, respectively.
Then for any parameter $c_{XY}$
}
\begin{equation} \label{eqCCBound}
- \frac{1}{n}\,{\log\mathstrut}_{\!b}\Pr \big\{
g({\bf Y}) \neq J \big\}
\;\; \leq \;\;
\min_{\substack{\\{P\mathstrut}_{\!Y|X}:\\
{P\mathstrut}_{\!XY}\,\in \, {\cal P}_{n}({\cal X}_{n}\,\times\, {\cal Y}_{n}),
\\
\mathbb{E}[{|Y-X|\mathstrut}^{q}] \; \leq \; c_{XY},
\\ I({P\mathstrut}_{\!X}, \, {P\mathstrut}_{\!Y|X}) \; \leq \; \widetilde{R} \, - \, o(1)
}}
\Big\{
D\big({P\mathstrut}_{\!Y|X}\,\|\, {W\mathstrut}_{\!n} \,|\,  {P\mathstrut}_{\!X}\big)
\Big\} \; + \; o(1),
\end{equation}
{\em where  
$\widetilde{R} = \frac{1}{n}\,{\log\mathstrut}_{\!b}\,M(n, R)$,
and
$o(1)\rightarrow 0$, as $n\rightarrow \infty$, 
depending only on 
$\alpha$, $\beta$, $\tilde{c}_{X}$, $c_{XY}$, $r$, $q$, $\nu$.}
\end{lemma}

\bigskip

\begin{proof}
Let $p = \min\{r, q\}\geq 1$.
For a joint type ${P\mathstrut}_{\!XY} \in {\cal P}_{n}({\cal X}_{n}\,\times\, {\cal Y}_{n})$ with the marginal type ${P\mathstrut}_{\!X}$,
and such that
$\mathbb{E}_{{P\mathstrut}_{\!XY}}\!\big[{|Y-X|\mathstrut}^{q}\big]\leq c_{XY}$, through 
various applications of Jensen's inequality we also obtain:
\begin{align}
\mathbb{E}_{{P\mathstrut}_{\!X}}\big[{|X|\mathstrut}^{p}\big]
\;\;
& \leq
\;\;
\big(\mathbb{E}_{{P\mathstrut}_{\!X}}\big[{|X|\mathstrut}^{r}\big]\big)^{p/r}
\;\;
\leq
\;\;
\tilde{c}_{X}^{\,p/r}
\;\; \triangleq \;\;
c_{X},
\nonumber \\
\mathbb{E}_{{P\mathstrut}_{\!XY}}\!\big[{|Y-X|\mathstrut}^{p}\big]
\;\;
& \leq
\;\;
\big(\mathbb{E}_{{P\mathstrut}_{\!XY}}\big[{|Y-X|\mathstrut}^{q}\big]\big)^{p/q}
\;\;
\leq
\;\;
c_{XY}^{\,p/q},
\nonumber \\
\mathbb{E}_{{P\mathstrut}_{\!Y}}\!\big[{|Y|\mathstrut}^{p}\big]
\;\; 
& \leq
\;\;
{2\mathstrut}^{p}
\mathbb{E}_{{P\mathstrut}_{\!XY}}\!\big[\tfrac{1}{2}
{|X|\mathstrut}^{p}
\,+\,
\tfrac{1}{2}{|Y-X|\mathstrut}^{p}\big]
\;\; \leq \;\;
{2\mathstrut}^{p\,-\,1}
\big(\tilde{c}_{X}^{\,p/r} + c_{XY}^{\,p/q}\big)
\;\; \triangleq \;\;
c_{Y}.
\label{eqBoundp}
\end{align}
Then with $\,{\bf Y}^{*} = Q_{\beta}({\bf Y})\in {\cal Y}_{n}^{n}\,$ for any $1 \leq j \leq M$ we obtain
\begin{align}
\Pr \Big\{
\big({\bf x}^{*}(J), \,{\bf Y}^{*}\big) \, \in \, T({P\mathstrut}_{\!XY}) \; \big| \;
J = j
\Big\}
\;\; & \overset{a}{\geq} \;\;
\big|\, T\big({P\mathstrut}_{\!Y|X}\,|\, {\bf x}^{*}(j)\big) \,\big| \cdot \Delta_{\beta,\,n}^{n}\cdot
b^{\,n\big(\mathbb{E}_{{P\mathstrut}_{\!XY}}\![\,{\log\mathstrut}_{\!b}\, w(Y \, | \, X)\,] \; + \; o(1)\big)}
\nonumber \\
& \overset{b}{\geq} \;\;
b^{\,-n\big(D({P\mathstrut}_{\!Y|X}\,\|\, {W\mathstrut}_{\!n} \,|\,  {P\mathstrut}_{\!X}) \; + \; o(1)\big)},
\;\;\;\;\;\; \forall j,
\label{eqConTypeExp}
\end{align}
where ($a$) follows by Lemma~\ref{LemPDFExponent}, and ($b$) follows by (\ref{eqYgivenX}) of Lemma~\ref{LemCondTypeSize} and (\ref{eqChanApprox}).
This gives
\begin{align}
\Pr \Big\{
\big({\bf x}^{*}(J), \,{\bf Y}^{*}\big) \, \in \, T({P\mathstrut}_{\!XY})
\Big\}
\;\; & \geq \;\;
b^{\,-n\big(D({P\mathstrut}_{\!Y|X}\,\|\, {W\mathstrut}_{\!n} \,|\,  {P\mathstrut}_{\!X}) \; + \; o(1)\big)}.
\label{eqPrior}
\end{align}
Now we are ready to apply Lemma~\ref{LemConvLem}:
\begin{align}
\Pr \big\{
g({\bf Y}) \neq J \big\}
\;\; & \geq \;\;
\Pr \Big\{
\big({\bf x}^{*}(J), \,{\bf Y}^{*}\big) \, \in \, T({P\mathstrut}_{\!XY})
\Big\}\cdot
\Pr \Big\{
g({\bf Y}) \neq J \; \big| \;
\big({\bf x}^{*}(J), \,{\bf Y}^{*}\big) \, \in \, T({P\mathstrut}_{\!XY})
\Big\}
\nonumber \\
& \overset{a}{\geq} \;\;
b^{\,-n\big(D({P\mathstrut}_{\!Y|X}\,\|\, {W\mathstrut}_{\!n} \,|\,  {P\mathstrut}_{\!X}) \; + \; o(1)\big)}
\cdot
\Big[
1 \, - \,
b^{\,-n\big(\widetilde{R} \; - \; I({P\mathstrut}_{\!X}, \, {P\mathstrut}_{\!Y|X})  
\; + \; o(1)\big)}
\Big]
\nonumber \\
& \overset{b}{\geq} \;\;
b^{\,-n\big(D({P\mathstrut}_{\!Y|X}\,\|\, {W\mathstrut}_{\!n} \,|\,  {P\mathstrut}_{\!X}) \; + \; o(1)\big)}
\cdot 1/2,
\nonumber
\end{align}
where ($a$) follows by (\ref{eqPrior}) and Lemma~\ref{LemConvLem},
and ($b$) holds for $I({P\mathstrut}_{\!X}, {P\mathstrut}_{\!Y|X}) \,\leq \, \widetilde{R} - {\log\mathstrut}_{\!b}(2)/n + o(1)$.
\end{proof}

\bigskip

\begin{lemma}[Type constraint] \label{LemTypeConstraint}
{\em
For any $\epsilon > 0$ there exists $n_{0} = n_{0}(\alpha, s, r, \epsilon) \in \mathbb{N}$,
such that for any $n > n_{0}\,$ for any codeword ${\bf x} \in \mathbb{R}{\mathstrut}^{n}$,
satisfying the generalized power constraint (\ref{eqPowerConstraint}), its quantized version (\ref{eqQuantizer})
satisfies (\ref{eqPowerConstraint}) within $\epsilon$, that is
with $s^{r}$ replaced by $s^{r} \!+ \epsilon$.
}
\end{lemma}
The proof is similar to the proof of (\ref{eqUpperB}), with the quantized version put on the LHS of that inequality.

\bigskip

\begin{lemma}[Error exponent] \label{LemAllCodebooks}
{\em
Let $J \sim \text{Unif}\,\big(\{1, 2, .\,.\,.\, , \, M\}\big)$ be a random variable,
independent of the channel noise, and let ${\bf x}(J) \rightarrow {\bf Y}$ be the random channel-input and channel-output vectors. 
Then for any 
$c_{XY}$ and $\epsilon > 0$ there exists
$n_{0} = n_{0}(\alpha, \,\beta, \,s, \,c_{XY}, \, r, \,q, \,\nu, 
\,\epsilon) \in \mathbb{N}$,
such that for any $n > n_{0}$
}
\begin{equation} \label{eqBoundTypes}
- \frac{1}{n}\,{\log\mathstrut}_{\!b}\Pr \big\{
g({\bf Y}) \neq J \big\}
\;\; \leq \;\;
\max_{\substack{\\{P\mathstrut}_{\!X{\color{white}|}}\!\!:\\
{P\mathstrut}_{\!X}\,\in \, {\cal P}_{n}({\cal X}_{n}),
\\
\mathbb{E}[{|X|\mathstrut}^{r}] \; \leq \; s^{r} + \, \epsilon}}
\;\;
\min_{\substack{\\{P\mathstrut}_{\!Y|X}:\\
{P\mathstrut}_{\!XY}\,\in \, {\cal P}_{n}({\cal X}_{n}\,\times\, {\cal Y}_{n}),
\\
\mathbb{E}[{|Y-X|\mathstrut}^{q}] \; \leq \; c_{XY},
\\ I({P\mathstrut}_{\!X}, \, {P\mathstrut}_{\!Y|X}) \; \leq \; R \, - \, \epsilon
}}
\Big\{
D\big({P\mathstrut}_{\!Y|X}\,\|\, {W\mathstrut}_{\!n} \,|\,  {P\mathstrut}_{\!X}\big)
\Big\} \; + \; o(1),
\end{equation}
{\em where $o(1)\rightarrow 0$, as $n\rightarrow \infty$, 
depending only on the parameters
$\alpha$, $\beta$, $s^{r} \!+ \epsilon$, $c_{XY}$, $r$, $q$, $\nu$.}
\end{lemma}

\bigskip

\begin{proof}
For a type ${P\mathstrut}_{\!X} \in {\cal P}_{n}({\cal X}_{n})$ let us define $\,M({P\mathstrut}_{\!X}) \, \triangleq \,
\big|\,
\big\{
j : \,
{\bf x}^{*}(j) \in T({P\mathstrut}_{\!X})
\big\}\,\big|$.
Then for any $n$ greater than $n_{0}$ of Lemma~\ref{LemTypeConstraint}
there exists at least one type ${P\mathstrut}_{\!X}$ such that
\begin{equation} \label{eqFrequentType}
M({P\mathstrut}_{\!X})
\;\; \geq \;\;
\frac{M}{\big|\,  {\cal P}_{n}\big({\cal X}_{n}, \,s^{r} \!+ \epsilon\big) \,\big|}
\;\; \geq \;\;
\frac{M}{\big((n + 1)c\big)^{\tilde{c}\, n^{(1\,+\, r\alpha)/(1\,+\,r)}}},
\end{equation}
where the second inequality follows by Lemma~\ref{LemNumofTypes} applied with $c_{X} = s^{r} \!+ \epsilon$ and $p = r$.
Then we can use such a type for a bound:
\begin{align}
& \Pr \big\{
g({\bf Y}) \neq J \big\} \;\; = \;\; \frac{1}{M}\sum_{j \, = \, 1}^{M} \Pr \big\{{\bf Y} \notin {\cal D}{\mathstrut}_{J} \; | \; J = j\big\}
\;\; \geq \;\;
\frac{1}{M}\sum_{\substack{1\,\leq\,j\,\leq\,M \, : \\
{\bf x}^{*}(j) \, \in \, T({P\mathstrut}_{\!X})
}} \Pr \big\{{\bf Y} \notin {\cal D}{\mathstrut}_{J} \; | \; J = j\big\}
\nonumber \\
&
\;\;\;\;\;\;\;\;\;\;\;\;\;\;\;\;\;\;\;\;\;\;\;\;\;\;\;\;\;\;\;\;\;\;\;\;\;\;\;\;\;\;\;\;\;\;\;
\;\;\;\;\;\;\;\;\;\;\;\,\,\,\,
\overset{a}{\geq} \;\;
b^{\,n \,\cdot \, o(1)}\,
\frac{1}{M({P\mathstrut}_{\!X})}\sum_{\substack{1\,\leq\,j\,\leq\,M \, : \\
{\bf x}^{*}(j) \, \in \, T({P\mathstrut}_{\!X})
}} \Pr \big\{{\bf Y} \notin {\cal D}{\mathstrut}_{J} \; | \; J = j\big\}
\nonumber \\
&
\;\;\;\;\;\;\;\;\;\;\;\;\;\;\;\;\;\;\;\;\;\;\;\;\;\;\;\;\;\;\;\;\;\;\;\;\;\;\;\;\;\;\;\;\;\;\;
\;\;\;\;\;\;\;\;\;\;\;\,\,\,\,
\overset{b}{=} \;\;
b^{\,n \,\cdot \, o(1)}\,
\Pr \big\{\,
\widetilde{g}(\widetilde{\bf Y}) \,\neq\, \widetilde{J} \,\big\},
\label{eqSmallerCodebook}
\end{align}
where ($a$) follows by (\ref{eqFrequentType}),
and ($b$) holds for the random variable $\widetilde{J} \sim \text{Unif}\,\big(\{1, 2, .\,.\,.\, , \, M({P\mathstrut}_{\!X})\}\big)$,
independent of the channel noise,
and for the 
channel input/output
random vectors
${\bf x}\big(m_{\widetilde{J}}\big) \rightarrow \widetilde{\bf Y}$
with the decoder
\begin{displaymath} 
\widetilde{g}({\bf y})
\;\; \triangleq \;\;
\left\{
\begin{array}{r l}
0, & \;\;\; {\bf y} \in \bigcap_{\,j \, = \, 1}^{\,M({P\mathstrut}_{\!X})} {\cal D}{\mathstrut}_{m_{j}}^{c}, \\
j, & \;\;\; {\bf y} \in {\cal D}{\mathstrut}_{m_{j}}, \;\;\; 
j \in \{1, 2, .\,.\,.\, , \, M({P\mathstrut}_{\!X})\},
\end{array}
\right.
\end{displaymath}
where $m_{1} < .\,.\,. < m_{M({P\mathstrut}_{\!X})}$ are the indices of the codewords in ${\cal C}$ with
their quantized versions in $T({P\mathstrut}_{\!X})$.

It follows now from (\ref{eqSmallerCodebook}) that the LHS of (\ref{eqBoundTypes}) can be upper-bounded by
(\ref{eqCCBound}) of Lemma~\ref{LemConstComp}
with $\widetilde{R} = \frac{1}{n}\,{\log\mathstrut}_{\!b}\,M({P\mathstrut}_{\!X})$.
Substituting then (\ref{eqFrequentType}) in place of $M({P\mathstrut}_{\!X})$
we obtain a stricter condition under the minimum of (\ref{eqCCBound}), leading to 
an upper bound with
a condition $I({P\mathstrut}_{\!X}, {P\mathstrut}_{\!Y|X}) \,\leq \, R - o(1)$
and to (\ref{eqBoundTypes}).
\end{proof}



\section{PDF to type}\label{PDFtypePDF}


Lemma~\ref{LemPDFtoT} of this section relates between minimums over types and over PDFs.
The next Lemma~\ref{LemQuant}, which has a long proof, is required in the proof of 
Lemma~\ref{LemPDFtoT}, used in turn for Theorem~\ref{thmErrorExp}.

\bigskip

\begin{lemma}[Quantization of PDF] \label{LemQuant}
{\em Let ${\cal X}_{n}$ be an alphabet defined as in (\ref{eqAlphabets}), (\ref{eqDelta})
with $\alpha \in \big(0, \tfrac{r}{1\,+\,2r}\big)$, $r\geq 1$.
Let ${P\mathstrut}_{\!X} \in {\cal P}_{n}({\cal X}_{n})$ be a type and ${p\mathstrut}_{Y|X}(\,\cdot\,|\,x) \in {\cal L}$, $\forall x \in {\cal X}_{n}\,$,
be a collection of functions from (\ref{eqLipschitz}), such that $\mathbb{E}_{{P\mathstrut}_{\!X}}\!\big[{|X|\mathstrut}^{r}\big] \leq c_{X}$,
 $\mathbb{E}_{{P\mathstrut}_{\!X}{p\mathstrut}_{Y|X}}\!\big[{|Y|\mathstrut}^{p}\big]\leq c_{Y}$, $1\leq p \leq r$, and 
$\mathbb{E}_{{P\mathstrut}_{\!X}{p\mathstrut}_{Y|X}}\!\big[{|Y-X|\mathstrut}^{q}\big]\leq c_{XY}$.
Then
for any alphabet ${\cal Y}_{n}$ defined as in (\ref{eqAlphabets}), (\ref{eqDelta})
with $\alpha < \beta < \tfrac{r}{1\,+\,r}(1-\alpha)$,
there exists a joint type ${P\mathstrut}_{\!XY} \in {\cal P}_{n}({\cal X}_{n}\times {\cal Y}_{n})$ with the marginal type ${P\mathstrut}_{\!X}$,
such that}
\begin{align}
\sum_{x\,\in\,{\cal X}_{n}}{P\mathstrut}_{\!X}(x)\int_{\mathbb{R}}{p\mathstrut}_{Y|X}(y\,|\,x)
{\log\mathstrut}_{\!b}\,{p\mathstrut}_{Y|X}(y\,|\,x)dy
\;\; & \geq \;\;
\sum_{\substack{x\,\in\,{\cal X}_{n}\\
y\,\in\,{\cal Y}_{n}}}
{P\mathstrut}_{\!XY}(x, y){\log\mathstrut}_{\!b}\,\frac{{P\mathstrut}_{\!Y|X}(y\,|\,x)}{\Delta_{\beta,\,n}}
\, + \, o(1),
\label{eqCondEntropy} \\
\sum_{x\,\in\,{\cal X}_{n}}{P\mathstrut}_{\!X}(x)\int_{\mathbb{R}}{p\mathstrut}_{Y|X}(y\,|\,x)\,
{|\,y - x\,|\mathstrut}^{q}dy
\;\; & \geq \;\;
\sum_{\substack{x\,\in\,{\cal X}_{n}\\
y\,\in\,{\cal Y}_{n}}}
{P\mathstrut}_{\!XY}(x, y)\,{|\,y - x\,|\mathstrut}^{q} \, + \, o(1),
\label{eqLogGaussian} \\
\int_{\mathbb{R}}{p\mathstrut}_{Y}(y)
{\log\mathstrut}_{\!b}\,{p\mathstrut}_{Y}(y)dy
\;\; & \leq \;\;
\sum_{y\,\in\,{\cal Y}_{n}}
{P\mathstrut}_{\!Y}(y){\log\mathstrut}_{\!b}\,\frac{{P\mathstrut}_{\!Y}(y)}{\Delta_{\beta,\,n}}
\, + \, o(1),
\label{eqMargEntropy}
\end{align}
{\em where ${p\mathstrut}_{Y}(y) = \sum_{\,x\,\in\,{\cal X}_{n}}{P\mathstrut}_{\!X}(x){p\mathstrut}_{Y|X}(y\,|\,x)$, $\forall y \in \mathbb{R}$,
and $o(1)\rightarrow 0$, as $n\rightarrow \infty$, depending only on the parameters $\alpha$, $\beta$, $c_{X}$, $c_{Y}$, $c_{XY}$, $r$, 
$p$, $q$,
and $K$ of (\ref{eqLipschitz}).}
\end{lemma}
The proof is given in the Appendix B.

\bigskip

\begin{lemma}[PDF to type]\label{LemPDFtoT}
{\em Let ${\cal X}_{n}$ and ${\cal Y}_{n}$ be two alphabets defined as in (\ref{eqAlphabets}), (\ref{eqDelta})
with $\alpha \in \big(0, \tfrac{r}{1\,+\,2r}\big)$ 
and $\alpha < \beta < \tfrac{r}{1\,+\,r}(1-\alpha)$, $r\geq 1$.
Then for any $c_{X}$, $c_{XY}$, $\epsilon > 0$ there exists 
$n_{0} = n_{0}(\alpha, \,\beta, \,c_{X}, \,c_{XY}, \,r, \,q, \,K, \,\epsilon) \in \mathbb{N}$,
such that for any $n > n_{0}$ and
for any type ${P\mathstrut}_{\!X} \in {\cal P}_{n}({\cal X}_{n})$ with $\mathbb{E}_{{P\mathstrut}_{\!X}}\!
\big[{|X|\mathstrut}^{r}\big] \leq c_{X}$:
}
\begin{equation} \label{eqTransitionB}
\inf_{\substack{
{p\mathstrut}_{Y|X}:
\\
{p\mathstrut}_{Y|X}(\, \cdot \, \,|\, x) \, \in \, {\cal L}, \; \forall x,
\\
\mathbb{E}_{{P\mathstrut}_{\!X}{p\mathstrut}_{Y|X}}[{|Y-X|\mathstrut}^{q}] \; \leq \; c_{XY},
\\
I({P\mathstrut}_{\!X}, \,\, {p\mathstrut}_{Y|X}) \; \leq \; R \, - \, 2\epsilon
}}
\!\!\!\!
\Big\{
D\big(\,{p\mathstrut}_{Y|X}\,\|\, \, w \, \,|\, {P\mathstrut}_{\!X}\big)
\Big\}
\; \;
\geq
\min_{\substack{\\{P\mathstrut}_{\!Y|X}:\\
{P\mathstrut}_{\!XY}\,\in \, {\cal P}_{n}({\cal X}_{n}\,\times\, {\cal Y}_{n}),
\\
\mathbb{E}_{{P\mathstrut}_{\!XY}}[{|Y-X|\mathstrut}^{q}] \; \leq \; c_{XY} + \, \epsilon,
\\ I({P\mathstrut}_{\!X}, \, {P\mathstrut}_{\!Y|X}) \; \leq \; R \, - \, \epsilon
}}
\!\!\!\!
\Big\{
D\big({P\mathstrut}_{\!Y|X}\,\|\, {W\mathstrut}_{\!n} \,|\,  {P\mathstrut}_{\!X}\big)
\Big\} \; + \; o(1),
\end{equation}
{\em where $o(1)\rightarrow 0$, as $n\rightarrow \infty$, 
and depends only on the parameters $\alpha$, $\beta$, $c_{X}$, $c_{XY}$, $r$, $q$, $\nu$.}
\end{lemma}

\bigskip

\begin{proof}
Consider 
a type ${P\mathstrut}_{\!X} \in {\cal P}_{n}({\cal X}_{n})$ with a collection of ${p\mathstrut}_{Y|X} \in {\cal L}$ such that $\mathbb{E}_{{P\mathstrut}_{\!X}}\!
\big[{|X|\mathstrut}^{r}\big] \leq c_{X}$ 
and $\mathbb{E}_{{P\mathstrut}_{\!X}{p\mathstrut}_{Y|X}}\!\big[{|Y-X|\mathstrut}^{q}\big]\leq c_{XY}$.  
Using $p = \min\{r, q\}\geq 1$ as in (\ref{eqBoundp}) 
we obtain also $\mathbb{E}_{{P\mathstrut}_{\!X}{p\mathstrut}_{Y|X}}\!\big[{|Y|\mathstrut}^{p}\big] \leq
{2\mathstrut}^{p\,-\,1}\!\big(c_{X}^{\,p/r} + c_{XY}^{\,p/q}\big) \triangleq c_{Y}$. 
Then by Lemma~\ref{LemQuant} there exists a joint type 
${P\mathstrut}_{\!XY} \in {\cal P}_{n}({\cal X}_{n}\times {\cal Y}_{n})$
with the marginal type ${P\mathstrut}_{\!X}$, such that simultaneously
the three inequalities (\ref{eqCondEntropy})-(\ref{eqMargEntropy}) are satisfied. 
It also follows by (\ref{eqChanApprox}) and (\ref{eqLogGaussian}) that
\begin{equation} \label{eqConseq4}
-\mathbb{E}_{{P\mathstrut}_{\!X}{p\mathstrut}_{Y|X}}\!\big[{\log\mathstrut}_{\!b}\,w(Y\,|\,X)\big]
\;\; \geq \;\;
-\mathbb{E}_{{P\mathstrut}_{\!XY}}\!\big[{\log\mathstrut}_{\!b}\,{W\mathstrut}_{\!n}(Y\,|\,X)\big]
 \, + \,
{\log\mathstrut}_{\!b}\,\Delta_{\beta,\,n}
\, + \, \nu\cdot o(1).
\end{equation}
Then the sum of (\ref{eqCondEntropy}) and (\ref{eqConseq4}) gives
\begin{equation} \label{eqObjective}
D\big(\,{p\mathstrut}_{Y|X}\,\|\, \, w \, \,|\, {P\mathstrut}_{\!X}\big)
\;\; \geq \;\;
D\big({P\mathstrut}_{\!Y|X}\,\|\, {W\mathstrut}_{\!n} \,|\,  {P\mathstrut}_{\!X}\big)
\, + \, o(1),
\end{equation}
while the difference of
(\ref{eqCondEntropy})
and
(\ref{eqMargEntropy}) gives
\begin{equation} \label{eqMutualIneq}
I\big({P\mathstrut}_{\!X}, {P\mathstrut}_{\!Y|X}\big)
\;\; \leq \;\;
I\big({P\mathstrut}_{\!X}, \, {p\mathstrut}_{Y|X}\big)
\, + \, o(1).
\end{equation}
Note that all $o(1)$ in the above relations are independent 
of the joint type ${P\mathstrut}_{\!XY}$
and the functions
${p\mathstrut}_{Y|X}$. Therefore
by (\ref{eqObjective}), (\ref{eqMutualIneq}), and (\ref{eqLogGaussian})
we conclude, that
given any $\epsilon > 0$ for $n$ sufficiently large 
for every type ${P\mathstrut}_{\!X}$ with the prerequisites of this lemma and every collection of ${p\mathstrut}_{Y|X}$
which satisfy the conditions under the infimum on the LHS of (\ref{eqTransitionB})
there exists a joint type ${P\mathstrut}_{\!XY}$
such that simultaneously
\begin{align}
I\big({P\mathstrut}_{\!X}, {P\mathstrut}_{\!Y|X}\big)
\;\; & \leq \;\;
I\big({P\mathstrut}_{\!X}, \, {p\mathstrut}_{Y|X}\big)
\, + \, \epsilon,
\nonumber \\
\mathbb{E}_{{P\mathstrut}_{\!XY}}\!\big[{|Y - X|\mathstrut}^{q}\big]
\;\; & \leq \;\;
\mathbb{E}_{{P\mathstrut}_{\!X}{p\mathstrut}_{Y|X}}\!\big[{|Y - X|\mathstrut}^{q}\big]
\, + \, \epsilon,
\nonumber
\end{align}
and (\ref{eqObjective}) 
holds with a uniform $o(1)$, i.e., independent of ${P\mathstrut}_{\!XY}$ and ${p\mathstrut}_{Y|X}$.
It follows that such
${P\mathstrut}_{\!XY}$
satisfies also the conditions under the {\em minimum} on the RHS of (\ref{eqTransitionB})
and results in the objective function of (\ref{eqTransitionB}) 
satisfying (\ref{eqObjective}) with the uniform $o(1)$.
Then the minimum itself, which can only possibly be taken over 
a greater variety of ${P\mathstrut}_{\!XY}$,
satisfies the inequality (\ref{eqTransitionB}).
\end{proof}

We use Lemma~\ref{LemPDFtoT} in Section~\ref{Main}
in the derivation of Theorem~\ref{thmErrorExp}.





\section*{Appendix A}\label{AppendixA}

For generality,
let us define $k \geq 1$ countable alphabets and consider their product alphabet:
\begin{align}
{\cal X}_{n}^{(j)} \; & \triangleq \;
\textstyle\bigcup_{\,i \, \in \, \mathbb{Z}}\big\{i \Delta_{j,\,n}\big\}, \;\;\;
\Delta_{j,\,n} \triangleq n^{-\alpha_{j}}, \;\;\;
\alpha_{j}\geq 0, \;\;\;
j = 1, \, . \, . \, . \,, k;
\;\;\;\;\;\;
{\cal A}_{n} \; \triangleq \; \textstyle \prod_{j\,=\,1}^{k}{\cal X}_{n}^{(j)} \; \subset \; {\mathbb{R}\mathstrut}^{k}.
\nonumber
\end{align}
Let then ${P\mathstrut}_{\!{\bf X}} \in {\cal P}_{n}({\cal A}_{n})$
be a {\em joint type} of $k$ symbols with the denominator $n$. Its support is defined as
\begin{align}
{\cal S}({P\mathstrut}_{\!{\bf X}}) \;\; & \triangleq \;\;
\big\{
{\bf x} \in {\cal A}_{n} : \; 
{P\mathstrut}_{\!{\bf X}}({\bf x}) > 0
\big\}.
\nonumber
\end{align}

\begin{lemma}[Support of a joint type] \label{LemSupportXk}\newline
{\em 
Let ${P\mathstrut}_{\!{\bf X}} \in {\cal P}_{n}({\cal A}_{n})$
be a joint type, such that 
$\mathbb{E}_{{P\mathstrut}_{\!{\bf X}}}\!\big[{\|{\bf X}\|\mathstrut}^{p}_{p}\big] \leq c$ with a real $p\geq 1$.
Then
}
\begin{align}
\big|\, {\cal S}({P\mathstrut}_{\!{\bf X}}) \,\big| \;\; \leq \;\;\, &
\sigma(k, p, c)
\cdot
n^{\frac{k\,+\,p\sum_{j=1}^{k}\alpha_{j}}{k\,+\,p}},
\nonumber \\
& \sigma(k, p, c)
\; \triangleq \;
\big[V_{k}(1)\big]^{\frac{p}{k\,+\,p}}
\big[
{\textstyle
\frac{k\,+\,1}{k}
\big(\max\,\{1, \, k^{1/2\,-\,1/p}\}\cdot c^{1/p} \, + \, \min\,\{k/4, \,\sqrt{k/12}\}\big)
}
\big]^{\frac{kp}{k\,+\,p}},
\nonumber
\end{align}
\end{lemma}
{\em where $V_{k}(1)$ is the volume of the $k$-dimensional ball of radius $1$.
}

\bigskip

\begin{proof}
Starting from the condition on ${P\mathstrut}_{\!{\bf X}}$, we obtain the following succession of inequalities:
\begin{align}
\sum_{{\bf x} \, \in \, {\cal S}({P\mathstrut}_{\!{\bf X}})}
{P\mathstrut}_{\!{\bf X}}({\bf x}) {\|{\bf x}\|\mathstrut}^{p}_{p}
\;\; & \leq \;\;
c,
\nonumber \\
\sum_{{\bf x} \, \in \, {\cal S}({P\mathstrut}_{\!{\bf X}})}
\tfrac{1}{n} {\|{\bf x}\|\mathstrut}^{p}_{p}
\;\; & \leq \;\;
c,
\nonumber \\
\sum_{{\bf x} \, \in \, {\cal S}({P\mathstrut}_{\!{\bf X}})}
\tfrac{1}{|\, {\cal S}({P\mathstrut}_{\!{\bf X}}) \,|} {\|{\bf x}\|\mathstrut}^{p}_{p}
\;\; & \leq \;\;
\tfrac{nc}{|\, {\cal S}({P\mathstrut}_{\!{\bf X}}) \,|},
\nonumber \\
\mathbb{E} \big[{\|{\bf U}\|\mathstrut}^{p}_{p}\big]
\;\; & \leq \;\;
\tfrac{nc}{|\, {\cal S}({P\mathstrut}_{\!{\bf X}}) \,|},
\nonumber \\
\big(\mathbb{E} \big[{\|{\bf U}\|\mathstrut}_{p}\big]\big)^{p}
\;\; & \leq \;\;
\tfrac{nc}{|\, {\cal S}({P\mathstrut}_{\!{\bf X}}) \,|},
\label{eqJensenpk} \\
\big(\mathbb{E} \big[k^{\min\{0, \; 1/p\, - \, 1/2\}}{\|{\bf U}\|\mathstrut}_{2}\big]\big)^{p}
\;\; & \leq \;\;
\tfrac{nc}{|\, {\cal S}({P\mathstrut}_{\!{\bf X}}) \,|},
\label{eqMonotNormk}
\end{align}
where ${\bf U} \sim \text{Discrete-Uniform}\big({\cal S}({P\mathstrut}_{\!{\bf X}})\big)$, 
(\ref{eqJensenpk}) follows by Jensen's inequality for $p\geq 1$,
while (\ref{eqMonotNormk}) follows 
by the relation between the norms of a real vector of length $k$:
${\|{\bf U}\|\mathstrut}_{p} \geq k^{\min\{0, \; 1/p\, - \, 1/2\}}{\|{\bf U}\|\mathstrut}_{2}$. 
To the discrete random vector ${\bf U}$ let us add a {\em continuously} distributed random vector
\begin{displaymath}
{\bf D} \; \sim \; \text{Continuous-Uniform}\big(\,
\textstyle\prod_{j\,=\,1}^{k}[-\Delta_{j,\,n}/2, \, \Delta_{j,\,n}/2)
\,\big),
\end{displaymath}
independent of ${\bf U}$.
Then $\widetilde{\bf U} \, = \, {\bf U} + {\bf D} \, \sim \, \text{Continuous-Uniform}(A)$,
where
\begin{displaymath}
A \;\; \triangleq \;\; \bigcup_{{\bf x} \, \in \, {\cal S}({P\mathstrut}_{\!{\bf X}})}
\big\{
{\bf x} \, + \,\textstyle\prod_{j\,=\,1}^{k}
[-\Delta_{j,\,n}/2, \, \Delta_{j,\,n}/2)
\big\}.
\end{displaymath}
Using the property ${\|{\bf D}\|\mathstrut}_{2} \leq {\|{\bf D}\|\mathstrut}_{1}$ and Jensen's inequality,
we write:
\begin{align}
\mathbb{E}\big[{\|{\bf D}\|\mathstrut}_{2}\big] \; \leq \;
\min_{r\,\in\,\{1, \,2\}} \mathbb{E}\big[{\|{\bf D}\|\mathstrut}_{r}\big]
\; \leq \;
\min_{r\,\in\,\{1, \,2\}} \big(\mathbb{E}\big[{\|{\bf D}\|\mathstrut}_{r}^{r}\big]\big)^{1/r}
\; & = \;
\min\Big\{
\tfrac{1}{4}\textstyle\sum_{j \, = \, 1}^{k}\Delta_{j,\,n}, \;
\sqrt{\tfrac{1}{12}\textstyle\sum_{j \, = \, 1}^{k}\Delta_{j,\,n}^{2}}
\Big\}
\nonumber \\
\; & \leq \;
\min\Big\{
\tfrac{k}{4}, \; \sqrt{\tfrac{k}{12}}
\Big\} \; \triangleq \; m_{1}(k).
\label{eqD2min}
\end{align}
Using bounds (\ref{eqMonotNormk}) and (\ref{eqD2min}), we obtain
\begin{align}
k^{m_{2}(p)}\tfrac{(nc)^{1/p}}{|\, {\cal S}({P\mathstrut}_{\!{\bf X}}) \,|^{1/p}}
\, + \,
m_{1}(k)
\;\; & \geq \;\;
\mathbb{E} \big[{\|{\bf U}\|\mathstrut}_{2}\big] \, + \, \mathbb{E} \big[{\|{\bf D}\|\mathstrut}_{2}\big]
\;\; = \;\;
\mathbb{E} \big[{\|{\bf U}\|\mathstrut}_{2} + {\|{\bf D}\|\mathstrut}_{2}\big]
\;\; \geq \;\;
\mathbb{E} \big[{\|\widetilde{\bf U}\|\mathstrut}_{2}\big]
\nonumber \\
& = \;\;
\tfrac{1}{\text{vol}(A)}
\int_{A}{\|\widetilde{\bf u}\|\mathstrut}_{2}\,d\widetilde{\bf u}
\nonumber \\
& = \;\;
\tfrac{1}{\text{vol}(A)}
\bigg[
\int_{A\,\cap\,B}{\|\widetilde{\bf u}\|\mathstrut}_{2}\,d\widetilde{\bf u}
\, +
\int_{A\,\cap\,B{\mathstrut}^{c}}{\|\widetilde{\bf u}\|\mathstrut}_{2}\,d\widetilde{\bf u}
\bigg]
\nonumber \\
& \overset{a}{\geq} \;\;
\tfrac{1}{\text{vol}(A)}
\bigg[
\int_{A\,\cap\,B}{\|\widetilde{\bf u}\|\mathstrut}_{2}\,d\widetilde{\bf u}
\, +
\int_{A{\mathstrut}^{c}\,\cap\,B}{\|{\bf t}\|\mathstrut}_{2}\,d{\bf t}
\bigg]
\nonumber \\
& = \;\;
\tfrac{1}{\text{vol}(A)}
\int_{B}{\|{\bf t}\|\mathstrut}_{2}\,d{\bf t}
\;\; = \;\; \tfrac{V_{k}(1)}{\text{vol}(A)}\int_{0}^{\left(\text{vol}(A)/V_{k}(1)\right)^{1/k}}
rdr^{k},
\label{eqCenteredBallk}
\end{align}
where $m_{2}(p) \triangleq \max\big\{0, \, 1/2 - 1/p\big\}$ 
and for ($a$) we use the $k$-dimensional ball region $B \, \triangleq \, \big\{\widetilde{\bf u}: \; V_{k}({\|\widetilde{\bf u}\|\mathstrut}_{2}) \, \leq \, \text{vol}(A)\big\}$,
centered around zero and of the same total volume as $A$,
so that
\begin{align}
\text{vol}(A\cap B) \, + \, \text{vol}(A\cap B{\mathstrut}^{c})
\;\; = \;\;
\text{vol}(A)
\;\; & = \;\;
\text{vol}(B)
\;\; = \;\;
\text{vol}(A\cap B) \, + \, \text{vol}(A{\mathstrut}^{c}\cap B),
\nonumber \\
\text{vol}(A\cap B{\mathstrut}^{c})
\;\; & = \;\;
\text{vol}(A{\mathstrut}^{c}\cap B),
\nonumber \\
\int_{A\,\cap\,B{\mathstrut}^{c}}{\|\widetilde{\bf u}\|\mathstrut}_{2}\,d\widetilde{\bf u}
\;\; \geq \;\;
\Big[\tfrac{\text{vol}(A)}{V_{k}(1)}\Big]^{1/k}\int_{A\,\cap\,B{\mathstrut}^{c}}d\widetilde{\bf u}
\;\; & = \;\;
\Big[\tfrac{\text{vol}(A)}{V_{k}(1)}\Big]^{1/k}\int_{A{\mathstrut}^{c}\,\cap\,B}d{\bf t}
\;\; \geq \;\;
\int_{A{\mathstrut}^{c}\,\cap\,B}{\|{\bf t}\|\mathstrut}_{2}\,d{\bf t}.
\nonumber
\end{align}
Completing the integration on the RHS of (\ref{eqCenteredBallk}), we obtain
\begin{align}
\tfrac{k}{k \,+\, 1}
\Big[\tfrac{\text{vol}(A)}{V_{k}(1)}\Big]^{1/k}
 = \;\;
\tfrac{k}{k \,+\, 1}
\Big[
\tfrac{|\, {\cal S}({P\mathstrut}_{\!{\bf X}}) \,| \,\cdot\, \prod_{j \, = \, 1}^{k}\Delta_{j,\,n}}{V_{k}(1)}
\Big]^{1/k}
 & \leq \;\;
k^{m_{2}(p)}\tfrac{(nc)^{1/p}}{|\, {\cal S}({P\mathstrut}_{\!{\bf X}}) \,|^{1/p}}
\, + \,
m_{1}(k),
\nonumber \\
\tfrac{k}{k \,+\, 1}\,
|\, {\cal S}({P\mathstrut}_{\!{\bf X}}) \,|^{1/p\,+\,1/k}
\Big[
\tfrac{\prod_{j \, = \, 1}^{k}\Delta_{j,\,n}}{V_{k}(1)}
\Big]^{1/k}
 & \leq \;\;
k^{m_{2}(p)}(nc)^{1/p}
\, + \,
m_{1}(k)
{\underbrace{|\, {\cal S}({P\mathstrut}_{\!{\bf X}}) \,|}_{\leq \; n}}^{\,1/p},
\nonumber \\
\tfrac{k}{k \,+\, 1}\,
|\, {\cal S}({P\mathstrut}_{\!{\bf X}}) \,|^{1/p\,+\,1/k}
\Big[
\tfrac{\prod_{j \, = \, 1}^{k}\Delta_{j,\,n}}{V_{k}(1)}
\Big]^{1/k}
 & \leq \;\;
\big(k^{m_{2}(p)}c^{1/p} + m_{1}(k)\big)n^{1/p},
\nonumber \\
|\, {\cal S}({P\mathstrut}_{\!{\bf X}}) \,|^{1/p\,+\,1/k}
\;\; & \leq \;\;
\tfrac{k \,+\, 1}{k}
\big[V_{k}(1)\big]^{1/k}
\big(k^{m_{2}(p)}c^{1/p} + m_{1}(k)\big)
n^{1/p \, + \, (1/k)\sum_{j=1}^{k}\alpha_{j}}
.
\nonumber
\end{align}
\end{proof}

\section*{Appendix B}\label{AppendixB}


\subsection*{Proof of Lemma~\ref{LemQuant}:}


Using ${P\mathstrut}_{\!X}$ and ${p\mathstrut}_{Y|X}$, it is convenient to define a joint probability density function over $\mathbb{R}^{2}$ as
\begin{align}
{p\mathstrut}_{XY}(x, y) \;\; & \triangleq \;\;
\frac{{P\mathstrut}_{\!X}(i\Delta_{\alpha,\,n})}{\Delta_{\alpha,\,n}}
{p\mathstrut}_{Y|X}(y\,|\,i\Delta_{\alpha,\,n}),
\nonumber \\
& \;\;\;\;\;\;\;\;\;\;\;\;\;\,
\forall x \in \big[(i - 1/2)\Delta_{\alpha,\,n}, \;  (i + 1/2)\Delta_{\alpha,\,n}\big),
\;\forall i \in \mathbb{Z}, \;\forall y \in \mathbb{R},
\label{eqJointPDF}
\end{align}
which is changing only stepwise in $x$-direction.
Note that ${p\mathstrut}_{Y}$ 
of (\ref{eqMargEntropy}) is the $y$-marginal of ${p\mathstrut}_{XY}$.
This gives
\begin{equation} \label{eqJPDFVariance}
\mathbb{E}_{{p\mathstrut}_{XY}}\!\big[{|Y|\mathstrut}^{p}\big] \; = \;
\mathbb{E}_{{P\mathstrut}_{\!X}{p\mathstrut}_{Y|X}}\!\big[{|Y|\mathstrut}^{p}\big]
\; \leq \; c_{Y}.
\end{equation}
Using Jensen's inequalities 
for $p \leq r$, and 
for $r\geq 1$ of the form 
$|x + t{|\mathstrut}^{r} \leq {2\mathstrut}^{r}\big(\tfrac{1}{2}|x{|\mathstrut}^{r} + \tfrac{1}{2}|t{|\mathstrut}^{r}\big)$,
we obtain
\begin{align}
&
\mathbb{E}_{{p\mathstrut}_{XY}}\!\big[{|X|\mathstrut}^{p}\big]
\, \leq \,
\big(
\mathbb{E}_{{p\mathstrut}_{XY}}\!\big[{|X|\mathstrut}^{r}\big]
\big)^{p/r}
\leq \,
\Big(
{2\mathstrut}^{r\,-\,1}\mathbb{E}_{{P\mathstrut}_{\!X}}\!\big[{|X|\mathstrut}^{r}\big]
+
\tfrac{\Delta_{\alpha,\,n}^{r}}{2(r\,+\,1)}
\Big)^{\!p/r}
\leq \,
\Big(
{2\mathstrut}^{p\,-\,1}c_{X} + \tfrac{1}{2(p\,+\,1)} 
\Big)^{\!p/r}
\, \triangleq \, \tilde{c}_{X}.
\label{eqJPDFVarianceX}
\end{align}
We proceed in two stages. First we quantize ${p\mathstrut}_{XY}(x, y)$ by rounding it {\em down} and check
the effect of this on the LHS of (\ref{eqCondEntropy})-(\ref{eqMargEntropy}).
Then we complement the total probability back to $1$, so that the type ${P\mathstrut}_{\!X}$ is conserved,
and check the effect of this on the RHS of (\ref{eqCondEntropy})-(\ref{eqMargEntropy}).

The quantization of ${p\mathstrut}_{XY}(x, y)$ is done by first replacing it with its infimum in each
rectangle $$\big[(i-1/2)\Delta_{\alpha,\,n}, \; (i+1/2)\Delta_{\alpha,\,n}\big) \; \times \;
\big[(j-1/2)\Delta_{\beta,\,n}, \; (j+1/2)\Delta_{\beta,\,n}\big):$$
\begin{align}
{p\mathstrut}_{XY}^{\inf}(x, y)
\;\; & \triangleq \;\;
\inf_{(j \, - \, 1/2)\Delta_{\beta,\,n} \; \leq \; \tilde{y} \; < \; (j \, + \, 1/2)\Delta_{\beta,\,n}}
{p\mathstrut}_{XY}(x, \tilde{y})
,
\nonumber \\
&
\;\;\;\;\;\;\;\;\;\;\;\;\;\;\;\;\;\;\;\;\;\;\;\;
\forall y \in \big[(j - 1/2)\Delta_{\beta,\,n}, \;  (j + 1/2)\Delta_{\beta,\,n}\big),
\;\forall j \in \mathbb{Z}, \;\forall x \in \mathbb{R},
\label{eqInf}
\end{align}
and then quantizing this infimum
down
to the nearest value $k\Delta_{\gamma,\,n}$, $k = 0, 1, 2, .\,.\,.\,$:  
\begin{align}
\widehat{p}{\mathstrut}_{XY}(x, y)
\;\; & \triangleq \;\;
\left\lfloor
{p\mathstrut}_{XY}^{\inf}(x, y)/\Delta_{\gamma,\,n}\right\rfloor\cdot\Delta_{\gamma,\,n},
\;\;\;\;\;\;
\forall (x, y) \in \mathbb{R}^{2}.
\label{eqQuantized}
\end{align}
Due to (\ref{eqCube}), the integral of $\widehat{p}{\mathstrut}_{XY}$ over $\mathbb{R}^{2}$ can be smaller than $1$ only by an integer 
multiple
of $1/n$.
The resulting difference from ${p\mathstrut}_{XY}(x, y)$ at each point $(x, y) \in \mathbb{R}^{2}$ can be bounded 
by a sum of two terms as
\begin{align}
0 \;\; \leq \;\;
{p\mathstrut}_{XY}(x, y) \, - \, \widehat{p}{\mathstrut}_{XY}(x, y)
\;\; & \leq \;\;
\underbrace{(K{P\mathstrut}_{\!X}(i\Delta_{\alpha,\,n})/\Delta_{\alpha,\,n})\cdot\Delta_{\beta,\,n}}_{\text{minimization}} \; + \!\!
\underbrace{\Delta_{\gamma,\,n}}_{\text{quantization}} \; \triangleq \;\; h(x)
\label{eqDifferencex} \\
& \leq \;\;
K n^{\alpha \, - \, \beta}
\, + \, n^{\alpha \, + \, \beta \, - \, 1}
\;\; \leq \;\;
(K + 1)n^{-\delta}
\;\; \triangleq \;\; h,
\label{eqDifference} \\
& 
\forall x \in \big[(i - 1/2)\Delta_{\alpha,\,n}, \;  (i + 1/2)\Delta_{\alpha,\,n}\big),
\;\forall i \in \mathbb{Z}, \;\forall y \in \mathbb{R},
\nonumber
\end{align}
where $\delta \,\triangleq\, \min\,\{\beta, \, 1-\beta\} \, - \, \alpha$ and $K$ is the parameter from (\ref{eqLipschitz}).

For (\ref{eqMargEntropy}) we will require the $y$-marginal of $\widehat{p}{\mathstrut}_{XY}$ from (\ref{eqQuantized}), defined in the usual manner:
\begin{align}
\widehat{p}{\mathstrut}_{Y}(y) \;\; \triangleq \;\;
\int_{\mathbb{R}}\widehat{p}{\mathstrut}_{XY}(x, y)dx
\;\; & \underset{a}{=} \;\;
\sum_{x\,\in\, {\cal S}({P\mathstrut}_{\!X})}\widehat{p}{\mathstrut}_{XY}(x, y)\Delta_{\alpha,\,n}
\;\; \underset{b}{\geq} \;\;
\sum_{x\,\in\, {\cal S}({P\mathstrut}_{\!X})}\big({p\mathstrut}_{XY}(x, y) - h(x)\big)\Delta_{\alpha,\,n}
\label{eqMargUsual} \\
& = \;\;
\sum_{x\,\in\, {\cal S}({P\mathstrut}_{\!X})}{p\mathstrut}_{XY}(x, y)\Delta_{\alpha,\,n} \,\,\,\, - \;\; \sum_{x\,\in\, {\cal S}({P\mathstrut}_{\!X})}\big[K{P\mathstrut}_{\!X}(x)\Delta_{\beta,\,n}
+ \Delta_{\gamma,\,n}\Delta_{\alpha,\,n}\big]
\nonumber \\
& = \;\; {p\mathstrut}_{Y}(y) \; - \; 
K\Delta_{\beta,\,n}
\; - \; |\, {\cal S}({P\mathstrut}_{\!X}) \,|\,\Delta_{\gamma,\,n}\Delta_{\alpha,\,n},
\;\;\;\;\;\; \forall y \in \mathbb{R},
\nonumber
\end{align}
where
the equality
($a$) follows from (\ref{eqJointPDF}), (\ref{eqInf}), (\ref{eqQuantized}),
and the inequality ($b$) follows by (\ref{eqDifferencex}).
Then
\begin{align}
0 \;\; \leq \;\;
{p\mathstrut}_{Y}(y) \, - \, \widehat{p}{\mathstrut}_{Y}(y)
\;\; & \leq \;\;
K\Delta_{\beta,\,n}
\, + \, |\, {\cal S}({P\mathstrut}_{\!X}) \,|\,\Delta_{\gamma,\,n}\Delta_{\alpha,\,n}
\nonumber \\
& \overset{a}{\leq} \;\;
Kn^{-\beta} \, + \,
(4c_{X}^{1/r} + 1)^{r/(1\,+\,r)}n^{-1 \, + \, \beta \, + \,
(1 \, + \, r\alpha)/(1\,+\,r)}
\nonumber \\
& \overset{b}{\leq} \;\;
\big[K\, + \,
(4c_{X}^{1/r} + 1)^{r/(1\,+\,r)}\big]n^{-\delta_{1}}
\;\; \triangleq \;\; h_{1},
\;\;\;\;\;\; \forall y \in \mathbb{R},
\label{eqYDifference}
\end{align}
where ($a$) follows by (\ref{eqSX}) of Lemma~\ref{LemSupportXY} and (\ref{eqDelta}),
and ($b$) holds with $\delta_{1} \,\triangleq\, \min\,\{\beta, \; (1 -\alpha)r/(1 + r) - \beta\}$. 
Comparing two parameters $\delta$ and $\delta_{1}$ of 
the bounds (\ref{eqDifference}) and (\ref{eqYDifference}), respectively, we find that simultaneously
$(\delta > 0) \land (\delta_{1} > 0)$  
if and only if $\alpha \in \big(0, \tfrac{r}{1\,+\,2r}\big)$ and
$\alpha < \beta < \tfrac{r}{1\,+\,r}(1-\alpha)$.

\bigskip

{\em The LHS of (\ref{eqCondEntropy})}


Now 
consider
the LHS of (\ref{eqCondEntropy}).
Note that each function ${p\mathstrut}_{Y|X}(\,\cdot\,|\,x)$ in (\ref{eqCondEntropy})
is bounded and cannot exceed $\sqrt{K}$. Then its differential entropy is {\em lower-bounded} by $-{\log\mathstrut}_{\!b}\,\sqrt{K}$.
On the other hand, the condition 
of the lemma $\mathbb{E}_{{P\mathstrut}_{\!X}{p\mathstrut}_{Y|X}}\!\big[{|Y-X|\mathstrut}^{q}\big]\leq c_{XY}$
implies that 
each function ${p\mathstrut}_{Y|X}(\,\cdot\,|\,x)$ in (\ref{eqCondEntropy}) at ${P\mathstrut}_{\!X}(x) > 0$ 
satisfies $\mathbb{E}_{{P\mathstrut}_{\!X}{p\mathstrut}_{Y|X}}\!\big[{|Y-x|\mathstrut}^{q}\, | \, X = x\big]\leq nc_{XY}$.
Then, similarly as in \cite[Theorem.~8.6.5]{CoverThomas}, its differential entropy is {\em upper-bounded}
by the entropy of a generalized Gaussian PDF with the power parameter $q$. 
It follows that each ${p\mathstrut}_{Y|X}(\,\cdot\,|\,x)$ in (\ref{eqCondEntropy}) has 
a finite
differential entropy. 
With (\ref{eqJointPDF}) we can rewrite the LHS of (\ref{eqCondEntropy}) as
\begin{align}
\sum_{x\,\in\,{\cal X}_{n}}{P\mathstrut}_{\!X}(x)\int_{\mathbb{R}}{p\mathstrut}_{Y|X}(y\,|\,x)
{\log\mathstrut}_{\!b}\,{p\mathstrut}_{Y|X}(y\,|\,x)dy
\;\; = \;\; &
\int\!\!\!\!\int_{\mathbb{R}^{2}}{p\mathstrut}_{XY}(x, y)
{\log\mathstrut}_{\!b}\,{p\mathstrut}_{XY}(x, y)dxdy
\nonumber \\
&
- \sum_{x\,\in\,{\cal X}_{n}}
{P\mathstrut}_{\!X}(x){\log\mathstrut}_{\!b}\,\frac{{P\mathstrut}_{\!X}(x)}{\Delta_{\alpha,\,n}}.
\label{eqCondEntIdeal}
\end{align}
Let us examine the possible increase in (\ref{eqCondEntIdeal})
when ${p\mathstrut}_{XY}$ is replaced with $\widehat{p}{\mathstrut}_{XY}$ defined by (\ref{eqInf})-(\ref{eqQuantized}).
For this, let us define a set in $\mathbb{R}^{2}$ with respect to the parameter $h$ of (\ref{eqDifference}):
\begin{align}
A \;\; & \triangleq \;\;
\big\{
(x, y): \;\; {p\mathstrut}_{XY}(x, y) > h
\big\},
\label{eqCountableUnion}
\end{align}
which is a countable union of disjoint rectangles
by the definition of ${p\mathstrut}_{XY}$ in (\ref{eqJointPDF}). Then
\begin{align}
&
\;\;\;\;
\int\!\!\!\!\int_{\mathbb{R}^{2}}\widehat{p}{\mathstrut}_{XY}(x, y)
{\log\mathstrut}_{\!b}\,\widehat{p}{\mathstrut}_{XY}(x, y)dxdy
\, - \,
\int\!\!\!\!\int_{\mathbb{R}^{2}}{p\mathstrut}_{XY}(x, y)
{\log\mathstrut}_{\!b}\,{p\mathstrut}_{XY}(x, y)dxdy
\nonumber \\
=  &
\;\;
\bigg[\int\!\!\!\!\int_{A{\mathstrut}^{c}}\widehat{p}{\mathstrut}_{XY}(x, y)
{\log\mathstrut}_{\!b}\,\widehat{p}{\mathstrut}_{XY}(x, y)dxdy
\, - \,
\int\!\!\!\!\int_{A{\mathstrut}^{c}}{p\mathstrut}_{XY}(x, y)
{\log\mathstrut}_{\!b}\,{p\mathstrut}_{XY}(x, y)dxdy
\bigg]
\nonumber \\
&
+
\int\!\!\!\!\int_{A}
\big[\,\widehat{p}{\mathstrut}_{XY}(x, y)
{\log\mathstrut}_{\!b}\,\widehat{p}{\mathstrut}_{XY}(x, y)
\, - \,
{p\mathstrut}_{XY}(x, y)
{\log\mathstrut}_{\!b}\,{p\mathstrut}_{XY}(x, y)\big]
\, dxdy.
\label{eqByParts}
\end{align}
Note that the minimum of the function $f(t) = t\,{\log\mathstrut}_{\!b} \,t$ occurs at $t = 1/e$.
Then for $h \leq 1/e$
we have ${p\mathstrut}_{XY}(x, y)
{\log\mathstrut}_{\!b}\,{p\mathstrut}_{XY}(x, y) \, \leq \, \widehat{p}{\mathstrut}_{XY}(x, y)
{\log\mathstrut}_{\!b}\,\widehat{p}{\mathstrut}_{XY}(x, y) \, \leq \, 0$
for all $(x, y) \in A{\mathstrut}^{c}$
and the first of the two terms in (\ref{eqByParts}) is upper-bounded as
\begin{align}
&
\;\;\,\,\,
\int\!\!\!\!\int_{A{\mathstrut}^{c}}\widehat{p}{\mathstrut}_{XY}(x, y)
{\log\mathstrut}_{\!b}\,\widehat{p}{\mathstrut}_{XY}(x, y)dxdy
\, - \,
\int\!\!\!\!\int_{A{\mathstrut}^{c}}{p\mathstrut}_{XY}(x, y)
{\log\mathstrut}_{\!b}\,{p\mathstrut}_{XY}(x, y)dxdy
\label{eqLocalEntropy} \\
\overset{h\, \leq \, 1/e}{\leq} \;\; &
-\int\!\!\!\!\int_{A{\mathstrut}^{c}}{p\mathstrut}_{XY}(x, y)
{\log\mathstrut}_{\!b}\,{p\mathstrut}_{XY}(x, y)dxdy
\;\; \overset{(*)}{=} \;\;
-p_{0}
\int\!\!\!\!\int_{\mathbb{R}^{2}}{p\mathstrut}_{XY}^{c}(x, y)
{\log\mathstrut}_{\!b}\,{p\mathstrut}_{XY}^{c}(x, y)dxdy
\, - \,
p_{0}\,{\log\mathstrut}_{\!b}\,p_{0},
\nonumber
\end{align}
where the equality ($*$) is appropriate
for the case when the upper bound is positive, 
with the definitions:
\begin{align}
p_{0} \;\; & \triangleq \;\;
\int\!\!\!\!\int_{A{\mathstrut}^{c}}{p\mathstrut}_{XY}(x, y)dxdy,
\label{eqP} \\
{p\mathstrut}_{XY}^{c}(x, y)
\;\; & \triangleq \;\;
\Bigg\{
\begin{array}{r r}
{p\mathstrut}_{XY}(x, y)/p_{0}, & (x, y) \in  A{\mathstrut}^{c}, \\
0, & \text{o.w.}
\end{array}
\nonumber
\end{align}
Next we upper-bound the entropy of the probability density function ${p\mathstrut}_{XY}^{c}$ on the RHS of (\ref{eqLocalEntropy})
by that of 
a generalized Gaussian PDF, similarly as in \cite[Theorem.~8.6.5]{CoverThomas}, using (\ref{eqJPDFVariance}) and (\ref{eqJPDFVarianceX}). 
By (\ref{eqJPDFVariance}) we have
\begin{align}
c_{Y}
\; \geq \;
\int\!\!\!\!\int_{A{\mathstrut}^{c}}{p\mathstrut}_{XY}(x, y){|y|\mathstrut}^{p}dxdy
\; = \;
p_{0}\int\!\!\!\!\int_{\mathbb{R}^{2}}{p\mathstrut}_{XY}^{c}(x, y){|y|\mathstrut}^{p}dxdy
\; & = \;
p_{0} \int_{\mathbb{R}}{p\mathstrut}_{Y}^{c}(y){|y|\mathstrut}^{p}dy,
\nonumber \\
\tfrac{1}{p}\,
{\log\mathstrut}_{\!b}\big(
[2\Gamma(1/p)]^{p}p^{1\,-\,p}e c_{Y}/p_{0}
\big)
\; & \geq \;
-\int_{\mathbb{R}}
{p\mathstrut}_{Y}^{c}(y)
{\log\mathstrut}_{\!b}\,{p\mathstrut}_{Y}^{c}(y)
dy,
\nonumber
\end{align}
where the LHS of the last inequality represents the differential entropy of the generalized Gaussian PDF ${p\mathstrut}_{Y}^{G}(y)$
with the power parameter $p$, with zero mean and the scaling parameter such that 
$\mathbb{E}_{{p\mathstrut}_{Y}^{G}}\!\big[{|Y|\mathstrut}^{p}\big] = c_{Y}/p_{0}$. 
Repeating this with $\tilde{c}_{X}$ of (\ref{eqJPDFVarianceX}), we can write:
\begin{align}
\tfrac{2}{p}\,
{\log\mathstrut}_{\!b}\big(
[2\Gamma(1/p)]^{p}p^{1\,-\,p}e [\tilde{c}_{X}c_{Y}]^{1/2}/p_{0}
\big)
\; & \geq \;
-\int\!\!\!\!\int_{\mathbb{R}^{2}}{p\mathstrut}_{XY}^{c}(x, y)
{\log\mathstrut}_{\!b}\,{p\mathstrut}_{XY}^{c}(x, y)dxdy.
\nonumber
\end{align}
So we can rewrite the bound of (\ref{eqLocalEntropy}) in terms of $p_{0}$ defined by (\ref{eqP}):
\begin{align}
&
\int\!\!\!\!\int_{A{\mathstrut}^{c}}\widehat{p}{\mathstrut}_{XY}(x, y)
{\log\mathstrut}_{\!b}\,\widehat{p}{\mathstrut}_{XY}(x, y)dxdy
\, - \,
\int\!\!\!\!\int_{A{\mathstrut}^{c}}{p\mathstrut}_{XY}(x, y)
{\log\mathstrut}_{\!b}\,{p\mathstrut}_{XY}(x, y)dxdy
\nonumber \\
\overset{h\, \leq \, 1/e}{\leq} \;\; &
-\tfrac{p\,+\,2}{p}\,p_{0}\,{\log\mathstrut}_{\!b}\,p_{0}
\, + \, p_{0}\,{\log\mathstrut}_{\!b}\,c_{4},
\;\;\;\;\;\;\;\;\;
c_{4} \; \triangleq \; [2\Gamma(1/p)]^{2}[p^{1\,-\,p}e]^{2/p}[\tilde{c}_{X}c_{Y}]^{1/p}.
\label{eqLocalp}
\end{align}
From (\ref{eqP}) and (\ref{eqCountableUnion}) it is clear that $p_{0} \rightarrow 0$ as $h \rightarrow 0$.
In order to relate between them, 
we use (\ref{eqJPDFVariance}), (\ref{eqJPDFVarianceX}):
\begin{align}
& 
2^{\max\{0, \; p/2 \, - \, 1\}}(\tilde{c}_{X} + c_{Y})
\;\; \geq \;\;
2^{\max\{0, \; p/2 \, - \, 1\}}\!
\int\!\!\!\!\int_{A{\mathstrut}^{c}}{p\mathstrut}_{XY}(x, y)({|x|\mathstrut}^{p} + {|y|\mathstrut}^{p})dxdy
\nonumber \\
\overset{a}{\geq}
\;\; &
\int\!\!\!\!\int_{A{\mathstrut}^{c}}{p\mathstrut}_{XY}(x, y)(x^{2} + y^{2})^{p/2}dxdy
\nonumber \\
\overset{b}{=}
\;\; &
\int\!\!\!\!\int_{B{\mathstrut}_{1}}h\cdot(x^{2} + y^{2})^{p/2}dxdy
\; + \;
\int\!\!\!\!\int_{A{\mathstrut}^{c}\,\cap\,B{\mathstrut}^{c}_{1}}{p\mathstrut}_{XY}(x, y)(x^{2} + y^{2})^{p/2}dxdy
\; - \;
\int\!\!\!\!\int_{A\,\cap\,B{\mathstrut}_{1}}h\cdot(x^{2} + y^{2})^{p/2}dxdy
\nonumber \\
&
\;\;\;\;\;\;\;\;\;\;\;\;\;\;\;\;\;\;\;\;\;\;\;\;\;\;\;\;\;\;\;\;\;\;\;\;\;\;\;\;
\;\;\;\;\;\;\;\;\;\;\;\;\;\;\;\;\;\;\;\;\;\;\;\;\;\;\;\;\;\;\;\;\;\;\;\;\;\;\,\,
-\int\!\!\!\!\int_{A{\mathstrut}^{c}\,\cap\,B{\mathstrut}_{1}}\big(h - {p\mathstrut}_{XY}(x, y)\big)(x^{2} + y^{2})^{p/2}dxdy
\nonumber \\
\geq
\;\; &
\int\!\!\!\!\int_{B{\mathstrut}_{1}}h\cdot(x^{2} + y^{2})^{p/2}dxdy
\; + \;
\Big[\frac{p_{0}}{h\pi}\Big]^{p/2}
\int\!\!\!\!\int_{A{\mathstrut}^{c}\,\cap\,B{\mathstrut}^{c}_{1}}{p\mathstrut}_{XY}(x, y)dxdy
\; - \;
\Big[\frac{p_{0}}{h\pi}\Big]^{p/2}\int\!\!\!\!\int_{A\,\cap\,B{\mathstrut}_{1}}hdxdy
\nonumber \\
&
\;\;\;\;\;\;\;\;\;\;\;\;\;\;\;\;\;\;\;\;\;\;\;\;\;\;\;\;\;\;\;\;\;\;\;\;\;\;\;\;
\;\;\;\;\;\;\;\;\;\;\;\;\;\;\;\;\;\;\;\;\;\;\;\;\;\;\;\;\;\;\;\,
-\Big[\frac{p_{0}}{h\pi}\Big]^{p/2}\int\!\!\!\!\int_{A{\mathstrut}^{c}\,\cap\,B{\mathstrut}_{1}}\big(h - {p\mathstrut}_{XY}(x, y)\big)dxdy
\nonumber \\
= \;\; &
\int\!\!\!\!\int_{B{\mathstrut}_{1}}h\cdot(x^{2} + y^{2})^{p/2}dxdy \;\; = \;\; 
\tfrac{2}{p\,+\,2}{(h\pi)\mathstrut}^{-p/2}p_{0}^{(p\,+\,2)/2},
\nonumber
\end{align}
where in ($a$) we use the relation between the norms: 
${2\mathstrut}^{\max\{0, \; 1/2 \, - \, 1/p\}}({|x|\mathstrut}^{p} + {|y|\mathstrut}^{p}{)\mathstrut}^{1/p} \geq (x^{2} + y^{2}{)\mathstrut}^{1/2}$,
in ($b$) we use the disk set $B{\mathstrut}_{1} \triangleq \big\{ (x, y): \; h\pi(x^{2} + y^{2}) \, \leq \, p_{0}\big\}$,
centered around zero. 
This results in the following upper bound on $p_{0}$ in terms of $h$:
\begin{equation} \label{eqHP}
(c_{5}h)^{p/(p\,+\,2)}
\;\; \geq \;\; p_{0},
\end{equation}
where 
$c_{5} \triangleq 
\pi2^{\max\{0, \, 1 \, - \, 2/p \}}\big[\tfrac{p\,+\,2}{2}\,(\tilde{c}_{X} + c_{Y})\big]^{2/p}$. 
Substituting the LHS of (\ref{eqHP}) in (\ref{eqLocalp})
in place of $p_{0}$, we obtain the following upper bound on
the first half 
of (\ref{eqByParts}) in terms of $h$
of (\ref{eqDifference}), (\ref{eqCountableUnion}):
\begin{align}
&
\int\!\!\!\!\int_{A{\mathstrut}^{c}}\widehat{p}{\mathstrut}_{XY}(x, y)
{\log\mathstrut}_{\!b}\,\widehat{p}{\mathstrut}_{XY}(x, y)dxdy
\, - \,
\int\!\!\!\!\int_{A{\mathstrut}^{c}}{p\mathstrut}_{XY}(x, y)
{\log\mathstrut}_{\!b}\,{p\mathstrut}_{XY}(x, y)dxdy
\nonumber \\
\leq \;\; &
(c_{5}h)^{p/(p\,+\,2)}
\,{\log\mathstrut}_{\!b}\big(\!\max\{1, \, c_{4}\}/(c_{5}h)\big),
\;\;\;\;\;\;
h \,\leq \,
1/\max\{e, \, c_{5}e^{(p\,+\,2)/p}\}.
\label{eqFirstTerm}
\end{align}

In the second term of (\ref{eqByParts})
for $(x, y) \in A$ the integrand can be upper-bounded by Lemma~\ref{Lemxlogx} with its parameters $t$ and $t_{1}$ such that
\begin{align}
t_{1} \; = \; \widehat{p}{\mathstrut}_{XY}(x, y)
\; & \leq \; {p\mathstrut}_{XY}(x, y) \; = \; t_{1} + t,
\;\;\;\;\;\;\;\;\;
t \, \leq \, h \, \leq \, 1/e.
\nonumber
\end{align}
This gives
\begin{align}
&
\int\!\!\!\!\int_{A}
\big[\,\widehat{p}{\mathstrut}_{XY}(x, y)
{\log\mathstrut}_{\!b}\,\widehat{p}{\mathstrut}_{XY}(x, y)
\, - \,
{p\mathstrut}_{XY}(x, y)
{\log\mathstrut}_{\!b}\,{p\mathstrut}_{XY}(x, y)\big]
\,
dxdy
\nonumber \\
\;\; 
\leq
\;\; &
\text{vol}(A)\,
h\,
{\log\mathstrut}_{\!b}(1/h),
\;\;\;\;\;\; h \, \leq \, 1/e,
\label{eqVolh}
\end{align}
where
$\text{vol}(A)$ is the total area of $A$.
To find an upper bound on $\text{vol}(A)$, we use (\ref{eqJPDFVariance}), (\ref{eqJPDFVarianceX}):
\begin{align}
2^{\max\{0, \; p/2 \, - \, 1\}}(\tilde{c}_{X} + c_{Y})
\;\; & \geq \;\;
2^{\max\{0, \; p/2 \, - \, 1\}}
\!
\int\!\!\!\!\int_{A}{p\mathstrut}_{XY}(x, y)({|x|\mathstrut}^{p} + {|y|\mathstrut}^{p})dxdy
\nonumber \\
& \overset{a}{\geq} \;\;
\int\!\!\!\!\int_{A}{p\mathstrut}_{XY}(x, y)(x^{2} + y^{2})^{p/2}dxdy
\;\; \geq \;\;
\int\!\!\!\!\int_{A}h \cdot(x^{2} + y^{2})^{p/2}dxdy
\nonumber \\
& = \;\;
h\left[
\int\!\!\!\!\int_{A\,\cap\,B{\mathstrut}_{2}}(x^{2} + y^{2})^{p/2}dxdy
\, +
\int\!\!\!\!\int_{A\,\cap\,B{\mathstrut}^{c}_{2}}(x^{2} + y^{2})^{p/2}dxdy
\right]
\nonumber \\
& \overset{b}{\geq} \;\;
h\left[
\int\!\!\!\!\int_{A\,\cap\,B{\mathstrut}_{2}}(x^{2} + y^{2})^{p/2}dxdy
\, +
\int\!\!\!\!\int_{A{\mathstrut}^{c}\,\cap\,B{\mathstrut}_{2}}(x^{2} + y^{2})^{p/2}dxdy
\right]
\nonumber \\
& = \;\;
\int\!\!\!\!\int_{B{\mathstrut}_{2}}h \cdot(x^{2} + y^{2})^{p/2}dxdy
\;\; = \;\;
\tfrac{2}{p \, + \, 2}\,h\pi^{-p/2}(\text{vol}(A))^{(p\,+\,2)/2},
\nonumber
\end{align}
where in ($a$) we use the relation between the norms of a real vector of length $2$, 
and to achieve ($b$) we use the disk set
$B{\mathstrut}_{2} \, \triangleq \, \big\{(x, y): \; \pi(x^{2} + y^{2}) \, \leq \, \text{vol}(A)\big\}$,
centered around zero,
of the same total area as $A$, and the resulting property that
$\text{vol}(A{\mathstrut}^{c}\cap B{\mathstrut}_{2}) = \text{vol}(A\cap B{\mathstrut}^{c}_{2})$.
So that
\begin{equation} \label{eqVol}
c_{5}^{\,p/(p\,+\,2)}{h\mathstrut}^{-2/(p\,+\,2)}
\;\; \geq \;\;
\text{vol}(A).
\end{equation}
Continuing (\ref{eqVolh}), therefore we obtain the following
upper bound on the second term in (\ref{eqByParts}):
\begin{align}
\int\!\!\!\!\int_{A}
\big[\,\widehat{p}{\mathstrut}_{XY}(x, y)
{\log\mathstrut}_{\!b}\,\widehat{p}{\mathstrut}_{XY}(x, y)
\, - \,
{p\mathstrut}_{XY}(x, y)
{\log\mathstrut}_{\!b}\,{p\mathstrut}_{XY}(x, y)
\big]
dxdy
\;
& \overset{h\, \leq \, 1/e}{\leq}
\;
(c_{5}h)^{p/(p\,+\,2)}
{\log\mathstrut}_{\!b}(1/h).
\label{eqSecondTerm}
\end{align}
Putting (\ref{eqByParts}), (\ref{eqFirstTerm}) and (\ref{eqSecondTerm}) together:
\begin{align}
&
\int\!\!\!\!\int_{\mathbb{R}^{2}}\widehat{p}{\mathstrut}_{XY}(x, y)
{\log\mathstrut}_{\!b}\,\widehat{p}{\mathstrut}_{XY}(x, y)dxdy
\, - \,
\int\!\!\!\!\int_{\mathbb{R}^{2}}{p\mathstrut}_{XY}(x, y)
{\log\mathstrut}_{\!b}\,{p\mathstrut}_{XY}(x, y)dxdy
\nonumber \\
\leq \;\; &
(c_{5}h)^{p/(p\,+\,2)}
\,{\log\mathstrut}_{\!b}\big(\!\max\{1, \, c_{4}\}/(c_{5}h^{2})\big),
\;\;\;\;\;\;
h \,\leq \,
1/\max\{e, \, c_{5}e^{(p\,+\,2)/p}\},
\label{eqPutting}
\end{align}
where $c_{4}$, $c_{5}$,
and $h$ are such as in (\ref{eqLocalp}), (\ref{eqHP}),
and (\ref{eqDifference}), respectively.
So that if $\delta > 0$ in (\ref{eqDifference}), then the 
possible increase in (\ref{eqCondEntIdeal})
caused by substitution of $\widehat{p}{\mathstrut}_{XY}$ in place of ${p\mathstrut}_{XY}$ is at most $o(1)$.

Later on, for the RHS of (\ref{eqCondEntropy})-(\ref{eqMargEntropy}) we will require also
the loss in the total probability incurred in the replacement of ${p\mathstrut}_{XY}$ by $\widehat{p}{\mathstrut}_{XY}$.
This loss is strictly positive and tends to zero with $h$ of (\ref{eqDifference}): 
\begin{align}
0 \;\; < \;\; p_{1} \;\; & \triangleq \;\;
\underbrace{\int\!\!\!\!\int_{\mathbb{R}^{2}}{p\mathstrut}_{XY}(x, y)dxdy}_{= \; 1} \; - \;
\int\!\!\!\!\int_{\mathbb{R}^{2}}\widehat{p}{\mathstrut}_{XY}(x, y)dxdy
\nonumber \\
& \overset{a}{\leq} \;\;
\underbrace{\int\!\!\!\!\int_{A{\mathstrut}^{c}}{p\mathstrut}_{XY}(x, y)dxdy}_{= \; p_{0}} \; + \;
\int\!\!\!\!\int_{A}\underbrace{\big({p\mathstrut}_{XY}(x, y) - \widehat{p}{\mathstrut}_{XY}(x, y)\big)}_{\leq \; h}dxdy
\nonumber \\
& \overset{b}{\leq} \;\;
p_{0} \, + \, h\cdot\text{vol}(A) \;\; \overset{c}{\leq} \;\;
(c_{5}h)^{p/(p\,+\,2)} \, + \, (c_{5}h)^{p/(p\,+\,2)} \;\; = \;\; 2(c_{5}h)^{p/(p\,+\,2)},
\label{eqProbLoss}
\end{align}
where the set $A$ in ($a$) is defined in (\ref{eqCountableUnion}),
($b$) follows by (\ref{eqP}) and (\ref{eqDifference}),
and ($c$) follows by (\ref{eqHP}), 
(\ref{eqVol}).

\bigskip


{\em The LHS of (\ref{eqMargEntropy})}


Consider next the LHS of (\ref{eqMargEntropy}). Since ${p\mathstrut}_{Y} \in {\cal L}$
and $\mathbb{E}_{{p\mathstrut}_{Y}}\!\big[{|Y|\mathstrut}^{p}\big]\leq c_{Y}$, 
the differential entropy of ${p\mathstrut}_{Y}$ is finite.
Let us examine the possible decrease in the LHS of (\ref{eqMargEntropy})
when ${p\mathstrut}_{Y}$ is replaced with $\widehat{p}{\mathstrut}_{Y}$ defined in (\ref{eqMargUsual}).
For this, let us define a set in $\mathbb{R}$ with respect to the parameter $h_{1}$ of (\ref{eqYDifference}):
\begin{align}
{A\mathstrut}_{1} \;\; & \triangleq \;\;
\big\{
y: \;\; {p\mathstrut}_{Y}(y) > h_{1}
\big\},
\label{eqCountableUnion2}
\end{align}
which is a countable union of disjoint open intervals. Then
\begin{align}
&
\int_{\mathbb{R}}{p\mathstrut}_{Y}(y)
{\log\mathstrut}_{\!b}\,{p\mathstrut}_{Y}(y)dy
\, - \,
\int_{\mathbb{R}}\widehat{p}{\mathstrut}_{Y}(y)
{\log\mathstrut}_{\!b}\,\widehat{p}{\mathstrut}_{Y}(y)dy
\label{eqByParts2} \\
= \; &
\int_{{A\mathstrut}_{1}^{c}}
\big[\,{p\mathstrut}_{Y}(y)
{\log\mathstrut}_{\!b}\,{p\mathstrut}_{Y}(y)
\, - \,
\widehat{p}{\mathstrut}_{Y}(y)
{\log\mathstrut}_{\!b}\,\widehat{p}{\mathstrut}_{Y}(y)
\big]
\,dy
\; + \, 
\int_{{A\mathstrut}_{1}}
\big[\,
{p\mathstrut}_{Y}(y)
{\log\mathstrut}_{\!b}\,{p\mathstrut}_{Y}(y)
\, - \,
\widehat{p}{\mathstrut}_{Y}(y)
{\log\mathstrut}_{\!b}\,\widehat{p}{\mathstrut}_{Y}(y)
\big]
\,dy.
\nonumber
\end{align}
For $h_{1} \leq 1/e$
we have ${p\mathstrut}_{Y}(y)
{\log\mathstrut}_{\!b}\,{p\mathstrut}_{Y}(y) \, \leq \, \widehat{p}{\mathstrut}_{Y}(y)
{\log\mathstrut}_{\!b}\,\widehat{p}{\mathstrut}_{Y}(y) \, \leq \, 0$
for all $y \in {A\mathstrut}_{1}^{c}$
and the first of the two terms in (\ref{eqByParts2}) is non-positive:
\begin{align}
&
\int_{{A\mathstrut}_{1}^{c}}
\big[\,
{p\mathstrut}_{Y}(y)
{\log\mathstrut}_{\!b}\,{p\mathstrut}_{Y}(y)
\, - \,
\widehat{p}{\mathstrut}_{Y}(y)
{\log\mathstrut}_{\!b}\,\widehat{p}{\mathstrut}_{Y}(y)
\big]
\,dy
\;\; \leq \;\; 0.
\label{eqLocalEntropy2} 
\end{align}

In the second term of (\ref{eqByParts2})
for $y \in {A\mathstrut}_{1}$ the integrand can be upper-bounded by Lemma~\ref{Lemxlogx} with its parameters $t$ and $t_{1}$ such that
\begin{align}
t_{1} \; & = \; \widehat{p}{\mathstrut}_{Y}(y)
\; \leq \; t_{1} + t \; 
= \; {p\mathstrut}_{Y}(y)
\; \leq \; \sup_{y\, \in \, \mathbb{R}}{p\mathstrut}_{Y}(y)
\; \leq \; \sqrt{K},
\;\;\;\;\;\;\;\;\;
t \, 
\leq \, h_{1} \, \leq \, 1/e,
\nonumber
\end{align}
where $K$
is the parameter from (\ref{eqLipschitz}).
This gives
\begin{align}
\int_{{A\mathstrut}_{1}}
\big[\,
{p\mathstrut}_{Y}(y)
{\log\mathstrut}_{\!b}\,{p\mathstrut}_{Y}(y)
\, - \,
\widehat{p}{\mathstrut}_{Y}(y)
{\log\mathstrut}_{\!b}\,\widehat{p}{\mathstrut}_{Y}(y)
\big]
\,dy
\,
&
\overset{h_{1}\, \leq \, 1/e}{\leq}
\,
\text{vol}({A\mathstrut}_{1})\,
h_{1}
\max
\big\{
{\log\mathstrut}_{\!b}\,(1/h_{1}),
\;
{\log\mathstrut}_{\!b} (e\sqrt{K})
\big\},
\label{eqVolh2}
\end{align}
where $\text{vol}({A\mathstrut}_{1})$ is the total length of ${A\mathstrut}_{1}$.
It remains to find an upper bound on $\text{vol}({A\mathstrut}_{1})$. 
We use (\ref{eqJPDFVariance}):
\begin{align}
c_{Y}
\;\; \geq \;\;
\int_{{A\mathstrut}_{1}}{p\mathstrut}_{Y}(y){|y|\mathstrut}^{p}dy
\;\; 
\geq \;\;
\int_{{A\mathstrut}_{1}}h_{1} {|y|\mathstrut}^{p}dy
\;\; & = \;\;
h_{1}\left[
\int_{{A\mathstrut}_{1}\,\cap\,{B\mathstrut}_{3}}{|y|\mathstrut}^{p}dy
\, +
\int_{{A\mathstrut}_{1}\,\cap\,{B\mathstrut}_{3}^{c}}{|y|\mathstrut}^{p}dy
\right]
\nonumber \\
& \overset{a}{\geq} \;\;
h_{1}\left[
\int_{{A\mathstrut}_{1}\,\cap\,{B\mathstrut}_{3}}{|y|\mathstrut}^{p}dy
\, +
\int_{{A\mathstrut}_{1}^{c}\,\cap\,{B\mathstrut}_{3}}{|y|\mathstrut}^{p}dy
\right]
\nonumber \\
& = \;\;
\int_{{B\mathstrut}_{3}}h_{1} {|y|\mathstrut}^{p}dy
\;\; = \;\;
\tfrac{1}{{2\mathstrut}^{p}(p \, + \, 1)}\,h_{1}\cdot(\text{vol}({A\mathstrut}_{1}))^{p\,+\,1},
\nonumber
\end{align}
where in ($a$) we use the interval set
${B\mathstrut}_{3} \, \triangleq \, \big[
-\!\text{vol}({A\mathstrut}_{1})/2, \;
\text{vol}({A\mathstrut}_{1})/2
\,\big]$, centered around zero, and
of the same total length as ${A\mathstrut}_{1}$ with the resulting property that
$\text{vol}({A\mathstrut}_{1}^{c}\cap {B\mathstrut}_{3}) = \text{vol}({A\mathstrut}_{1}\cap {B\mathstrut}_{3}^{c})$.
So that
\begin{equation} \label{eqVol2}
c_{6}h^{-1/(p\,+\,1)}_{1}
\;\; \geq \;\;
\text{vol}({A\mathstrut}_{1}),
\end{equation}
where $c_{6} \triangleq \big[{2\mathstrut}^{p}(p + 1)c_{Y}\big]^{1/(p\,+\,1)}$.
Continuing (\ref{eqVolh2}), with (\ref{eqVol2}) we obtain the following
upper bound on the second term in (\ref{eqByParts2}), which is by (\ref{eqLocalEntropy2}) also
an upper bound
on both terms of (\ref{eqByParts2}):
\begin{align}
&
\int_{\mathbb{R}}
{p\mathstrut}_{Y}(y)
{\log\mathstrut}_{\!b}\,{p\mathstrut}_{Y}(y)dy
\, - \,
\!\!
\int_{\mathbb{R}}
\widehat{p}{\mathstrut}_{Y}(y)
{\log\mathstrut}_{\!b}\,\widehat{p}{\mathstrut}_{Y}(y)
dy
\overset{h_{1}\, \leq \, 1/e}{\leq}
c_{6}
h_{1}^{p/(p\,+\,1)}
\max
\big\{
{\log\mathstrut}_{\!b}\,(1/h_{1}),
\,
{\log\mathstrut}_{\!b} (e\sqrt{K})
\big\}.
\label{eqSecondTerm2}
\end{align}
So that if $\delta_{1} > 0$ in (\ref{eqYDifference}), then the possible decrease
caused by substitution of $\widehat{p}{\mathstrut}_{Y}$ in place of ${p\mathstrut}_{Y}$ on the LHS of (\ref{eqMargEntropy}) is at most $o(1)$.

\bigskip

{\em The LHS of (\ref{eqLogGaussian})}


With (\ref{eqQuantizer}) let us define a quantization error function 
\begin{equation} \label{eqRelationship}
v(y) \; \triangleq \; y - Q_{\beta}(y),
\;\;\;\;\;\;
|\,v(y)\,| \; \leq \; \tfrac{1}{2}\Delta_{\beta,\,n} \; \triangleq \; \Delta.
\end{equation}
Then, with $\widehat{p}{\mathstrut}_{XY}$ defined in (\ref{eqQuantized}), 
we can obtain a lower bound for the expression on the LHS of (\ref{eqLogGaussian}):
\begin{align}
&
\Delta_{\alpha,\,n}
\Delta_{\beta,\,n}
\sum_{\substack{x\,\in\,{\cal X}_{n}\\
y\,\in\,{\cal Y}_{n}}}
\widehat{p}{\mathstrut}_{XY}(x, y)\,
{|\,y - x\,|\mathstrut}^{q}
\;\; = \;\;
\Delta_{\alpha,\,n}
\sum_{x\,\in\,{\cal X}_{n}}
\int_{\mathbb{R}}
\widehat{p}{\mathstrut}_{XY}(x, y)\,
{|\,y - v(y) - x\,|\mathstrut}^{q}dy
\nonumber \\
\overset{a}{\leq} \;\; &
\mathbb{E}_{{P\mathstrut}_{\!X}{p\mathstrut}_{Y|X}}\!\Big[
{\big(\,|\,Y - X\,| + |\,v(Y)\,|\,\big)\mathstrut}^{q}
\Big]
\;\; \leq \;\;
\mathbb{E}_{{P\mathstrut}_{\!X}{p\mathstrut}_{Y|X}}\!\Big[
{\big(\,|\,Y - X\,| + \Delta \,\big)\mathstrut}^{q}
\Big]
\nonumber \\
\overset{b}{\leq} \;\; &
\mathbb{E}_{{P\mathstrut}_{\!X}{p\mathstrut}_{Y|X}}\!\Big[
\,{|\,Y - X\,|\mathstrut}^{q} \, + \, 
\Delta
\underbrace{q{\big(\,|\,Y - X\,| + \Delta \,\big)\mathstrut}^{q\,-\,1}}_{\text{derivative of $f(t) = t^{q}$}}
\Big]
\nonumber \\
\overset{c}{\leq} \;\; &
\mathbb{E}_{{P\mathstrut}_{\!X}{p\mathstrut}_{Y|X}}\!\big[
\,{|\,Y - X\,|\mathstrut}^{q}\big]
\, + \, 
\Delta q
\Big(\mathbb{E}_{{P\mathstrut}_{\!X}{p\mathstrut}_{Y|X}}\!\big[
{\big(\,|\,Y - X\,| + \Delta \,\big)\mathstrut}^{q}
\big]
\Big)^{\!(q\,-\,1)/q}
\nonumber \\
\overset{d}{\leq} \;\; &
\mathbb{E}_{{P\mathstrut}_{\!X}{p\mathstrut}_{Y|X}}\!\big[
\,{|\,Y - X\,|\mathstrut}^{q}\big]
\, + \, 
\Delta q
\Big({2\mathstrut}^{q}\mathbb{E}_{{P\mathstrut}_{\!X}{p\mathstrut}_{Y|X}}\!\big[
\tfrac{1}{2}\,{|\,Y - X\,|\mathstrut}^{q} + \tfrac{1}{2}\Delta^{q}
\big]
\Big)^{\!(q\,-\,1)/q}
\nonumber \\
\overset{e}{\leq} \;\; &
\mathbb{E}_{{P\mathstrut}_{\!X}{p\mathstrut}_{Y|X}}\!\big[
\,{|\,Y - X\,|\mathstrut}^{q}\big]
\, + \,
\underbrace{
\Delta q\,
{2\mathstrut}^{(q\,-\,1)^{2}/q}
\big(c_{XY} + \Delta^{q}
\big)^{\!(q\,-\,1)/q}}_{o(1)},
\label{eqLHS}
\end{align}
where ($a$) follows because $\widehat{p}{\mathstrut}_{XY}(x, y)\leq {p\mathstrut}_{XY}(x, y)$ and by (\ref{eqJointPDF});
in ($b$) we use the monotonically increasing derivative of $f(t) = t^{q}$; ($c$) follows by Jensen's inequality for a concave ($\cap$)
function $f(t) = t^{(q\,-\,1)/q}$; ($d$) follows by Jensen's inequality for a convex ($\cup$)
function $f(t) = t^{q}$, $q\geq 1$; and finally ($e$) follows by the condition of the lemma.

\bigskip

{\em Joint type ${P\mathstrut}_{\!XY}$}


Let us define two mutually-complementary probability masses for each $(x, y) \in {\cal X}_{n} \times {\cal Y}_{n}$:
\begin{align}
\,\widehat{\!P}{\mathstrut}_{\!XY}(x, y) \;\; & \triangleq \;\;
\widehat{p}{\mathstrut}_{XY}(x, y)\Delta_{\alpha,\,n}\Delta_{\beta,\,n},
\label{eqPDFtoType} \\
\,\widehat{\!P}{\mathstrut}_{\!XY}^{\,c}(x, y) \;\; & \triangleq \;\;
\Bigg\{
\begin{array}{r r}
{P\mathstrut}_{\!X}(x) -
\sum_{\tilde{y}\,\in\,{\cal Y}_{n}}\,\widehat{\!P}{\mathstrut}_{\!XY}(x, \tilde{y}), & \;\;\;
y =  Q_{\beta}(x), \\
0, & \;\;\; \text{o.w.},
\end{array}
\label{eqDiagonal} 
\end{align}
where $Q_{\beta}(\cdot)$ is defined as in (\ref{eqQuantizer}).
It follows from 
(\ref{eqQuantized}) and (\ref{eqCube}),
that each number $\,\widehat{\!P}{\mathstrut}_{\!XY}(x, y)$
is an integer multiple of $1/n$
and $\sum_{y\,\in\,{\cal Y}_{n}}\,\widehat{\!P}{\mathstrut}_{\!XY}(x, y) \leq {P\mathstrut}_{\!X}(x)$
for each $x \in {\cal X}_{n}$.
Then a joint type can be formed with the two definitions above:
\begin{equation} \label{eqJointType}
{P\mathstrut}_{\!XY}(x, y) \;\; \triangleq \;\; \,\widehat{\!P}{\mathstrut}_{\!XY}(x, y) \, + \, \,\widehat{\!P}{\mathstrut}_{\!XY}^{\,c}(x, y),
\;\;\;\;\;\; \forall (x, y) \in {\cal X}_{n}\times {\cal Y}_{n},
\end{equation}
such that ${P\mathstrut}_{\!XY} \in {\cal P}_{n}({\cal X}_{n}\times {\cal Y}_{n})$
and $\sum_{y\,\in\,{\cal Y}_{n}}{P\mathstrut}_{\!XY}(x, y) = {P\mathstrut}_{\!X}(x)$ for each $x \in {\cal X}_{n}$.

\bigskip

{\em The RHS of (\ref{eqLogGaussian})}


Having defined ${P\mathstrut}_{\!XY}$ and $\,\widehat{\!P}{\mathstrut}_{\!XY}$,
let us examine the possible decrease in the expression found on the RHS of (\ref{eqLogGaussian}) when 
${P\mathstrut}_{\!XY}$ inside that expression is replaced with $\,\widehat{\!P}{\mathstrut}_{\!XY}$:
\begin{align}
&
\sum_{\substack{x\,\in\,{\cal X}_{n}\\
y\,\in\,{\cal Y}_{n}}}
{P\mathstrut}_{\!XY}(x, y)\,
{|\,y - x\,|\mathstrut}^{q}
\; - \;
\sum_{\substack{x\,\in\,{\cal X}_{n}\\
y\,\in\,{\cal Y}_{n}}}
\,\widehat{\!P}{\mathstrut}_{\!XY}(x, y)\,
{|\,y - x\,|\mathstrut}^{q}
\;\;
\underset{a}{=} \;\;
\sum_{\substack{x\,\in\,{\cal X}_{n}\\
y\,\in\,{\cal Y}_{n}}}
\,\widehat{\!P}{\mathstrut}_{\!XY}^{\,c}(x, y)\,
{|\,y - x\,|\mathstrut}^{q}
\nonumber \\
\underset{b}{=} \;\; &
\sum_{\substack{x\,\in\,{\cal X}_{n}\\
y\,\in\,{\cal Y}_{n}}}
\,\widehat{\!P}{\mathstrut}_{\!XY}^{\,c}(x, y)\,
{|\,v(x)\,|\mathstrut}^{q}
\;\;
\underset{c}{\leq} \;\; 
p_{1}{2\mathstrut}^{-q}\Delta_{\beta,\,n}^{q}
\;\; \underset{d}{\leq} \;\;
\underbrace{(c_{5}h)^{p/(p\,+\,2)}{2\mathstrut}^{1\,-\,q}\Delta_{\beta,\,n}^{q}}_{o(1)},
\label{eqRHS}
\end{align}
where ($a$) follows by (\ref{eqJointType}),
($b$) follows according to the definitions (\ref{eqDiagonal}) and (\ref{eqRelationship}),
($c$) follows because $|\,v(x)\,| \, \leq \, \frac{1}{2}\Delta_{\beta,\,n}$ and because
\begin{align}
\sum_{\substack{x\,\in\,{\cal X}_{n}\\
y\,\in\,{\cal Y}_{n}}}
\,\widehat{\!P}{\mathstrut}_{\!XY}^{\,c}(x, y)
\;\; & \underset{(\ref{eqJointType})}{=} \;\;
1 \, - \,
\sum_{\substack{x\,\in\,{\cal X}_{n}\\
y\,\in\,{\cal Y}_{n}}}
\,\widehat{\!P}{\mathstrut}_{\!XY}(x, y)
\;\; \underset{(\ref{eqPDFtoType})}{=} \;\;
1 \, - \,
\Delta_{\alpha,\,n}
\Delta_{\beta,\,n}
\sum_{\substack{x\,\in\,{\cal X}_{n}\\
y\,\in\,{\cal Y}_{n}}}
\widehat{p}{\mathstrut}_{XY}(x, y)
\nonumber \\
& \overset{(\ref{eqQuantized})}{=} \;\;
1 \, - \,
\int\!\!\!\!\int_{\mathbb{R}^{2}}\widehat{p}{\mathstrut}_{XY}(x, y)dxdy
\;\; \overset{(\ref{eqProbLoss})}{=} \;\; p_{1} > 0,
\label{eqLostProb}
\end{align}
then ($d$) follows by the upper bound on $p_{1}$ of (\ref{eqProbLoss}).
Since by definition (\ref{eqPDFtoType})
we also have
\begin{displaymath}
\sum_{\substack{x\,\in\,{\cal X}_{n}\\
y\,\in\,{\cal Y}_{n}}}
\,\widehat{\!P}{\mathstrut}_{\!XY}(x, y)\,
{|\,y - x\,|\mathstrut}^{q} \;\;
=
\;\;
\Delta_{\alpha,\,n}
\Delta_{\beta,\,n}
\sum_{\substack{x\,\in\,{\cal X}_{n}\\
y\,\in\,{\cal Y}_{n}}}
\widehat{p}{\mathstrut}_{XY}(x, y)\,
{|\,y - x\,|\mathstrut}^{q},
\end{displaymath}
which is exactly the beginning of (\ref{eqLHS}),
then combining (\ref{eqLHS}) and (\ref{eqRHS}) we obtain (\ref{eqLogGaussian}).

The remainder of the proof for (\ref{eqCondEntropy}) and (\ref{eqMargEntropy}) will easily follow by Lemma~\ref{Lemxlogx}
applied to corresponding discrete entropy expressions with probability masses.

\bigskip

{\em The RHS of (\ref{eqCondEntropy})}


In order to upper-bound the expression on the RHS of (\ref{eqCondEntropy}), it is convenient to write:
\begin{align}
&
\sum_{\substack{x\,\in\,{\cal X}_{n}\\
y\,\in\,{\cal Y}_{n}}}
{P\mathstrut}_{\!XY}(x, y)\,{\log\mathstrut}_{\!b}\,\frac{{P\mathstrut}_{\!XY}(x, y)}
{\Delta_{\alpha,\,n}
\Delta_{\beta,\,n}}
\, - \,
\sum_{\substack{x\,\in\,{\cal X}_{n}\\
y\,\in\,{\cal Y}_{n}}}
\,\widehat{\!P}{\mathstrut}_{\!XY}(x, y)\,{\log\mathstrut}_{\!b}\,\frac{\,\widehat{\!P}{\mathstrut}_{\!XY}(x, y)}
{\Delta_{\alpha,\,n}
\Delta_{\beta,\,n}}
\nonumber \\
\underset{a}{=} \;\; &
(1-\gamma)
\sum_{\substack{x\,\in\,{\cal X}_{n}\\
y\,\in\,{\cal Y}_{n}}}
\,\widehat{\!P}{\mathstrut}_{\!XY}^{\,c}(x, y)\,{\log\mathstrut}_{\!b}\,n
\; + \;
\sum_{\substack{x\,\in\,{\cal X}_{n}\\
y\,\in\,{\cal Y}_{n}}}
\Big[\,
{P\mathstrut}_{\!XY}(x, y)\,{\log\mathstrut}_{\!b}\,{P\mathstrut}_{\!XY}(x, y)
\, - \,
\,\widehat{\!P}{\mathstrut}_{\!XY}(x, y)\,{\log\mathstrut}_{\!b}\,\,\widehat{\!P}{\mathstrut}_{\!XY}(x, y)
\,\Big]
\nonumber \\
\underset{b}{\leq} \;\; &
(1-\gamma)
\sum_{\substack{x\,\in\,{\cal X}_{n}\\
y\,\in\,{\cal Y}_{n}}}
\,\widehat{\!P}{\mathstrut}_{\!XY}^{\,c}(x, y)\,{\log\mathstrut}_{\!b}\,n
\; + \;
\sum_{\substack{x\,\in\,{\cal X}_{n}\\
y\,\in\,{\cal Y}_{n}}}
\max\Big\{\!
-\,\widehat{\!P}{\mathstrut}_{\!XY}^{\,c}(x, y)\,{\log\mathstrut}_{\!b}\,\,\widehat{\!P}{\mathstrut}_{\!XY}^{\,c}(x, y), \;\;
\,\widehat{\!P}{\mathstrut}_{\!XY}^{\,c}(x, y)\,{\log\mathstrut}_{\!b}\,e
\Big\}
\nonumber \\
\underset{c}{\leq} \;\; &
(1-\gamma)
\sum_{\substack{x\,\in\,{\cal X}_{n}\\
y\,\in\,{\cal Y}_{n}}}
\,\widehat{\!P}{\mathstrut}_{\!XY}^{\,c}(x, y)\,{\log\mathstrut}_{\!b}\,n
\; + \;
\sum_{\substack{x\,\in\,{\cal X}_{n}\\
y\,\in\,{\cal Y}_{n}}}
\,\widehat{\!P}{\mathstrut}_{\!XY}^{\,c}(x, y)\,{\log\mathstrut}_{\!b}\,n
\nonumber \\
\underset{d}{=} \;\; &
(2-\gamma)p_{1}\,{\log\mathstrut}_{\!b}\,n
\;\; \underset{e}{\leq} \;\;
\underbrace{4(c_{5}h)^{p/(p\,+\,2)}\,{\log\mathstrut}_{\!b}\,n}_{o(1)}
,
\label{eqDiscreteEntropy}
\end{align}
where ($a$) follows by (\ref{eqJointType}) and (\ref{eqDelta});
in ($b$) 
we apply the upper bound of Lemma~\ref{Lemxlogx} with its parameters $t_{1} = \,\widehat{\!P}{\mathstrut}_{\!XY}(x, y)$
and $t_{1} + t = {P\mathstrut}_{\!XY}(x, y) \leq 1$
with $t = \,\widehat{\!P}{\mathstrut}_{\!XY}^{\,c}(x, y)\leq p_{1}$ by (\ref{eqJointType}) and  
(\ref{eqLostProb}), and 
for $n$ sufficiently large such that $p_{1} \leq 1/e$ by (\ref{eqProbLoss}); 
($c$) follows for $n > e$ since $\,\widehat{\!P}{\mathstrut}_{\!XY}^{\,c}(x, y) \geq 1/n$ when positive; 
the equality ($d$) follows by (\ref{eqLostProb}) and the inequality ($e$) follows by (\ref{eqProbLoss}).
Now since
\begin{displaymath}
\sum_{\substack{x\,\in\,{\cal X}_{n}\\
y\,\in\,{\cal Y}_{n}}}
\,\widehat{\!P}{\mathstrut}_{\!XY}(x, y)\,{\log\mathstrut}_{\!b}\,\frac{\,\widehat{\!P}{\mathstrut}_{\!XY}(x, y)}
{\Delta_{\alpha,\,n}
\Delta_{\beta,\,n}}
\;\; = \;\;
\int\!\!\!\!\int_{\mathbb{R}^{2}}\widehat{p}{\mathstrut}_{XY}(x, y)
{\log\mathstrut}_{\!b}\,\widehat{p}{\mathstrut}_{XY}(x, y)dxdy,
\end{displaymath}
the inequality in (\ref{eqCondEntropy}) follows by comparing (\ref{eqCondEntIdeal}), (\ref{eqPutting}), and (\ref{eqDiscreteEntropy}).

\bigskip

{\em The RHS of (\ref{eqMargEntropy})}


With $\,\widehat{\!P}{\mathstrut}_{\!Y}(y) \triangleq \sum_{x\,\in \,{\cal X}_{n}}\,\widehat{\!P}{\mathstrut}_{\!XY}(x, y)$
and $\,\widehat{\!P}{\mathstrut}_{\!Y}^{\,c}(y) \triangleq \sum_{x\,\in \,{\cal X}_{n}}\,\widehat{\!P}{\mathstrut}_{\!XY}^{\,c}(x, y)$
we have
\begin{align}
&
\sum_{y\,\in\,{\cal Y}_{n}}
{P\mathstrut}_{\!Y}(y)\,{\log\mathstrut}_{\!b}\,\frac{{P\mathstrut}_{\!Y}(y)}
{\Delta_{\beta,\,n}}
\, - \,
\sum_{y\,\in\,{\cal Y}_{n}}
\,\widehat{\!P}{\mathstrut}_{\!Y}(y)\,{\log\mathstrut}_{\!b}\,\frac{\,\widehat{\!P}{\mathstrut}_{\!Y}(y)}
{\Delta_{\beta,\,n}}
\label{eqDiscreteMargEntropy} \\
\underset{a}{=} \;\; &
\beta
\sum_{y\,\in\,{\cal Y}_{n}}
\,\widehat{\!P}{\mathstrut}_{\!Y}^{\,c}(y)\,{\log\mathstrut}_{\!b}\,n
\; + \;
\sum_{y\,\in\,{\cal Y}_{n}}
\Big[\,
{P\mathstrut}_{\!Y}(y)\,{\log\mathstrut}_{\!b}\,{P\mathstrut}_{\!Y}(y)
\, - \,
\,\widehat{\!P}{\mathstrut}_{\!Y}(y)\,{\log\mathstrut}_{\!b}\,\,\widehat{\!P}{\mathstrut}_{\!Y}(y)
\,\Big]
\nonumber \\
\underset{b}{\geq} \;\; &
\beta
\sum_{y\,\in\,{\cal Y}_{n}}
\,\widehat{\!P}{\mathstrut}_{\!Y}^{\,c}(y)\,{\log\mathstrut}_{\!b}\,n
\; + \;
\sum_{y\,\in\,{\cal Y}_{n}}
\,\widehat{\!P}{\mathstrut}_{\!Y}^{\,c}(y)\,{\log\mathstrut}_{\!b}\,\,\widehat{\!P}{\mathstrut}_{\!Y}^{\,c}(y)
\nonumber \\
\underset{c}{\geq} \;\; &
\beta
\sum_{y\,\in\,{\cal Y}_{n}}
\,\widehat{\!P}{\mathstrut}_{\!Y}^{\,c}(y)\,{\log\mathstrut}_{\!b}\,n
\; + \;
\sum_{y\,\in\,{\cal Y}_{n}}
\,\widehat{\!P}{\mathstrut}_{\!Y}^{\,c}(y)\,{\log\mathstrut}_{\!b}\,(1/n)
\;\;
\underset{d}{=} \;\;
-(1 - \beta)p_{1}
\,{\log\mathstrut}_{\!b}\,n
\;\; \underset{e}{\geq} \;\;
\underbrace{-2(c_{5}h)^{p/(p\,+\,2)}\,{\log\mathstrut}_{\!b}\,n}_{o(1)},
\nonumber
\end{align}
where ($a$) follows by (\ref{eqJointType}) and (\ref{eqDelta});
in ($b$) we 
apply the lower bound of Lemma~\ref{Lemxlogx} with its parameters $t_{1} = \,\widehat{\!P}{\mathstrut}_{\!Y}(y)$
and $t_{1} + t = {P\mathstrut}_{\!Y}(y)$
with $t = \,\widehat{\!P}{\mathstrut}_{\!Y}^{\,c}(y)\leq p_{1}$ by (\ref{eqJointType}) and  
(\ref{eqLostProb}), 
and for $n$ sufficiently large such that $p_{1} \leq 1/e$ by (\ref{eqProbLoss}); 
($c$) follows because (\ref{eqLostProb}) is positive and 
$\,\widehat{\!P}{\mathstrut}_{\!Y}^{\,c}(y) \geq 1/n$ whenever positive; 
the equality ($d$) follows by (\ref{eqLostProb}) and the inequality ($e$) follows by (\ref{eqProbLoss}).
From (\ref{eqMargUsual}) and (\ref{eqPDFtoType}) we observe that $\,\widehat{\!P}{\mathstrut}_{\!Y}(y) = \widehat{p}{\mathstrut}_{Y}(y)\Delta_{\beta,\,n}$.
Since the function $\widehat{p}{\mathstrut}_{Y}(y)$
is piecewise constant in $\mathbb{R}$ by the definition of $\widehat{p}{\mathstrut}_{XY}$,
it follows that
\begin{displaymath}
\sum_{y\,\in\,{\cal Y}_{n}}
\,\widehat{\!P}{\mathstrut}_{\!Y}(y)\,{\log\mathstrut}_{\!b}\,\frac{\,\widehat{\!P}{\mathstrut}_{\!Y}(y)}
{\Delta_{\beta,\,n}}
\;\; = \;\;
\int_{\mathbb{R}}
\widehat{p}{\mathstrut}_{Y}(y)
{\log\mathstrut}_{\!b}\,\widehat{p}{\mathstrut}_{Y}(y)dy.
\end{displaymath}
Then the inequality (\ref{eqMargEntropy}) follows by comparing (\ref{eqSecondTerm2}),
(\ref{eqDiscreteMargEntropy}).
This concludes the proof of Lemma~\ref{LemQuant}.
$\square$

\bigskip

\begin{lemma}[\!\!\cite{TridenskiSomekhBaruh23}]\label{Lemxlogx}
{\em Let $f(x) = x\ln x\,$, then, for $0 < t \leq 1/e$ and $t_{1} > 0\,$,}
\begin{displaymath}
t \ln t
\;\; \leq \;\;
f(t_{1} + t) - f(t_{1})
\;\; \leq \;\; t\,\ln \max
\big\{
1/t, \; (t_{1} + t)e
\big\}.
\end{displaymath}
\end{lemma}


\bibliographystyle{IEEEtran}

\begin{thebibliography}{1}




\bibitem
{CoverThomas}
T.~M.~Cover and J.~A.~Thomas,
\newblock {\em ``Elements of Information Theory,''}
\newblock John Wiley \& Sons, 2006.



\bibitem
{Shannon48}
C.~E.~Shannon,
\newblock ``A mathematical theory of communication,''
\newblock {\em The Bell System Technical Journal}, vol. 27, pp. 379--423, 623--656, Jul, Oct 1948.




\bibitem
{Dytso17}
A.~Dytso, R.~Bustin, H.~V.~Poor, and S.~Shamai,
\newblock ``On Additive Channels with Generalized Gaussian Noise,''
\newblock in {\em IEEE International Symposium on Information Theory (ISIT)}, Aachen, Germany, Jun 2017.




\bibitem
{Narayanan22}
P.~Narayanan and L.~N.~Theagarajan,
\newblock ``Capacity of Peak and Average Power Constrained Additive Generalized Gaussian Noise Channels,''
\newblock {\em IEEE Communications Letters}, vol. 26, no. 12, pp. 2880--2883, Dec 2022.




\bibitem
{Csiszar98}
I.~Csisz\'ar,
\newblock ``The Method of Types,''
\newblock {\em IEEE Trans. on Information Theory}, vol. 44, no. 6, pp. 2505--2523, Oct 1998.





\bibitem
{TridenskiSomekhBaruh23}
S.~Tridenski and A.~Somekh-Baruch,
\newblock ``The Method of Types for the AWGN Channel,''
\newblock {\em Entropy}, 27(6), 621, Jun 2025.







%





\bibitem
{Shannon59}
C.~E.~Shannon,
\newblock ``Probability of Error for Optimal Codes in a Gaussian Channel,''
\newblock {\em The Bell System Technical Journal}, vol. 38, no. 3, pp. 611--656, May 1959.











\bibitem
{CsiszarKorner}
I.~Csisz\'ar and J.~K{\"o}rner,
\newblock {\em ``Information Theory: Coding Theorems for Discrete Memoryless Systems,''}
\newblock Academic Press, 1981.


























\end{thebibliography}

\end{document}